\documentclass[onecolumn,ldraftcls, 11pt]{IEEEtran}

\usepackage{pgfplots}
\pgfplotsset{compat=newest}
\usepackage{float}
\usepackage[utf8]{inputenc} 
\usepackage[T1]{fontenc}
\usepackage{url}
\usepackage{ifthen}
\usepackage{cite}
\usepackage{algorithm}
\usepackage{algpseudocode}
\usepackage{tablefootnote}
\usepackage[cmex10]{amsmath}
\usepackage{wrapfig}
\usepackage{tikz}
\usepackage{circuitikz}
\usepackage{multirow}
\usepackage{makecell}
\usetikzlibrary{shapes, arrows, positioning, patterns}
\usepackage[includefoot,left=25.4mm,right=26mm,top = 25.4mm, bottom=26mm]{geometry}

\usepackage{cite}
\usepackage[colorlinks=false]{hyperref}%
\usepackage{mleftright}       
\mleftright                   
\usepackage{subcaption}
\usepackage{graphicx}         
\usepackage{booktabs}         
\usepackage{amsthm}
 \usepackage{graphicx}
\usepackage{amsmath}
\usepackage{graphics,graphicx,psfrag,color,float}
\usepackage{epsfig,bbold,bbm}
\usepackage{bm}
\usepackage{mathtools}
\usepackage{amsmath}
\usepackage{amssymb}
\usepackage{amsfonts}
\usepackage{amsfonts}
\usepackage {times}
\usepackage{bbm}
\usepackage{soul}
\usepackage{cleveref}
\usepackage{enumitem}
\usepackage{xcolor}
\usepackage{txfonts}
\theoremstyle{definition} 
\newtheorem{example}{Example}
\makeatletter
\def\namedlabel#1#2{\begingroup
    #2%
    \def\@currentlabel{#2}%
    \phantomsection\label{#1}\endgroup
}
\makeatother
\DeclareUnicodeCharacter{0325}{}

\def\QED{\mbox{\rule[0pt]{1.5ex}{1.5ex}}}

\newtheorem{theorem}{\bf{Theorem}}
\newtheorem{fact}{Fact}

\newcommand{\supp}{\operatorname{supp}}
\newcommand{\RNum}[1]{\uppercase\expandafter{\romannumeral #1\relax}}

\newcommand{\beq}{\begin{equation}}
\newcommand{\enq}{\end{equation}}
\newcommand{\bel}{\begin{lemma}}
\newcommand{\enl}{\end{lemma}}

\newcommand{\bet}{\begin{theorem}}
\newcommand{\ent}{\end{theorem}}

\newcommand{\tr}{\mathrm{Tr}}

\newcommand{\nn}{\nonumber}

\newcommand{\ketbra}[1]{|#1\rangle\langle#1|}

\newcommand{\customfootnotetext}[2]{{
  \renewcommand{\thefootnote}{#1}
  \footnotetext[0]{#2}}}
\newcommand{\eps}{\varepsilon}

\newcommand*{\bbR}{\mathbb{R}}

\newcommand*{\bbD}{\mathbb{D}}

\newcommand*{\bbI}{\mathbb{I}}

\newcommand*{\bbE}{\mathbb{E}}

\newcommand*{\cP}{\mathcal{P}}

\newcommand*{\cR}{\mathcal{R}}
\newcommand*{\cH}{\mathcal{H}}

\newcommand*{\cF}{\mathcal{F}}

\newcommand*{\cD}{\mathcal{D}}

\newcommand*{\cK}{\mathcal{K}}

\newcommand*{\cN}{\mathcal{N}}
\newcommand*{\cO}{\mathcal{O}}

\newcommand*{\cT}{\mathcal{T}}
\newcommand*{\cX}{\mathcal{X}}

\newcommand*{\cE}{\mathcal{E}}

\newcommand*{\cY}{\mathcal{Y}}

\newcommand{\ket}[1]{|#1 \rangle}

\mathchardef\mhyphen="2D

\newcommand{\norm}[2]{{\left\lVert #1 \right\rVert}_{#2}}
\newcommand{\abs}[1]{\left\vert#1\right\vert}

\makeatletter
\newcommand*{\rom}[1]{\expandafter\@slowromancap\romannumeral #1@}
\makeatother
\mathchardef\mhyphen="2D
\newlist{steps}{enumerate}{1}
\setlist[steps, 1]{leftmargin = 1.1cm, label = Step \arabic*.}
\newtheorem{remark}{Remark}
\newtheorem{definition}{Definition}
\newtheorem{claim}{Claim}

\newtheorem{lemma}{Lemma}
\newtheorem{corollary}{Corollary}
\newtheorem{proposition}{Proposition}
\usepackage[utf8]{inputenc} 
\usepackage[T1]{fontenc}
\usepackage{url}
\usepackage{ifthen}
\usepackage[cmex10]{amsmath}
\usepackage{amsfonts} 

\begin{document}

\title{Quantum Information Ordering and Differential Privacy}

\author{
Naqueeb Ahmad Warsi\textsuperscript{$*$}, Ayanava Dasgupta\textsuperscript{$*$} and  Masahito Hayashi\textsuperscript{$\dagger$}

}

\customfootnotetext{$*$}{
Indian Statistical Institute,
Kolkata 700108, India.
Email: 
{\sf 
 naqueebwarsi@isical.ac.in, ayanavadasgupta\_r@isical.ac.in
}
}

\customfootnotetext{$\dagger$}{School of Data Science, The Chinese University of Hong Kong, Shenzhen, Longgang District, Shenzhen, 518172, China

International Quantum Academy, Futian District, Shenzhen 518048, China

Graduate School of Mathematics, Nagoya University, Nagoya, 464-8602, Japan.
Email: 
{\sf hmasahito@cuhk.edu.cn}}

\maketitle

\maketitle
\begin{abstract}
    We study quantum differential privacy (QDP) by defining a notion of the order of informativeness between two pairs of quantum states. In particular, we show that if the hypothesis testing divergence of the one pair dominates over that of the other pair, then this dominance holds for every $ f$-divergence. This approach completely characterizes $(\eps,\delta)$-QDP mechanisms by identifying the most informative $(\eps,\delta)$-DP quantum state pairs. We apply this to study precise limits for privatized hypothesis testing and privatized quantum parameter estimation, including tight upper-bounds on the quantum Fisher information under QDP. Finally, we establish near-optimal contraction bounds for differentially private quantum channels with respect to the Hockey-Stick divergence.
\end{abstract}


\section{Introduction}\label{sec_intro}

Differential privacy is fundamentally a constraint on distinguishability.
In the classical setting, differential privacy limits how well an adversary can distinguish the outputs of a mechanism applied to neighboring datasets \cite{dwork2006calibrating,dwork2006}.
In the quantum setting considered in this paper, the corresponding object is a pair of output states.
More precisely, two neighboring inputs of a mechanism give rise to two quantum states $\rho$ and $\sigma$, and privacy requires that no binary quantum measurement can distinguish these two states beyond the prescribed $(\eps,\delta)$ tolerance.
Thus, quantum differential privacy can be viewed as a constraint on the binary hypothesis testing region associated with the pair $(\rho,\sigma)$.

The central result of this paper is that the entire class of
$(\eps,\delta)$-QDP output-state pairs admits a single explicit maximal
element in this hypothesis-testing order: a commuting four-dimensional pair
$(\rho_{(\eps,\delta)},\sigma_{(\eps,\delta)})$ whose testing region equals
the full $(\eps,\delta)$-privacy region, and hence contains the testing
region of every $(\eps,\delta)$-DP pair.

In this paper, an output-state pair $(\rho,\sigma)$ is called
$(\eps,\delta)$-differentially private if every binary POVM applied to the
pair produces two classical distributions satisfying the usual two-sided
$(\eps,\delta)$-DP inequalities.
Equivalently, for every effect $0\preceq \Lambda\preceq \mathbb{I}$,
\begin{equation}
\tr[\Lambda\rho] \le e^\eps \tr[\Lambda\sigma]+\delta,\quad
\tr[\Lambda\sigma] \le e^\eps \tr[\Lambda\rho]+\delta .
\label{def_DP_state_eq}
\end{equation}
We denote by $\mathfrak{D}_{(\eps,\delta)}$ the collection of all output-state
pairs satisfying \eqref{def_DP_state_eq}. This notation is used throughout
the paper; further notation and information quantities are introduced in
Section~\ref{sec_def_not}.


Operationally, the two states may be viewed as the outputs of a respondent's
privacy-preserving quantum mechanism on two neighboring datasets, while an
investigator attempts to infer which neighboring dataset was used by applying
a quantum measurement to the observed output.
The hypothesis-testing region therefore quantifies the investigator's optimal
distinguishing power under the privacy constraint.

Several ingredients of this viewpoint are already known.
The relation between differential privacy, binary hypothesis testing, and
Hockey-Stick divergences has appeared in prior work, both in the classical and
quantum settings \cite{KOV15,HRF23,Nuradha25}.
Moreover, ordering binary statistical experiments by their testing regions,
Lorenz curves, or related distinguishability criteria is a well-established
theme in both classical and quantum statistical comparison theory
\cite{Blackwell51,BG54,BG17,Buscemi2012,Jencova2012,Matsumoto2014,GLW24}.
Building on these tools, the present paper gives an explicit extremal
characterization of the $(\eps,\delta)$-private quantum state-pair class and
derives consequences for privatized inference and contraction bounds.

The main gap addressed here is the absence, for the approximate
$(\eps,\delta)$-QDP setting, of a single explicit extremal state pair that
represents the whole privacy class.
Although the hypothesis-testing formulation of QDP gives a geometric
description of the allowed testing region, it does not by itself identify a
state pair that is maximal among all $(\eps,\delta)$-DP pairs.
We show that such a pair can be chosen explicitly.
More precisely, we construct a commuting four-dimensional pair
\(
(\rho_{(\eps,\delta)},\sigma_{(\eps,\delta)})
\)
whose testing region coincides with the full privacy region
$\mathcal{R}(\eps,\delta)$.
Consequently, every pair
$(\rho,\sigma)\in\mathfrak{D}_{(\eps,\delta)}$ is dominated by this explicit
pair in the testing-region order.
We call this pair the weakest, or equivalently the most informative,
$(\eps,\delta)$-DP pair: it is the least protected pair still satisfying the
privacy constraint.
Here ``weakest'' refers to the amount of privacy protection, whereas
``most informative'' refers to statistical distinguishability.

A key technical point is that this dominance cannot be justified by a
degrading quantum channel.
In the classical Blackwell setting, the informativeness order is equivalent
to the existence of a stochastic degradation map, and divergence inequalities
then follow from data processing.
For binary quantum experiments, however, testing-region dominance does not
in general imply the existence of a completely positive trace-preserving
(CP-TP), or even positive trace-preserving, map between the two pairs
\cite{Matsumoto2014,GLW24}.
Hence, the divergence bounds in this paper are proved directly from
Hockey-Stick dominance and the associated integral representations, rather
than from a channel realization.

This point is important not only for the divergence bounds themselves, but
also for the subsequent applications.
Once the explicit extremal pair is identified, it provides a universal
benchmark for the most distinguishable output behavior allowed by the privacy
constraint.
We then transfer this extremality to three operational settings: privatized
hypothesis testing, privatized parameter estimation, and the contraction
analysis of private channels.
In the first setting, the extremal pair determines the optimal one-shot
hypothesis-testing quantities and the corresponding asymptotic testing
exponents.
In the second setting, it yields the exact maximum SLD Fisher information
achievable under the $(\eps,\delta)$-DP constraint.
In the third setting, the same Hockey-Stick dominance principle leads to
bounds on the contraction coefficients of $(\eps,\delta)$-LDP channels and,
in the classical case, to a relative-entropy bound via a truncation argument.
Thus, these results should be viewed as consequences of a single
extremal-pair principle rather than as separate, unrelated contributions.

The main contributions of this paper are as follows.

\begin{itemize}
    \item \emph{Extremal private state pair.}
    We construct an explicit commuting four-dimensional pair
    $(\rho_{(\eps,\delta)},\sigma_{(\eps,\delta)})$ whose testing region is
    exactly the full $(\eps,\delta)$-privacy region.
    Consequently, every $(\eps,\delta)$-DP output-state pair is dominated by
    this pair in the hypothesis-testing order.

    \item \emph{Dominance bounds without degradation.}
    We show that the above extremality yields upper bounds for the
    Hockey-Stick integral $f$-divergences \cite{Hirche_2024} used throughout this paper and for
    the associated Hockey-Stick R\'enyi divergences.
    The proof does not use a degrading quantum channel, which is unavailable
    in general, but works directly through Hockey-Stick dominance and integral
    representations.

    \item \emph{Privatized hypothesis testing and parameter estimation.}
    For the binary-mixture model
    $\rho_\theta=\theta\rho+(1-\theta)\sigma$, we use the extremal-pair
    characterization to obtain worst-case performance bounds under the
    $(\eps,\delta)$-DP constraint.
    In hypothesis testing, the extremal pair maximizes the relevant one-shot
    testing quantities and determines the corresponding asymptotic testing
    exponent.
    In parameter estimation, the same extremal pair yields the exact maximum
    SLD Fisher information achievable under the privacy constraint.
        \item \emph{Contraction of private channels.}
    We derive nearly tight bounds on the contraction coefficient of
    $(\eps,\delta)$-LDP quantum channels with respect to the hockey-stick
    divergence.
    In the classical LDP case, we combine this contraction estimate with a
    truncation argument to obtain an upper bound on the relative entropy
    between output distributions in the $\delta>0$ regime.
\end{itemize}

These results clarify the role of the explicit extremal pair in quantum
differential privacy.
Rather than introducing a new order of informativeness, the paper uses the
existing hypothesis-testing order to identify the extremal private pair and
then transfers this extremality to inference and contraction problems.
Table~\ref{tab} is intended to separate known classical results, known
pure-QDP results, and the new approximate-QDP quantum results established in
this paper.

\begin{table}[h]
\caption{Relationship between results obtained in this work and related studies}\label{tab}
\begin{center}
\begin{tabular}{|c|c|c|c|c|}
\hline
 \multirow{2}{*}{Results} & \multicolumn{2}{|c|}{Classical Domain} & \multicolumn{2}{|c|}{Quantum Domain}\\
\cline{2-5}
&{$\eps$-DP} & {$(\eps,\delta)$-DP} & {$\eps$-DP} & {$(\eps,\delta)$-DP}\\
\hline
\makecell{Explicit weakest / most informative\\ private object}
& \makecell{\cite[Theorem 5]{KOV15}}
& \makecell{\cite[Theorem 18]{KOV15}}
& \makecell{\cite[Theorem 2]{YH20}}
& \makecell{Lemma~\ref{lemma_diff_priv_quant}\\ of this manuscript}\\
\hline
\makecell{Extremal-pair dominance\\ bounds for divergences}
& \makecell{\cite[Lemma D.8]{DRS22}}
&
& \makecell{Corollary~\ref{theo_QLDP_F_div_pure_bound}\\ of this manuscript}
& \makecell{Theorem~\ref{theo_QLDP_approx_F_div_bound}\\ of this manuscript}\\
\hline
\makecell{Privatized hypothesis\\ testing}
& \makecell{\cite[Theorem 18]{KOV15}{\textsuperscript{*}}}
& \cite[Theorem 18]{KOV15}{\textsuperscript{*}}
& \makecell{\cite[Section \RNum{2}B]{YH20}}
& \makecell{Theorem~\ref{THJ}\\ and Corollary~\ref{Cor12}\\ of this manuscript}\\
\hline
\makecell{Privatized parameter\\ estimation}
& \makecell{\cite[Theorem 18]{KOV15}{\textsuperscript{*}}}
& \cite[Theorem 18]{KOV15}{\textsuperscript{*}}
& \makecell{\cite[Section \RNum{2}A]{YH20}}
& \makecell{Theorem~\ref{theo_fisher_info_bound}\\ of this manuscript}\\
\hline
\makecell{Contraction coefficient\\ of trace distance}
& \cite[Corollary 11]{KOV15}
&
& \makecell{\cite[Theorem 2]{Nuradha25}}
& \makecell{\cite[Theorem 5]{Nuradha25}}\\
\hline
\makecell{Contraction coefficient\\ of Hockey-Stick divergence}
& \makecell{\cite[Theorem 1]{ZA23}}
& \makecell{Corollary~\ref{corr_contrac_classical_hockey_stick}\\ of this manuscript}
& \makecell{\cite[Theorem 1]{Nuradha25}}
& \makecell{Lemma~\ref{lemma_contrac_hockey_stick}\\ of this manuscript}\\
\hline
\makecell{Contraction-based upper bounds\\ on relative entropy of LDP channels}
&
& \makecell{Theorem~\ref{theorem_classical_kl_bound}\\ of this manuscript}
& \cite[Proposition 3]{Nuradha25}
& \\
\hline
\end{tabular}
\end{center}
\vspace{1ex}
\footnotesize{{\textsuperscript{*} \textbf{Note:} While \cite[Theorem 18]{KOV15} establishes the structural paradigm that optimal statistical utility is attained by extremal privacy mechanisms, it leaves the exact bounds implicit. In contrast, our results provide concrete, closed-form algebraic limits for these tasks.}}
\end{table}

The rest of this paper is organized as follows.
Section~\ref{sec_def_not} introduces notation and the information quantities
used throughout the paper.
Section~\ref{sec_char_region} describes the hypothesis-testing region
associated with $(\eps,\delta)$-QDP and constructs the extremal pair
$(\rho_{(\eps,\delta)},\sigma_{(\eps,\delta)})$.
Section~\ref{sec_diff_priv_quant} uses the testing-region order to prove the
corresponding divergence domination results.
Section~\ref{sec:fisher} applies the extremal-pair framework to privatized
hypothesis testing and to SLD Fisher information.
Section~\ref{sec_LDP_channel} studies contraction coefficients of private
channels and derives the classical relative-entropy bound for
$(\eps,\delta)$-LDP channels.
Table \ref{tab} above summarizes the relation between the results obtained in our manuscript and the existing results.

\section{Notations, Definitions and Facts}\label{sec_def_not}

We use $\cH$ to denote a finite-dimensional Hilbert space and we denote its dimension with $\abs{\cH}$, $\mathcal{D}(\cH)$ to represent the set of all state density matrices acting on $\cH$. For any quantum state $\rho \in \cD(\cH)$, we define $\mbox{supp}(\rho):=\mbox{Span}\{\ket{i}: \lambda_{i}> 0\},$ where $\left\{\lambda_i\right\}$ represent non-zero eigenvalues of $\rho$. For any finite and non-empty set $\cX$, we denote $\cP(\cX)$ as the set of all probability distributions over $\cX$. Similarly, for any distribution $P \in \cP(\cX)$, we denote $\mbox{supp}(P):=\{x\in \mathcal{X}: P(x)>0\}.$ 

{We now record the output-state formulation of quantum differential privacy
used throughout the paper.

\begin{definition}[$(\eps,\delta)$-DP output-state pair]\label{def_DP_state}
Let $\rho,\sigma \in \cD(\cH)$ be quantum states on a finite-dimensional
Hilbert space $\cH$.  For $\eps\geq 0$ and $\delta\in[0,1]$, the pair
$(\rho,\sigma)$ is called $(\eps,\delta)$-differentially private if every
binary POVM applied to the pair produces two classical distributions
satisfying the usual two-sided $(\eps,\delta)$-DP inequalities.
Equivalently, for every effect $0\preceq \Lambda\preceq \bbI$,
\begin{equation}
\begin{split}
\tr[\Lambda\rho] &\le e^\eps \tr[\Lambda\sigma]+\delta,\\
\tr[\Lambda\sigma] &\le e^\eps \tr[\Lambda\rho]+\delta .
\end{split}
\label{def_DP_state_eq}
\end{equation}
We denote by $\mathfrak{D}_{(\eps,\delta)}$ the collection of all
output-state pairs satisfying \eqref{def_DP_state_eq}.  For $\delta=0$, such
a pair is called pure $\eps$-DP, or simply $\eps$-DP.
\end{definition}}

\begin{definition}[Classical minimun type-\RNum{2} error]\label{def_min_type2_err}
    Given a pair of probability distributions $P,Q$ over a finite set $\cX$, the minimum type-\RNum{2} error of a fixed order $\alpha \in [0,1]$ is defined as follows,
    \begin{equation}
        \beta(\alpha|P\|Q) := \min_{\substack{\phi: \cX \to [0,1]\\ \sum_{x \in \cX} P(x)\phi(x) \leq \alpha}} 1 - \sum_{x \in \cX} Q(x)\phi(x).\label{def_min_type2_err_eq}
    \end{equation}
\end{definition}

\begin{definition}[Classical hypothesis testing divergence]\label{def_hyp_div}
    For any pair of probability distributions $P,Q$ over a finite set $\cX$, and for any $\alpha \in [0,1]$, the classical hypothesis testing divergence of order $\alpha$ is defined as
    \begin{equation}
        D^{\alpha}_{H}(P\|Q) := \max_{\substack{\phi: \cX \to [0,1]\\ \sum_{x \in \cX} P(x)\phi(x) \leq \alpha}} -\ln \left( 1 - \sum_{x \in \cX} Q(x)\phi(x) \right).
    \end{equation}

    Further, note that for any $\alpha \in [0,1]$, $D^{\alpha}_{H}(P\|Q) = -\ln {\beta(\alpha|P\|Q)}$ where $\beta(\alpha|P\|Q)$ is as defined in Definition \ref{def_min_type2_err}.
\end{definition}

\begin{definition}[Classical Hockey-Stick Divergence{\cite{Hull2003}}]\label{def_class_hockey_stick}
Let $P$ and $Q$ be probability distributions on a finite (or measurable) space $\mathcal{X}$, and let $\gamma \geq 1$.  
Then, the Hockey-Stick divergence between $P$ and $Q$ is defined as
\[
  E_{\gamma}(P \| Q)
  := \sup_{A \subseteq \mathcal{X}} \Big( P(A) - \gamma \, Q(A) \Big)
  = \sum_{x \in \mathcal{X}} \big[ P(x) - \gamma Q(x) \big]_+ ,
\]
where $[t]_+ := \max\{t,0\}$ denotes the positive part of $t$.
\end{definition}

\begin{definition}[Classical {R\'enyi} Divergence]\label{def_c-alpha_div}
    For any pair of probability distributions $P,Q$ over a finite set $\cX$, and for any $\alpha \in [0,1]$, the classical { R\'enyi} divergence of order $\alpha$ is defined as
    \begin{equation}
        D_{\alpha}(P\|Q) := 
        \frac{1}{\alpha-1} \log \sum_{x \in \cX} P(x)^\alpha Q(x)^{1-\alpha}.
    \end{equation}
\end{definition}

\begin{definition}\label{def_data_proc_div_class}
    A function $\bbD : \cP(\cX) \times \cP(\cX) \to \bbR^{+}$ (where $\cX$ is an arbitrary set of finite cardinality) is called a divergence if for  
    any classical channel $\cK: \cX \to \cY$, (where $\cY$ is another arbitrary set of finite cardinality), the following holds,
    {\begin{equation}
       \bbD(P\|Q) \geq  \bbD(\cK(P)\|
       \cK(Q)), \forall P,Q \in \cP(\cX),\label{def_data_proc_div_class_eq}
    \end{equation}
    where, $\cK(P)(y) := \bbE_{X \sim P} \left[\cK(y\mid X)\right]$.}
\end{definition}

\begin{definition}[Quantum minimun type-\RNum{2} error]\label{Q-def_min_type2_err}
    For any $\rho,\sigma \in \cD(\cH_A)$, 
    the minimum type-\RNum{2} error of a fixed order $\alpha \in [0,1]$ is defined as follows,
    \begin{equation}
        \beta(\alpha|\rho\|\sigma) := 
        \min_{\substack{0 \preceq \Lambda \preceq \bbI:\\
        \tr[\Lambda \rho] \leq \alpha}} 
        \tr[(\bbI - \Lambda)\sigma] .\label{Qdef_min_type2_err_eq}
    \end{equation}
\end{definition}

\begin{definition}[Quantum hypothesis testing divergence {\cite{Wang_2012}}]\label{d_hyp_func}
    For any $\rho,\sigma \in \cD(\cH_A)$, for any $\alpha \in [0,1]$, the quantum hypothesis testing divergence $D^{\alpha}_{H}{(\rho\|\sigma)}$ of order $\alpha$ is defined as follows,
    \begin{equation}
        D^{\alpha}_{H}{(\rho\|\sigma)} := \max_{\substack{0 \preceq \Lambda \preceq \bbI:\\
        \tr[\Lambda \rho] \leq \alpha}} -\ln \tr[(\bbI - \Lambda)\sigma] .\label{d_hyp_func_eq}
    \end{equation}
{Further, from Definition \ref{Q-def_min_type2_err}, we have the following eqivalent form for $D^{\alpha}_{H}{(\rho\|\sigma)}$,}
\begin{equation}
    { D^{\alpha}_{H}{(\rho\|\sigma)} = - \ln \beta(\alpha|\rho\|\sigma)}\label{fact_d_hyp_func_eq}.
\end{equation}
\end{definition}
{\begin{definition}[Quantum Hockey-Stick Divergence \cite{Hirche_2024}]\label{f_sym_hyp_func}
    For any $\rho,\sigma \in \cD(\cH_A)$, we define the quantum Hockey-Stick divergence $E_{\gamma}(\rho\|\sigma)$ of order $\gamma\geq 0$ as follows,
    \begin{equation}
        E_{\gamma}(\rho\|\sigma) := \tr|(\rho - \gamma\sigma)|_{+} - \left(\tr\left[\rho - \gamma \sigma\right]\right)_{+},\label{f_sym_hyp_func_eq}
    \end{equation}
    where for a hermitian operator $O$, $|O|_{+}$ denotes the positive part of $O$, i.e., $|O|_{+} := \sum_{i} \max\{0,\lambda_i\}\ketbra{\psi_i}$, where $\lambda_i$ and $\ket{\psi_i}$ are the eigenvalues and eigenvectors of $O$, respectively and for a $a \in \bbR$, $(a)_{+}:= max\{a,0\}$.
\end{definition}}

\begin{definition}[Quantum Hockey-Stick $f$-divergence \cite{Hirche_2024}]\label{def_f_div_integral_rep}
    Consider $\cF$ be the set of functions $f : (0,\infty) : \to \bbR$ that are convex and twice differentiable with $f(1) = 0$. Then, for any $f \in \cF$ and for any $\rho,\sigma \in \cD(\cH_A)$, we define the quantum $f$-divergence $D_f(\rho\|\sigma)$ as follows,
    \begin{equation}
        D_f(\rho\|\sigma) :=  \int_{1}^{\infty}\left( f''(\gamma) E_{\gamma}(\rho\|\sigma) + \gamma^{-3} f''(\gamma^{-1}) E_{\gamma}(\sigma\|\rho)\right) d\gamma.\label{def_f_div_integral_rep_eq}
    \end{equation}
    whenever, the integral is finite and $D_{f}(\rho\|\sigma) := +\infty$ otherwise.
\end{definition}

\begin{definition}[Umegaki's Quantum Relative Entropy {\cite{Umegaki1962}}]\label{def_quant_KL_div}
    Consider $\rho,\sigma \in \cD(\cH_A)$, then, we define the quantum relative entropy between $\rho$ and $\sigma$ as follows
    \begin{equation*}
        D(\rho \| \sigma) := \begin{cases}
            \tr[\rho(\log\rho -\log\sigma)], &\text{  if } \rho \ll \sigma, \\
             +\infty, &\text{ else}.
        \end{cases}
    \end{equation*}
\end{definition}

\begin{definition}[Hockey-Stick R\'enyi Divergence {\cite{Hirche_2024}}]\label{def_quant_renyi_hockey}
    Consider $\rho,\sigma \in \cD(\cH_A)$ and $\alpha \in [0,1) \cup (1,+\infty)$. Then, Hockey-Stick R\'enyi Divergence of order $\alpha$ between $\rho$ and $\sigma$ is defined as follows,
        \begin{align}
    D_\alpha(\rho \| \sigma) &:= \frac{1}{\alpha - 1} \log \big( 1 + (\alpha - 1) H_\alpha(\rho \| \sigma) \big),
\label{def_P_div_integral_rep_eq} \\
H_\alpha(\rho \| \sigma) &:= \alpha \int_{1}^{\infty} \gamma^{\alpha - 2} E_\gamma(\rho \| \sigma) + \gamma^{-\alpha - 1} E_\gamma(\sigma \| \rho) \, \mathrm{d}\gamma,
    \end{align}
when a parameter $\alpha \in [0,1) \cup (1,+\infty)$
satisfies the condition
$\alpha < 1 \cap \rho \not\perp \sigma$ or the condition $\rho \ll \sigma$.
\end{definition}

\begin{definition}[Measured R\'enyi Divergence {\cite{H2017QIT,Hirche_2024}}]\label{def_sand}
    Consider $\rho,\sigma \in \cD(\cH_A)$ and $\alpha \in [0,1) \cup (1,+\infty)$. Then, measured R\'enyi Divergence of order $\alpha$ between $\rho$ and $\sigma$ is defined as follows,
    \begin{align*}
        \check{D}_{\alpha} (\rho \| \sigma) 
:= \sup_{M: {\rm POVM}}D_\alpha(P^M_\rho\|P^M_\sigma),
    \end{align*}
    where $P^M_\rho$ is the distribution given as
    $P^M_\rho(i):= \tr\left[M_i \rho\right]$ for a POVM $M=\{M_i\}_i$.    
\end{definition}

\begin{definition}\label{def_data_proc_func}
    A function $\varmathbb{D} : \cD(\cH_A) \times \cD(\cH_A) \to \bbR^{+}$ (where $\cH_A$ is any arbitrary Hilbert space of finite dimension) is called a divergence if for  
    any completely positive trace preserving (CP-TP) map $\cE: \cH_A \to \cH_B$, (where $\cH_B$ is another arbitrary Hilbert space of finite dimension) the following holds,
    \begin{equation}
       \varmathbb{D}(\rho_1\|\rho_2) \geq  \varmathbb{D}(\cE(\rho_1)\|\cE(\rho_2)), \forall \rho_1,\rho_2 \in \cD(\cH_A).\label{def_data_proc_quantum_div_eq}
    \end{equation}
\end{definition}

\begin{fact}
\label{fact_hockey_approx}
 For any two  pair  $(P,Q)$ and $(P',Q')$ probability distributions, the following holds for any $\gamma \geq 1$,
 \begin{align}
     E_{\gamma}(P'\|Q') &\leq E_{\gamma}(P\|Q) + \frac{1}{2}\norm{P'-P}{1} + \frac{\gamma}{2}\norm{Q'-Q}{1}.\label{fact_hockey_approx_eq1}
 \end{align}
\end{fact}
\begin{proof}
    From Definition \ref{def_class_hockey_stick}, we have
    \begin{align*}
        E_{\gamma}(P'\|Q') &\overset{a}{=}  P'(A) - \gamma Q'(A)\\
        &{ =} P(A) - \gamma Q(A) +  P'(A) - P(A) +\gamma(Q(A) - Q'(A))\\
        &\leq \max_{B \subseteq\cX}(P(B) - \gamma Q(B)) + \max_{C \subseteq\cX} (P'(C) - P(C)) +\gamma\max_{D \subseteq\cX}( Q(D) - Q'(D))\nn\\
        &=E_{\gamma}(P\|Q) +  \frac{1}{2}\norm{P'-P}{1} +\frac{\gamma}{2}\norm{Q-Q'}{1}\nn\\
        &= E_{\gamma}(P\|Q) +  \frac{1}{2}\norm{P'-P}{1} +\frac{\gamma}{2}\norm{Q'-Q}{1}.
    \end{align*}
    where in $a$, let $A$ be the maximal set for  $E_{\gamma}(P'\|Q')$ according to Definition \ref{def_class_hockey_stick}. This proves Fact \ref{fact_hockey_approx}.
\end{proof}

\begin{fact}[H\"older's inequality]\label{holder_classic}
    For any two vectors $x=(x_1, \dots, x_n)$ and $y=(y_1, \dots, y_n)$ in $\mathbb{R}^n$ and for any real numbers $p,q \in [1,\infty)$ such that $\frac{1}{p}+\frac{1}{q} = 1$, we have,
    \begin{equation}
        \sum_{i=1}^n \abs{x_i y_i} \leq \left(\sum_{i=1}^n \abs{x_i}^p\right)^{1/p}\left(\sum_{i=1}^n \abs{y_i}^q\right)^{1/q}.\label{Holder_fact_classical}
    \end{equation}
    Further, for the case when $p =1$ and $q = \infty$, the right-hand side of \eqref{Holder_fact_classical} can be interpreted as follows,
    \begin{equation}
        \left(\sum_{i=1}^n \abs{x_i}^p\right)^{1/p}\left(\sum_{i=1}^n \abs{y_i}^q\right)^{1/q} = \left(\sum_{i=1}^n \abs{x_i}\right) \max_{i \in [n]} \abs{y_i}.\label{Holder_fact_classical_p1_qinf}
    \end{equation}
\end{fact}

{\begin{fact}[Classical KL representation via Hockey-Stick divergence
\cite{Sason_2016}]
\label{fact_classical_KL_div_integral_rep}
Let $P$ and $Q$ be probability distributions on a finite alphabet
$\mathcal{X}$. Then the classical relative entropy admits the following
integral representation in terms of the classical Hockey-Stick divergence:
\begin{equation}
    D(P\|Q)
    =
    \int_{1}^{\infty}
    \left(
    \frac{1}{\gamma} E_{\gamma}(P\|Q)
    +
    \frac{1}{\gamma^2} E_{\gamma}(Q\|P)
    \right)d\gamma .
\end{equation}
Here $E_{\gamma}$ denotes the classical Hockey-Stick divergence as defined in
Definition~\ref{def_class_hockey_stick}.
\end{fact}}


\begin{fact}[Integral representation of Quantum Hockey-Stick $f$-divergence \cite{Hirche_2024}]\label{fact_f_div_integral_rep}
    For any $f \in \cF$ (where $\cF$ is defined in Definition \ref{def_f_div_integral_rep}) and for any $\rho,\sigma \in \cD(\cH_A)$, we have,
    \begin{equation}
        D_f(\rho\|\sigma) :=  \int_{0}^{\infty} f''(\gamma) E_{\gamma}(\rho\|\sigma) d\gamma,\label{fact_f_div_integral_rep_eq}
    \end{equation}
    where $D_f(\rho\|\sigma)$ is defined in Definition \ref{def_f_div_integral_rep}.
\end{fact}


\begin{fact}[[\citenum{Hirche_2024}] ]\label{fact_P_div_integral_rep}
    For any $\rho,\sigma \in \cD(\cH_A)$, 
    we have the following relation between
        Hockey-Stick R\'enyi Divergence $D_\alpha(\rho\|\sigma)$ and measured R\'enyi Divergence $\check{D}_\alpha(\rho\|\sigma)$,
    \begin{align}
    \check{D}_\alpha(\rho \| \sigma) \le D_\alpha(\rho \| \sigma) ,
    \label{fact_P_div_integral_rep_eq} 
    \end{align}
for a parameter $\alpha \in [0,1) \cup (1,+\infty)$.
\end{fact}

\begin{fact}
\label{fact:commuting_renyi_full}
Let $\rho$ and $\sigma$ be two commuting quantum states on a Hilbert space $\mathcal{H}$ (i.e., $[\rho, \sigma] = 0$). For any order $\alpha \in (0, \infty) \setminus \{1\}$, the Hockey-Stick Rényi divergence coincides with the Measured Rényi divergence:
\begin{equation}
    D_{\alpha}(\rho \| \sigma) = \check{D}_{\alpha}(\rho \| \sigma) = \frac{1}{\alpha-1} \log \left( \sum_i p_i^\alpha q_i^{1-\alpha} \right),
\end{equation}
where $\{p_i\}$ and $\{q_i\}$ are the eigenvalues of $\rho$ and $\sigma$ in their common eigen-basis.
\end{fact}

{\begin{fact}[{[\citenum{H2017QIT}, Exercise 3.19]}]\label{fact_measured_renyi_half}
\leavevmode For any $\rho,\sigma \in \cD(\cH_A)$, the measured R\'enyi divergence of order $1/2$ is given by:
    \begin{equation}
        \check{D}_{1/2}(\rho \| \sigma) = -2\log \tr\left[\left|\sqrt{\rho}\sqrt{\sigma}\right|\right].\nn
    \end{equation}
\end{fact}}


\begin{fact}\label{T_hyp_func_equi}
    For any $\rho,\sigma \in \cD(\cH_A)$, for any $\alpha \in [0,1]$, we have the following equality,
    \begin{equation}\beta(\alpha|\rho\|\sigma) =
        \min_{\substack{0 \preceq \Lambda \preceq \bbI:\\
        \tr[\Lambda \rho] \leq \alpha}} \tr[(\bbI - \Lambda)\sigma] = \min_{\substack{0 \preceq \Lambda \preceq \bbI:\\
        \tr[\Lambda \rho] = \alpha}} \tr[(\bbI - \Lambda)\sigma].
        \label{T_hyp_func_equi_eq}
    \end{equation}
\end{fact}

\begin{proof}
    Given any $\alpha \in [0,1]$, we consider a POVM $\{\Tilde{\Lambda},\bbI - \Tilde{\Lambda}\}$, (where $0 \preceq \Tilde{\Lambda} \preceq \bbI$) such that $\tr[\Tilde{\Lambda} \rho] = \alpha' < \alpha$ for some $\alpha' \in (0,1)$. We denote $\beta' := \tr[(\bbI - \Tilde{\Lambda})\sigma]$ and $\delta := \frac{1 -\alpha}{1-\alpha'} \in (0,1)$. We now consider an operator $\hat{\Lambda} := \bbI -\delta\left(\bbI -\Tilde{\Lambda}\right)$. The facts $\delta \in (0,1)$ and $\left(\bbI - \Tilde{\Lambda}\right) \preceq \bbI$ yield that $0\preceq \hat{\Lambda} \preceq \bbI$. Thus, $\{\hat{\Lambda},\bbI - \hat{\Lambda}\}$ forms a valid POVM, for which $\tr[\hat{\Lambda} \rho] = 1 - \delta(1-\alpha') = \alpha$ and
    \begin{align*}
        \tr[(\bbI - \hat{\Lambda})\sigma] = \delta \tr[(\bbI - \Tilde{\Lambda})\sigma]
      \overset{(a)}{<} \beta',
    \end{align*}
    where $(a)$ follows from the fact that $\delta < 1$ and $\beta' = \tr[(\bbI - \Tilde{\Lambda})\sigma]$. Thus, for any POVM $\{\Tilde{\Lambda},\bbI - \Tilde{\Lambda}\}$,  such that $\tr[\Tilde{\Lambda} \rho] < \alpha$, there exists a POVM $\{\hat{\Lambda},\bbI - \hat{\Lambda}\}$ such that $\tr[\Tilde{\Lambda} \rho] = \alpha$ and $\tr[(\bbI - \hat{\Lambda})\sigma] < \tr[(\bbI - \Tilde{\Lambda})\sigma]$. Therefore, the optimum $\Lambda$ in the LHS of \eqref{T_hyp_func_equi_eq} satisfies $\tr[{\Lambda} \rho] = \alpha$. This proves Fact \ref{T_hyp_func_equi}.
\end{proof}

{\begin{fact}[{\cite[Lemma 2.3]{Hirche_2024}}]\label{Fact_hockey_equiv_form}
    For any $\rho,\sigma \in \cD(\cH_A)$, the quantum Hockey-Stick divergence $E_{\gamma}(\rho\|\sigma)$ of order $\gamma \geq 0$ (see Definition \ref{f_sym_hyp_func}) has the following equivalent form,
    \begin{equation}
        E_{\gamma}(\rho\|\sigma) = \max_{\substack{0 \preceq \Lambda \preceq \bbI}
        }\tr[\Lambda (\rho - \gamma\sigma)] -\left(\tr\left[\rho - \gamma \sigma\right]\right)_{+} .\label{Fact_hockey_equiv_form_eq}
    \end{equation}
\end{fact}

}

\begin{fact}[{\cite[Lemma $3.1$]{BST2019}}]\label{fact_state_trans}
Let $\rho_1, \sigma_1 \in \mathcal{S}(\mathbb{C}^2)$ be commuting qubit quantum states and 
let $\rho_2, \sigma_2 \in \mathcal{S}(\mathbb{C}^d)$.  
Then the following two conditions are equivalent:
\begin{enumerate}
    \item There exists a CP-TP map $\mathcal{E} : \mathcal{S}(\mathbb{C}^2) \to \mathcal{S}(\mathbb{C}^d)$ such that 
    $\mathcal{E}(\rho_1) = \rho_2$ and $\mathcal{E}(\sigma_1) = \sigma_2$;
    \item $D_{\max}(\rho_1 \Vert \sigma_1) \ge D_{\max}(\rho_2 \Vert \sigma_2)$ 
    and $D_{\max}(\sigma_1 \Vert \rho_1) \ge D_{\max}(\sigma_2 \Vert \rho_2)$.
\end{enumerate}
\end{fact}

\begin{fact}\label{T_hyp_func_equi2}
    For any $\rho,\sigma \in \cD(\cH_A)$, for any $\alpha \in [0,1]$ the following equality holds,
    \begin{equation}
        \max_{\substack{0 \preceq \Lambda \preceq \bbI:\\
        \tr[\Lambda \rho] \geq \alpha}} \tr[(\bbI - \Lambda)\sigma] = \max_{\substack{0 \preceq \Lambda \preceq \bbI:\\
        \tr[\Lambda \rho] = \alpha}} \tr[(\bbI - \Lambda)\sigma].\label{T_hyp_func_equi2_eq}
    \end{equation}
\end{fact}

\begin{proof}
    Proof of Fact \ref{T_hyp_func_equi2} follows directly from the proof of Fact \ref{T_hyp_func_equi}.
\end{proof}

\section{Characteristic Region of $(\eps,\delta)$-Differentially Private Quantum Mechanisms}\label{sec_char_region}
{Differential privacy is naturally described in terms of two players.  The
first player is the respondent, who releases an output through a
privacy-preserving quantum mechanism while wishing to hide which of two
private alternatives is responsible for the output.  The second player is the
investigator, who observes the released output and tries to infer this hidden
alternative.

In this section, the respondent's side is parametrized by a prescribed
privacy budget $(\eps,\delta)$.  The parameter $\eps$ controls the allowed
multiplicative change in the probabilities of measurement outcomes, while
$\delta$ allows an additive slack.  Thus, a smaller privacy budget means that
the respondent requires the two possible outputs to be less distinguishable
to every investigator.  In the quantum output-state formulation, the two
private alternatives give rise to two possible output states, denoted by
$\rho$ and $\sigma$, and the respondent's privacy requirement is precisely
that the pair $(\rho,\sigma)$ belongs to
$\mathfrak{D}_{(\eps,\delta)}$ in the sense of
Definition~\ref{def_DP_state}.

The investigator's side is described by a binary quantum hypothesis test.
If the released state is $\rho$, we regard the null hypothesis as true; if
the released state is $\sigma$, we regard the alternative hypothesis as true:
\begin{align}
H &:= \begin{cases}
H_0: \mbox{the released state is } \rho,\\
H_1: \mbox{the released state is } \sigma.
\end{cases}
\end{align}
The investigator chooses a binary measurement, represented by an effect
$0\preceq \Lambda\preceq \bbI$, and uses it to decide between these two
hypotheses.  The privacy budget $(\eps,\delta)$ chosen on the respondent's
side restricts, uniformly over all such measurements, how small the
investigator's error probabilities can be.

The characteristic region introduced below records the achievable pairs of
type-I and type-II errors for this testing problem.  It therefore gives a
geometric representation of the distinguishability that remains compatible
with the respondent's privacy requirement.

For a rejection effect $0\preceq \Lambda\preceq \bbI$ for $H_0$, the
investigator rejects $H_0$ with probability $\tr[\Lambda\rho]$ when the true
state is $\rho$, and accepts $H_0$ with probability
$1-\tr[\Lambda\sigma]$ when the true state is $\sigma$.  Hence we define
$    \alpha_\Lambda := \tr[\Lambda\rho]$ and
$    \beta_\Lambda := 1-\tr[\Lambda\sigma]$.
Now, for any pair  $(\rho,\sigma)$ of quantum states, the characteristic region is defined as follows.}

\begin{definition}\label{def_char_region}
    For any pair of quantum states $(\rho,\sigma)$, the characteristic region $\cR(\rho,\sigma)$ is defined as follows,
    \begin{equation}
        \cR(\rho,\sigma) := \left\{(\alpha_{\Lambda},\beta_{\Lambda}) : 0 \preceq \Lambda \preceq \bbI\right\}.\label{def_char_region_eq}
    \end{equation}
\end{definition}

Further, if $(\rho,\sigma)$ is $(\eps,\delta)$-QDP, then for any rejection POVM $\Lambda$, from \eqref{def_DP_state_eq} the following constraints hold,
\begin{align}
    \beta_{\Lambda} &\geq e^{-\eps}(1-\delta -\alpha_{\Lambda}),\label{char_region_eq1}\\
    \beta_{\Lambda} &\geq 1 -\delta - e^{\eps}\alpha_{\Lambda},\label{char_region_eq2}\\
    \beta_{\Lambda} &\leq 1 - e^{-\eps}(\alpha_{\Lambda} - \delta),\label{char_region_eq3}\\
    \beta_{\Lambda} &\leq e^{\eps}(1- \alpha_{\Lambda}) + \delta.\label{char_region_eq4}
\end{align} 

{Thus, the respondent's privacy budget $(\eps,\delta)$ determines a universal
outer region for the investigator's error pairs, independently of the
particular private mechanism.  This motivates the following definition of
the privacy region $\cR(\eps,\delta)$.}

\begin{figure}[h]
    \centering
    \begin{tikzpicture}

\coordinate (A) at (1.75,1.75);
\coordinate (B) at (0,5);
\coordinate (C) at (0,5.5);
\coordinate (D) at (0.5,5.5);
\coordinate (E) at (3.75,3.75);
\coordinate (F) at (5.5,0.5);
\coordinate (G) at (5.5,0);
\coordinate (H) at (5,0);

\draw[fill=red] (A) circle (0.04); 
\draw[fill=red] (B) circle (0.04); 
\draw[fill=red] (C) circle (0.04); 
\draw[fill=red] (D) circle (0.04); 
\draw[fill=red] (E) circle (0.04); 
\draw[fill=red] (F) circle (0.04); 
\draw[fill=red] (G) circle (0.04); 
\draw[fill=red] (H) circle (0.04); 
\draw[fill=red] (0,0) circle (0.04); 
\draw[fill=red] (5.5,5.5) circle (0.04);

\draw[->] (0,0) -- (6,0) node[below,xshift = -4cm,yshift = 0.1cm] {\makecell{$\alpha \longrightarrow$ \\ (Type-I error)}};
\draw[->] (0,0) -- (0,6) node[left, rotate = 90, xshift = -1.75cm, yshift = 0.5cm] {\makecell{$\beta \longrightarrow$ \\ (Type-II error)}};
\draw[dotted,thick] (0,0) -- (5.5,5.5) node[anchor = north east, rotate = 45,xshift = -0.75cm] {$\alpha = \beta$};
\draw[dotted,thick] (0,5.5) -- (5.5,0 ) node[anchor = north east, rotate = -45,xshift = -2cm] {$\alpha + \beta = 1$};
\draw[-] (0,5.5) -- (5.5,5.5) node[above, xshift = -2.5cm] {$\beta = 1$};
\draw[-] (5.5,0) -- (5.5,5.5) node[right, xshift = 0.25cm, yshift = -3cm, rotate = 90] {$\alpha = 1$};
\node[anchor = north east] at (0,0) {$(0,0)$};
\node[anchor =  east] at (0,5.5) {$(0,1)$};
\node[anchor = north west] at (5.5,0) {$(1,0)$};
\node[anchor = south west] at (5.5,5.5) {$(1,1)$};

\filldraw[fill=gray, pattern=north east lines, draw=black, very thick, opacity=0.5] (1.75,1.75) -- (0,5) -- (0,5.5)-- (0.5,5.5) -- (3.75,3.75) -- (5.5, 0.5)-- (5.5,0) -- (5,0)-- cycle;

\node[anchor =  east] at (0, 5) {$(0, 1-\delta)$};
\node[anchor =  north, xshift = -0.25 cm] at (5, 0) {$(1-\delta,0)$};
\node[anchor =  south] at (D) {$(\delta, 1)$};
\node[anchor =  west, xshift = 0 cm] at (F) {$(1,\delta)$};

\node[anchor =  east, xshift = -3 cm, yshift = -0.25 cm] (WFP) at (1.75,1.75) {\makecell{$\left(\frac{1-\delta}{1+e^\eps}, \frac{1-\delta}{1+e^\eps}\right)$\\(Worst symmetric error point)}};
\node[anchor =  west,xshift=  3cm, yshift = 0.25 cm] (BFP) at (3.75,3.75) {\makecell{$\left(\frac{e^\eps + \delta}{1+e^\eps}, \frac{e^\eps + \delta}{1+e^\eps}\right)$\\(Best symmetric error point)}};
\draw[-{Triangle[width=4pt,length=4pt]}] (WFP) -- (A);
\draw[-{Triangle[width=4pt,length=4pt]}] (BFP) -- (E);

\node[anchor = east,xshift = -0.1cm, rotate = -60] at (1,2.8) {\textcolor{red}{(a)}};
\node[anchor = east,xshift = 0.8cm, yshift = -0.5cm, rotate = -30] at (2.8,1) {\textcolor{blue}{(b)}};
\node[anchor = west,xshift = -0.1cm, rotate = -30] at (2,5) {\textcolor{green}{(c)}};

\node[anchor = south,yshift = 0.6cm] at (2.75,2.75) {$\cR(\eps,\delta)$};

\node[anchor = east,xshift = -0cm, rotate = -60] at (5,2) {\textcolor{purple}{(d)}};

\end{tikzpicture}
    \caption{Graphical representation of  $\cR(\eps,\delta)$ : characteristic region of $(\eps,\delta)$-QDP (in shaded region), where we define the boundaries  \textcolor{red}{(a)}, \textcolor{blue}{(b)}, \textcolor{green}{(c)} and \textcolor{purple}{(d)} to be $\beta = 1- \delta - e^\eps \alpha$, $\beta = e^{-\eps} (1-\delta-\alpha)$, $\beta = 1 - e^{-\eps}(\alpha - \delta)$ and $\beta = e^\eps(1-\alpha) + \delta$ respectively and $\cR(\eps,\delta)$ has two symmetric error points $(\frac{1-\delta}{1+e^\eps},\frac{1-\delta}{1+e^\eps})$ and $(\frac{e^\eps + \delta}{1+e^\eps}, \frac{e^\eps + \delta}{1+e^\eps})$ as its extremal points, where the former and the latter points are known to be the worst and the best-case symmetric error points respectively from the perspective of privacy.}
    \label{fig:privacy_region}
\end{figure}
In Figure \ref{fig:privacy_region}, we illustrate a graphical representation of the characteristic region $\cR(\eps,\delta)$ of $(\eps,\delta)$-QDP.

Therefore, under the constraints mentioned in \cref{char_region_eq1,char_region_eq2,char_region_eq3,char_region_eq4}, we define the characteristic (or operating) region of quantum $(\eps,\delta)$-QDP as follows,

\begin{definition}[Characteristic region of $(\eps,\delta)$-QDP]\label{def_DP_char}
    for some fixed $\eps \geq 0$ and $\delta \in [0,1]$, we define the characteristic region of $(\eps,\delta)$-DP as follows,

    \begin{equation}
        \cR(\eps,\delta) := \left\{(\alpha,\beta) \in [0,1]^2 : ~~\begin{split}
            \beta &\geq e^{-\eps}(1-\delta -\alpha),\\
        \beta &\geq 1 -\delta - e^{\eps}\alpha,\\
        \beta &\leq 1 - e^{-\eps}(\alpha - \delta),\\
        \beta &\leq e^{\eps}(1- \alpha) + \delta.
        \end{split}\right\}.
    \end{equation}
\end{definition}

In Figure \ref{fig:privacy_region} above, we illustrate a graphical representation of the characteristic region of $(\eps,\delta)$-DP. In \cite{KOV15}, the authors studied a characteristic region for classical $(\eps,\delta)$-DP. However, they only considered the region below $\alpha + \beta=1$. This is because, in Figure \ref{fig:privacy_region}, the characteristic region  $\cR(\eps,\delta)$ is symmetric over the straight line $\alpha + \beta=1$. Hirche et al. in \cite{HRF23} plotted $ \cR(\eps,\delta)$ for certain values of $\eps, \delta$.

From \cref{char_region_eq1,char_region_eq2,char_region_eq3,char_region_eq4}, it is trivial to see that for any pair of quantum states $(\rho,\sigma) \in \mathfrak{D}_{(\eps,\delta)}$ (for some fixed $\eps \geq 0$ and $\delta \in [0,1]$), we have $\cR(\rho,\sigma) \subseteq \cR(\eps,\delta)$. 
{ Moreover, the characteristics region $\mathcal{R}(\eps,\delta)$ is inherently closed and convex by definition, as it forms the epigraph of a function constructed from the pointwise maximum of affine functions.}

Further observe that, in Figure \ref{fig:privacy_region}, $\cR(\eps,\delta)$ has two symmetric error points (points where $\alpha = \beta$) at $(\frac{1-\delta}{e^\eps+1}, \frac{1-\delta}{e^\eps+1})$ and $(\frac{e^\eps + 
\delta}{e^\eps+1}, \frac{e^\eps +  \delta}{e^\eps+1})$. The former is the worst-case symmetric error point and the latter is the best-case symmetric error point of $\cR(\eps,\delta)$ from the perspective of privacy. This is because the worst-case symmetric error point represents the scenario where the adversary has the most amount of information about the underlying quantum state, i.e., it's least private, while the best-case symmetric error point represents the scenario where the attacker has the least information about the underlying quantum state, i.e., it's most private.

The characteristic region $\cR(\eps,\delta)$ also has six more corner/extremal points at $(0,1-\delta)$, $(1-\delta,0)$, $(\delta,1)$, $(1,\delta)$, $(1,0)$ and $(0,1)$ along with the two symmetric error points.

{ In this section, we explicitly define a pair of quantum states $(\rho_{(\eps,\delta)},\sigma_{(\eps,\delta)}) \in \mathfrak{D}_{(\eps,\delta)}$ and show that they exactly achieve the privacy region $\cR(\eps,\delta)$. While we systematically derive the structure of these states in Appendix \ref{app_optimal_states}, for the purpose of our main results, it is sufficient to explicitly define them here and verify their properties.

Consider a four-dimensional Hilbert space $\cH$ with the orthonormal basis $\{\ket{00},\ket{01},\ket{10},\ket{11}\}$. We define the pair of states as follows:
\begin{align}
    \rho_{(\eps,\delta)} & = \delta \ketbra{00} + (1-\delta)\left(\frac{e^\eps}{1+e^\eps} \ketbra{01} + \frac{1}{1+e^\eps} \ketbra{10}\right),\label{eq_rho_simp}\\
    \sigma_{(\eps,\delta)} & =  (1-\delta)\left(\frac{1}{1+e^\eps} \ketbra{01} + \frac{e^\eps}{1+e^\eps} \ketbra{10}\right) + \delta \ketbra{11}.\label{eq_sigma_simp}
\end{align}

To show that $\cR(\rho_{(\eps,\delta)},\sigma_{(\eps,\delta)}) = \cR(\eps,\delta)$, we demonstrate that under specific choices of measurement projectors $M$, the pair of Type-I error ($\alpha = 1 - \tr[M\rho_{(\eps,\delta)}]$) and Type-II error ($\beta = \tr[M\sigma_{(\eps,\delta)}]$) achieves all the extremal corner points of $\cR(\eps,\delta)$. 
\begin{itemize}
    \item For $M = \ketbra{00}$, we have $\tr[\ketbra{00}\rho_{(\eps,\delta)}] = \delta$ and $\tr[\ketbra{00}\sigma_{(\eps,\delta)}] = 0$. This yields $\alpha = 1 - \delta$ and $\beta = 0$, achieving the corner point $(1-\delta,0)$.
    \item Consequently, the complementary projector $M = \bbI - \ketbra{00}$ yields $\alpha = \delta$ and $\beta = 1$, achieving the corner point $(\delta,1)$.
    \item For $M = \ketbra{11}$, we have $\tr[\ketbra{11}\rho_{(\eps,\delta)}] = 0$ and $\tr[\ketbra{11}\sigma_{(\eps,\delta)}] = \delta$. This yields $\alpha = 1$ and $\beta = \delta$, achieving the corner point $(1,\delta)$.
    \item Consequently, the complementary projector $M = \bbI - \ketbra{11}$ yields $\alpha = 0$ and $\beta = 1-\delta$, achieving the corner point $(0,1-\delta)$.
    \item For $M = \ketbra{00} + \ketbra{01}$, we have $\tr[M\rho_{(\eps,\delta)}] = \frac{e^\eps+\delta}{1+e^\eps}$ and $\tr[M\sigma_{(\eps,\delta)}] = \frac{1-\delta}{1+e^\eps}$. This yields $\alpha = \frac{1-\delta}{1+e^\eps}$ and $\beta = \frac{1-\delta}{1+e^\eps}$, achieving the worst-case symmetric error point $\left(\frac{1-\delta}{1+e^\eps}, \frac{1-\delta}{1+e^\eps}\right)$.
    \item For $M = \ketbra{11} + \ketbra{10}$, we have $\tr[M\rho_{(\eps,\delta)}] = \frac{1-\delta}{1+e^\eps}$ and $\tr[M\sigma_{(\eps,\delta)}] = \frac{e^\eps+\delta}{1+e^\eps}$. This yields $\alpha = \frac{e^\eps+\delta}{1+e^\eps}$ and $\beta = \frac{e^\eps+\delta}{1+e^\eps}$, achieving the best-case symmetric error point $\left(\frac{e^\eps + \delta}{1+e^\eps}, \frac{e^\eps + \delta}{1+e^\eps}\right)$.
\end{itemize}
Furthermore, the trivial deterministic projectors $M = \bbI$ and $M = 0$ naturally achieve the extremal points $(0,1)$ and $(1,0)$, respectively. Specifically, $M = \bbI$ yields $\alpha = 0, \beta = 1$, and $M = 0$ yields $\alpha = 1, \beta = 0$.

Since the pair of quantum states achieves all the non-trivial corner points as well as the trivial boundary points of $\cR(\eps,\delta)$, and the achievable region for any pair of states is necessarily convex, we conclude that $\cR(\rho_{(\eps,\delta)},\sigma_{(\eps,\delta)}) = \cR(\eps,\delta)$. }

\begin{figure}[h]
    \centering
    \begin{tikzpicture}

\coordinate (A) at (1.75,1.75);
\coordinate (B) at (0,5);
\coordinate (C) at (0,5.5);
\coordinate (D) at (0.5,5.5);
\coordinate (E) at (3.75,3.75);
\coordinate (F) at (5.5,0.5);
\coordinate (G) at (5.5,0);
\coordinate (H) at (5,0);
\coordinate (I) at (1.25,1.25);
\coordinate (J) at (4.25,4.25);

\draw[fill=red] (A) circle (0.04); 
\draw[fill=red] (B) circle (0.04); 
\draw[fill=red] (C) circle (0.04); 
\draw[fill=red] (D) circle (0.04); 
\draw[fill=red] (E) circle (0.04); 
\draw[fill=red] (F) circle (0.04); 
\draw[fill=red] (G) circle (0.04); 
\draw[fill=red] (H) circle (0.04); 
\draw[fill=red] (I) circle (0.04);
\draw[fill=red] (J) circle (0.04);
\draw[fill=red] (0,0) circle (0.04); 
\draw[fill=red] (5.5,5.5) circle (0.04);

\draw[->] (0,0) -- (6,0) node[below,xshift = -4cm,yshift = 0.1cm] {\makecell{$\alpha \longrightarrow$ \\ (Type-I error)}};
\draw[->] (0,0) -- (0,6) node[left, rotate = 90, xshift = -1.75cm, yshift = 0.5cm] {\makecell{$\beta \longrightarrow$ \\ (Type-II error)}};
\draw[dotted,thick] (0,0) -- (5.5,5.5) node[anchor = north east, rotate = 45,xshift = -0.75cm] {$\alpha = \beta$};
\draw[dotted,thick] (0,5.5) -- (5.5,0 ) node[anchor = north east, rotate = -45,xshift = -2cm] {$\alpha + \beta = 1$};
\draw[-] (0,5.5) -- (5.5,5.5) node[above, xshift = -2.5cm] {$\beta = 1$};
\draw[-] (5.5,0) -- (5.5,5.5) node[right, xshift = 0.25cm, yshift = -3cm, rotate = 90] {$\alpha = 1$};
\node[anchor = north east] at (0,0) {$(0,0)$};
\node[anchor =  east] at (0,5.5) {$(0,1)$};
\node[anchor = north west] at (5.5,0) {$(1,0)$};
\node[anchor = south west] at (5.5,5.5) {$(1,1)$};

\filldraw[fill=red, pattern=north east lines, draw=black, very thick, opacity=0.5] (1.75,1.75) -- (0,5) -- (0,5.5)-- (0.5,5.5) -- (3.75,3.75) -- (5.5, 0.5)-- (5.5,0) -- (5,0)-- cycle;
\filldraw[fill= blue, pattern=north west lines, draw=black, very thick, opacity=0.5] (I) -- (0,5) -- (0,5.5)-- (0.5,5.5) -- (J) -- (5.5, 0.5)-- (5.5,0) -- (5,0)-- cycle;

\node[anchor =  east] at (0, 5) {$(0, 1-\delta)$};
\node[anchor =  north, xshift = -0.25 cm] at (5, 0) {$(1-\delta,0)$};
\node[anchor =  south] at (D) {$(\delta, 1)$};
\node[anchor =  west, xshift = 0 cm] at (F) {$(1,\delta)$};

\node[anchor =  east, xshift = -3 cm, yshift = -0.25 cm] (WFP) at (I) {\makecell{$\left(\frac{1 -\delta}{1+e^\eps} - \delta,\frac{1 -\delta}{1+e^\eps} - \delta\right)$ \\ (Point outside $\cR(\eps,\delta)$)}};
\node[anchor =  west,xshift=  3cm, yshift = 0.25 cm] (BFP) at (J) {\makecell{$\left(\frac{e^\eps+\delta}{1 + e^\eps} +\delta, \frac{e^\eps+\delta}{1 + e^\eps} +\delta\right)$\\(Point outside $\cR(\eps,\delta)$)}};
\draw[-{Triangle[width=4pt,length=4pt]}] (WFP) -- (I);
\draw[-{Triangle[width=4pt,length=4pt]}] (BFP) -- (J);

\end{tikzpicture}
    \caption{Graphical representation of  $\cR(\rho'_{(\eps,\delta)},\sigma'_{(\eps,\delta)} )$ (in the whole shaded region) which is strictly larger than $\cR(\eps,\delta)$ (in inner shaded region), where the two extremal points of $\cR(\rho'_{(\eps,\delta)},\sigma'_{(\eps,\delta)} )$ outside $\cR(\eps,\delta)$ are $\left(\frac{1 -\delta}{1+e^\eps} - \delta,\frac{1 -\delta}{1+e^\eps} - \delta\right)$ and $\left(\frac{e^\eps+\delta}{1 + e^\eps} +\delta, \frac{e^\eps+\delta}{1 + e^\eps} +\delta\right)$.}
    \label{fig:privacy_region_invalid}
\end{figure}

\begin{remark}
    In \cite{BG17}, the authors obtained a figure similar to Figure \ref{fig:privacy_region} by plotting type-1 and type-2 errors for a pair of states $\{\rho_1,\rho_2\}$. In contrast, Figure \ref{fig:privacy_region} is with respect to all pairs $\{\rho_1,\rho_2\}$ which are $(\eps,\delta)$-DP and therefore it helps to characterize the worst-case (most informative) $(\eps,\delta)$-DP pair of quantum states.
\end{remark}

\section{Quantum Information Ordering and Differential Privacy}\label{sec_diff_priv_quant}

\subsection{Quantum ordering of informativeness}

In \cite{Blackwell51,BG54}, Blackwell introduced a notion of an ordering of classical informativeness between two pairs of distributions $\{P_1,P_2\}$ and $\{Q_1,Q_2\}$. { We formalize this order via the hypothesis testing divergence:

\begin{definition}[Blackwell's Order of Informativeness]\label{def_blackwell_order}
    Given two pairs of probability distributions $(P_1,P_2)$ and $(Q_1,Q_2)$ over two finite sets $\cX$ and $\cY$, we say that $\{P_1,P_2\}$ dominates $\{Q_1,Q_2\}$ in the sense of Blackwell (denoted $\{P_1,P_2\} \succeq_{B} \{Q_1,Q_2\}$) if for all $\alpha \in [0,1]$,
    $$ D^{\alpha}_{H}{(P_1\|P_2)} \geq D^{\alpha}_{H}{(Q_1\|Q_2)} $$
    where $D^{\alpha}_{H}{(\cdot\|\cdot)}$ is the classical hypothesis testing divergence of order $\alpha$ (see Definition \ref{def_hyp_div}).
\end{definition}

In particular, Blackwell proved that this testing-based dominance is perfectly equivalent to the existence of a stochastic map transforming one experiment into the other:

\begin{proposition}[Blackwell's Theorem of Informativeness Ordering {\cite[Theorems $12.2.2$ \& $12.4.2$]{BG54}}]\label{prop_blackwell}
    Given two pairs of probability distributions $(P_1,P_2)$ and $(Q_1,Q_2)$ over two finite sets $\cX$ and $\cY$, respectively, the following statements are equivalent:
    \begin{enumerate}
        \item $\{P_1,P_2\} \succeq_{B} \{Q_1,Q_2\}$.
        \item There exists a stochastic map $\cT : \cX \to \cY$ such that $Q_i = \cT(P_i)$ for each $i \in \{1,2\}$.\label{map_exist}
    \end{enumerate}
\end{proposition}}

{ Analogous to Blackwell's order of classical informativeness (as mentioned in Proposition 1), the concept of ordering quantum statistical experiments based on their distinguishability has been explored extensively in the literature (e.g., by Buscemi and Gour \cite{BG17}). Building upon these foundational concepts, we adopt this quantum information order specifically via the quantum hypothesis testing divergence to explicitly map it to the operational constraints of differential privacy. We formalize this utilized order in the definition below.}



\begin{definition}[Quantum information  order]\label{def_quant_blackwell}
    Consider two pairs of quantum states $(\rho_1,\rho_2)$ and  $(\sigma_1,\sigma_2)$. We say that the pair $\left(\rho_1,\rho_2\right)$ dominates (or, is more informative than) the  pair $\left(\sigma_1,\sigma_2\right)$ (denoted as $\left\{\rho_1,\rho_2\right\} \succeq_{Q} \left\{\sigma_1,\sigma_2\right\}$), if for any $\alpha \in [0,1]$, the following holds,
    \begin{equation}
        D^{\alpha}_{H}{(\rho_1\|\rho_2)} \geq D^{\alpha}_{H}{(\sigma_1\|\sigma_2)}.\label{def_quant_blackwell_eq}
    \end{equation}
\end{definition}

 Intuitively, the order $\left\{\rho_1,\rho_2\right\} \succeq_{Q} \left\{\sigma_1,\sigma_2\right\}$ indicates that the states $\rho_1$ and $\rho_2$ are more distinguishable as compared to the states $\sigma_1$ and $\sigma_2$. 

Unlike the classical case, in the quantum setting, from the results of \cite{Matsumoto2014} and  of \cite[Lemma $18$ \& Corollary $19$]{GLW24}, it is known that \eqref{def_quant_blackwell_eq} \textit{does not} always imply the existence of a CP-TP map $\cE$ such that $\sigma_i = \cE(\rho_i)$ for each $i \in \{1,2\}$. This yields that we \textit{can not} invoke the monotonicity property of quantum divergences (see Definition \ref{def_data_proc_func}) to conclude that $\varmathbb{D}(\rho_1\| \rho_2) \geq \varmathbb{D}(\sigma_1\| \sigma_2)$.

However, we can circumvent this issue by directly showing that if for any two pairs of quantum states  $(\rho_1,\rho_2)$ and $(\sigma_1,\sigma_2)$, \eqref{def_quant_blackwell_eq} holds, then, every quantum $f$-divergence (see Definition \ref{def_f_div_integral_rep}) between the former pair dominates that of the latter pair. We formalize this claim in the Theorem below.

 \begin{theorem}\label{lemma_t_hyp_F_Div}
    Consider two pairs of quantum states  $(\rho_1,\rho_2)$ and $(\sigma_1,\sigma_2)$, for which the following holds,

\begin{equation}
    D^{\alpha}_{H}{(\rho_1\|\rho_2)} \geq D^{\alpha}_{H}{(\sigma_1\|\sigma_2)}, \forall \alpha \in [0,1].\label{lemma_t_hyp_F_div_eq}
\end{equation}

Then, for any $f \in \cF$ (see Definition \ref{def_f_div_integral_rep}) we have, 
$$D_f(\rho_1\|\rho_2) \geq  D_f(\sigma_1\|\sigma_2),$$ whenever $D_f$ is well-defined (see Definition \ref{def_f_div_integral_rep}) for the pairs $(\rho_1,\rho_2)$ and $(\sigma_1,\sigma_2)$.
\end{theorem}

    Before proceeding with the proof of Theorem \ref{lemma_t_hyp_F_Div}, we first state the following lemma, which will be required in proving Theorem \ref{lemma_t_hyp_F_Div}.
    \begin{lemma}\label{fact_equi_ah_sd}
        Consider $\rho_1,\rho_2 \in \cD(\cH_A)$ and $\sigma_1,\sigma_2 \in \cD(\cH_B)$. If
        \begin{equation}
            D^{\alpha}_{H}{(\rho_1\|\rho_2)} \geq D^{\alpha}_{H}{(\sigma_1\|\sigma_2)}, \forall \alpha \in [0,1],\label{fact_equi_ah_sd_eq}
        \end{equation}
        then,
\begin{align}
            E_{\gamma}(\rho_1\|\rho_2) &\geq E_{\gamma}(\sigma_1\|\sigma_2), \forall \gamma\geq 0, \label{BN1}
        \end{align}
        where $E_{\gamma} (\cdot \|\cdot)$ 
        is defined in Definition \ref{f_sym_hyp_func}.
    \end{lemma}
    \begin{proof}
        See Appendix \ref{proof_fact_equi_ah_sd} for the proof.
    \end{proof}
    Using the above lemma, we now prove Theorem \ref{lemma_t_hyp_F_Div} using the following series of inequalities,
    \begin{align*}
        D_{f}(\rho_1\|\rho_2) &\overset{(a)}{=} \int_{0}^{\infty} f''(\gamma) E_{\gamma}(\rho_1\|\rho_2) d\gamma\nn\\
        &\overset{(b)}{\geq} \int_{0}^{\infty} f''(\gamma) E_{\gamma}(\sigma_1\|\sigma_2) d\gamma\nn\\
        &\overset{(c)}{=} D_{f}(\sigma_1\|\sigma_2),
    \end{align*}
    where $(a)$ and $(c)$ follows from Fact \ref{fact_f_div_integral_rep} and $(b)$ follows from \eqref{BN1} of Lemma \ref{fact_equi_ah_sd} and the fact that the double derivative of $f$ exists and $f$ is a convex function. This completes the proof of Theorem \ref{lemma_t_hyp_F_Div}.

\begin{remark}
    We note here that \eqref{BN1} of Lemma \ref{fact_equi_ah_sd} was also proven in \cite[Theorem $2$]{BG17}.
\end{remark}

{\begin{remark}\label{remark_literature_connection}
    There have been various notions of ordering of informativeness in the quantum case in earlier works (see, for example, \cite{ReebKastoryanoWolf2011, Buscemi2012, Jencova2012, Matsumoto:2015vgh, Buscemi2016, BG17, BuscemiSutterTomamichel2019, WW19}). These partial orders have built upon foundational work in the comparison of quantum statistical models \cite{Buscemi2012, Buscemi2016, BG17} and Hilbert's projective metric \cite{ReebKastoryanoWolf2011}. Specifically, the generalization of quantum relative Lorenz curves to quantum binary experiments \cite{Jencova2012} yields the exact partial order we discuss here when restricted to binary measurements.

    Crucially, however, \cite{Matsumoto2014,GLW24} also established that, in general, unlike the classical Blackwell setting, no completely positive trace-preserving (CPTP) map exists that transforms the dominating pair into the dominated pair. The absence of such a map traditionally hinders the ability to upper-bound divergences using standard data-processing inequalities. 
    
    Our primary technical contribution lies in overcoming this limitation. We demonstrate that even without the existence of a connecting CPTP map, we can map this statistical ordering directly to the asymmetric constraints of $(\eps,\delta)$-quantum differential privacy. By doing so, we explicitly construct the optimal bounding states $(\rho_{(\eps,\delta)}, \sigma_{(\eps,\delta)})$ and successfully derive tight bounds for generalized quantum divergences, entirely bypassing the need for standard data-processing arguments.
\end{remark}}

\subsection{Differential privacy: A quantum hypothesis testing perspective}

In this subsection, we consider differential privacy from a quantum hypothesis testing perspective. The main goal of quantum differential privacy is to ensure that a quantum mechanism produces output states that are nearly indistinguishable when applied to neighboring quantum inputs. Formally, for two neighboring quantum inputs producing output states $\rho, \sigma$, the distinguishability between $\rho$ and $\sigma$ must be limited by the privacy parameters. Intuitively, this means that even if an investigator (adversary) knows the mechanism and observes its output, the investigator should not be able to reliably infer which quantum input (respondent) generated it. More formally, with high probability, no hypothesis test—regardless of the investigator’s strategy—can reliably infer the respondent’s individual contribution from the output.

Motivated by the close relation between privacy and hypothesis testing, in the lemma below, we give an equivalent condition for a pair of quantum states to be $(\eps,\delta)$-DP (see Definition \ref{def_DP_state}). 

\begin{lemma}\label{lemma_DP_t_hyp_rel}
    Consider any pair of quantum states $\rho, \sigma \in \cD(\cH_A)$. Then, the pair $(\rho,\sigma) \in \mathfrak{D}_{(\eps,\delta)}$ (see Definition \ref{def_DP_state}) for some fixed $\eps \geq 0$ and $\delta \in [0,1)$ if and only if the following holds:
\begin{equation}
    D^{\alpha}_{H} {(\rho\|\sigma)} \leq -\log f_{\eps,\delta}(\alpha), \forall \alpha \in [0,1], \label{lemma_DP_t_hyp_rel_eq}
\end{equation}
or, equivalently,
\begin{equation}
    \beta(\alpha|\rho\|\sigma) \geq  f_{\eps,\delta}(\alpha), \forall \alpha \in [0,1], \label{lemma_DP_t_hyp_rel_eq1}
\end{equation}
where $f_{\eps,\delta}(\alpha) := \max\left\{1 - \delta - e^{\eps}\alpha, e^{-\eps}(1 -\delta-\alpha),0\right\}$ for any $\alpha \in [0,1]$.
\end{lemma}

\begin{proof}
    See Appendix \ref{proof_lemma_DP_t_hyp_rel} for the proof.
\end{proof}

Note that the above version of quantum $(\eps,\delta)$-DP implies that if for each $\alpha \in [0,1]$ $D^{\alpha}_H(\rho'\|\sigma') =-\log f_{\eps,\delta}(\alpha)$ for some pair of quantum states $\rho',\sigma'$ over any arbitrary finite dimensional Hilbert space, then for any pair of quantum states $\rho,\sigma \in \mathfrak{D}_{(\eps,\delta)}$ (see Definition \ref{def_DP_state}) is at least as hard as distinguishing between the pair $(\rho',\sigma')$.
The above intuition provides a notion of the weakest (most informative) quantum $(\eps,\delta)$-DP pairs of quantum states, which are the least ``hard to differentiate''. We formalize this notion in the definition below.

\begin{definition}[Weakest quantum $(\eps,\delta)$-DP]\label{def_opt_DP_state}
     For $\eps \geq 0 $ and $\delta \in [0,1),$ a pair of quantum states $(\rho,\sigma)$ is defined to be the weakest (most informative) quantum $(\eps,\delta)$-DP if
\begin{equation}
    D^{\alpha}_{H} {(\rho\|\sigma)} = -\log f_{\eps,\delta}(\alpha), \forall \alpha \in [0,1], \label{def_opt_DP_state_eq}
\end{equation}
\end{definition}

In the lemma below, we show the existence of a pair of quantum states that is the weakest (most informative) quantum $(\eps,\delta)$-DP  for some fixed $\eps \geq 0$ and $\delta \in [0,1]$.

\begin{lemma}\label{lemma_diff_priv_quant}
    Consider the quantum states $\rho_{(\eps,\delta)}$ and $\sigma_{(\eps,\delta)}$ mentioned in \eqref{eq_rho_simp} and \eqref{eq_sigma_simp}, for some fixed $\eps \geq 0$ and $\delta \in [0,1]$. 
    Then, we have $D^{\alpha}_{H}(\rho_{(\eps,\delta)}\|\sigma_{(\eps,\delta)}) = -\log f_{\eps,\delta}(\alpha)$ for any $\alpha \in [0,1]$.
\end{lemma}
\begin{proof}
    See Appendix \ref{proof_lemma_diff_priv_quant} for the proof.
\end{proof}

Therefore, for a pair $(\rho,\sigma) \in \mathfrak{D}_{(\varepsilon,\delta)}$ with some fixed $\varepsilon \geq 0$ and $\delta \in [0,1]$,
Lemmas~\ref{lemma_DP_t_hyp_rel} and~\ref{lemma_diff_priv_quant} guarantee the relation,
    \begin{align}
        D^{\alpha}_{H}(\rho\|\sigma) &\leq -\log f_{\eps,\delta}(\alpha)
        = D^{\alpha}_{H}(\rho_{(\eps,\delta)}\|\sigma_{(\eps,\delta)}),\label{theo_QLDP_approx_F_div_eq1}
    \end{align}
 for every $\alpha \in [0,1]$.

 Further, from Lemmas \ref{lemma_DP_t_hyp_rel} and \ref{lemma_diff_priv_quant} and Definition \ref{def_quant_blackwell}, we note that any pair of quantum states $(\rho,\sigma) \in \mathfrak{D}_{(\eps,\delta)}$ for some fixed $\eps \geq 0$ and $\delta \in [0,1]$, gets dominated by the pair $(\rho_{(\eps,\delta)},\sigma_{(\eps,\delta)})$. Formally, it implies the following.

\begin{corollary}\label{corr_QLDP_blackwell}
    If any pair of quantum states $(\rho,\sigma) \in \mathfrak{D}_{(\eps,\delta)}$ for some fixed $\eps \geq 0$ and $\delta \in [0,1]$, then the following holds,
    \begin{equation}
       \left\{\rho_{(\eps,\delta)},\sigma_{(\eps,\delta)}\right\} \succeq_{Q}  \left\{\rho,\sigma\right\} .\label{corr_QLDP_blackwell_eq}
    \end{equation}
\end{corollary}

\begin{remark}\hfill
   \begin{itemize}
    \item { Note that the quantum states $\rho_{(\eps,\delta)}$ and $\sigma_{(\eps,\delta)}$ mentioned in Lemma \ref{lemma_diff_priv_quant}, which satisfy the weakest (most informative) $(\eps,\delta)$-DP condition for fixed $\eps \geq 0$ and $\delta \in [0,1]$, achieve the symmetric error point where the Type-I error equals the Type-II error (i.e., $\alpha = \beta = \frac{1-\delta}{e^\eps+1}$). This point is operationally significant as it represents the minimum symmetric testing probability for distinguishing this optimal pair of states.}
    \item For $\delta = 0$ and any $\eps \geq 0$, we define two quantum states, 
    \begin{align}
        \rho_{(\eps)} &:= \frac{e^\eps}{1+e^\eps}\ketbra{0} + \frac{1}{1+e^\eps}\ketbra{1},\label{rho_pure_weak}\\
        \sigma_{(\eps)} &:= \frac{1}{1+e^\eps}\ketbra{0} + \frac{e^\eps}{1+e^\eps}\ketbra{1},\label{sigma_pure_weak}
    \end{align}
     Observe that it can be verified that $\rho_{(\eps)}$ and $\sigma_{(\eps)}$ are the weakest (most informative) $\eps$-DP pair of quantum states.
   \end{itemize}
\end{remark}

We now state the main theorem of this section, which provides a bound on $D_f(\rho\|\sigma)$ (see Definition \ref{def_f_div_integral_rep}) for any pair of quantum states $(\rho,\sigma) \in \mathfrak{D}_{(\eps,\delta)}$.

\begin{corollary}\label{theo_QLDP_approx_F_div_bound}
    If a pair of quantum states $(\rho,\sigma)$ belongs to $\mathfrak{D}_{(\eps,\delta)}$ for some fixed $\eps \geq 0$ and $\delta \in [0,1]$, 
    then for any $f \in \cF$, we have,
    \begin{equation}
        D_f(\rho\|\sigma) \leq D_f(\rho_{(\eps,\delta)}\|\sigma_{(\eps,\delta)}),\label{theo_QLDP_approx_F_div_eq}
    \end{equation}
    whenever $D_f$ is well-defined for the pairs $(\rho_1,\rho_2)$ and $(\rho_{(\eps,\delta)},\sigma_{(\eps,\delta)})$.
Also, for any $\gamma \geq 0$
    and $\alpha \in [0,1) \cup (1,+\infty)$, we have,
    \begin{align}
        E_{\gamma}(\rho\|\sigma)& \leq E_{\gamma}(\rho_{(\eps,\delta)}\|\sigma_{(\eps,\delta)}),\label{corr_QLDP_approx_hockey_eq} \\
        D_{\alpha}(\rho\|\sigma)& \leq D_{\alpha}(\rho_{(\eps,\delta)}\|\sigma_{(\eps,\delta)}),\label{corr_QLDP_P} \\
        \check{D}_{\alpha}(\rho\|\sigma)& \leq \check{D}_{\alpha}(\rho_{(\eps,\delta)}\|\sigma_{(\eps,\delta)}).\label{corr_QLDP_S}
    \end{align}
\end{corollary}
\begin{proof}
Suppose a pair $(\rho,\sigma)$ belongs to 
$\mathfrak{D}_{(\varepsilon,\delta)}$ for some fixed $\varepsilon \geq 0$ and $\delta \in [0,1]$. 
   Then, the combination of \eqref{theo_QLDP_approx_F_div_eq1} and
Theorem \ref{lemma_t_hyp_F_Div} implies the relation $D_f(\rho\|\sigma) \leq D_f(\rho_{(\eps,\delta)}\|\sigma_{(\eps,\delta)})$. 
{ The combination of \eqref{theo_QLDP_approx_F_div_eq1} and
Lemma \ref{fact_equi_ah_sd}
also implies 
\eqref{corr_QLDP_approx_hockey_eq}. Further, the combination of the relation in \eqref{corr_QLDP_approx_hockey_eq} and \eqref{def_P_div_integral_rep_eq} of Definition \ref{def_quant_renyi_hockey} yields \eqref{corr_QLDP_P}. }
    Finally, using Fact \ref{fact_P_div_integral_rep}, we have, 
    \begin{align}
            \check{D}_{\alpha}(\rho\|\sigma)\leq 
                    D_{\alpha}(\rho\|\sigma) \leq D_{\alpha}(\rho_{(\eps,\delta)}\|\sigma_{(\eps,\delta)})=
            \check{D}_{\alpha}(\rho_{(\eps,\delta)}\|\sigma_{(\eps,\delta)}),    \end{align}
    where the final equality follows from Fact \ref{fact:commuting_renyi_full}, since 
    $\rho_{(\eps,\delta)}$ and $\sigma_{(\eps,\delta)}$
   are commutative with each other. 
    This proves Corollary \ref{theo_QLDP_approx_F_div_bound}.
\end{proof}

Further, for $\delta =0$, the theorem below gives an upper-bound on $D_f$ (for any $f \in \cF$) between any pair of quantum states $(\rho,\sigma) \in \mathfrak{D}_{(\eps,0)}$.
\begin{corollary}\label{theo_QLDP_F_div_pure_bound}
    If any pair of quantum states $(\rho,\sigma) \in \mathfrak{D}_{(\eps,0)}$ for some fixed $\eps \geq 0$, then for any $f \in \cF$, $\gamma \geq 0$
    and $\alpha \in [0,1) \cup (1,+\infty)$, we have,
    \begin{align}
        D_f(\rho\|\sigma) &\leq D_f(\rho_{(\eps)}\|\sigma_{(\eps)}),\label{theo_QLDP_F_div_pure_bound_eq}\\
        E_{\gamma}(\rho\|\sigma)& \leq E_{\gamma}(\rho_{(\eps)}\|\sigma_{(\eps)}),\label{theo_QLDP_F_div_pure_bound_eq2} \\
        D_{\alpha}(\rho\|\sigma)
        & \leq D_{\alpha}(\rho_{(\eps)}\|\sigma_{(\eps)}),\label{theo_QLDP_F_div_pure_bound_eq3}\\
        \check{D}_{\alpha}(\rho\|\sigma)& \leq 
        \check{D}_{\alpha}(\rho_{(\eps)}\|\sigma_{(\eps)}).\label{theo_QLDP_F_div_pure_bound_eq4}
    \end{align}
   whenever $D_f$ is well-defined for the pairs $(\rho,\sigma)$ and $(\rho_{(\eps)},\sigma_{(\eps)})$ (defined in \eqref{rho_pure_weak} and \eqref{sigma_pure_weak} respectively).
\end{corollary}

\begin{remark}\label{special_data_proc}
    Note that for some fixed $\eps \geq 0$, {  the pair $(\rho_{(\eps)},\sigma_{(\eps)})$ consists of commuting qubit (two-dimensional) states, and $D_{\max}(\rho_{(\eps)}\|\sigma_{(\eps)})= D_{\max}(\sigma_{(\eps)}\|\rho_{(\eps)}) = \eps$.}  Further, Definition \ref{def_DP_state} directly guarantees that any pair of quantum states $(\rho,\sigma) \in \mathfrak{D}_{(\eps,0)}$ satisfies 
    \begin{align*}
        D_{\max}(\rho\|\sigma) & \leq \eps = D_{\max}(\rho_{(\eps)}\|\sigma_{(\eps)}), \\
        D_{\max}(\sigma\|\rho) &\leq \eps = D_{\max}(\sigma_{(\eps)}\|\rho_{(\eps)}).
    \end{align*}
    Therefore, from Fact \ref{fact_state_trans}, there exists a CP-TP map $\cE$, such that $\rho = \cE(\rho_{(\eps)}) $ and $\sigma =\cE(\sigma_{(\eps)})$. Hence, Definition \ref{def_data_proc_func} implies 
    \begin{equation*}
        \varmathbb{D}(\rho\|\sigma) \leq \varmathbb{D}(\rho_{(\eps)}\|\sigma_{(\eps)}).
    \end{equation*}
\end{remark}

It is important to note that for $\delta > 0$, $\supp(\rho_{(\eps,\delta)}) \not\subseteq \supp(\sigma_{(\eps,\delta)})$. Therefore, the divergence that requires support inclusion is not well-defined for the pair of quantum states $\rho_{(\eps,\delta)}$ and $\sigma_{(\eps,\delta)}$.
However, when $\delta = 0$, since $\supp(\rho_{\eps}) = \supp(\sigma_{\eps})$, for any $\eps$-DP quantum state pairs, using Corollary \ref{theo_QLDP_F_div_pure_bound}, we derive upper-bounds on quantum relative entropy and $L_1$ distance as mentioned in the corollary below.
\begin{corollary}\label{corr_QLDP_Divs_bound}
    If a pair of quantum states $(\rho,\sigma) \in \mathfrak{D}_{(\eps,0)}$ for some fixed $\eps \geq 0$, 
    then, we have the following upper-bounds,

    \begin{enumerate}[label=(\roman*)]
        \item $D(\rho\|\sigma) \leq  \eps \tanh \left(\frac{\eps}{2}\right)$\footnote{ We note that the bound in (i) can also be perfectly recovered by applying the known reverse Pinsker inequality (see \cite[Proposition $5.3$]{Hirche_2024}) under pure differential privacy, $D(\rho\|\sigma) \leq \frac{\eps}{2} \|\rho - \sigma\|_1$, directly to the trace distance bound established in (ii).}\label{corr_QLDP_KL_bound_eq},
        \vspace*{5pt} 
        \item $\norm{\rho-\sigma}{1} \leq 2 \tanh \left(\frac{\eps}{2}\right)$. \label{corr_QLDP_Hockey_Stick_bound_eq}
    \end{enumerate}
\end{corollary}



\begin{proof}
    \hfill\\
     {\emph{({\romannumeral 1})} As established in Remark \ref{special_data_proc}, for any pair of quantum states $(\rho, \sigma) \in \mathfrak{D}_{(\eps,0)}$, there exists a completely positive trace-preserving (CP-TP) map $\cE$ such that $\rho = \cE(\rho_{(\eps)})$ and $\sigma = \cE(\sigma_{(\eps)})$. 
    
    By setting the generic divergence $\varmathbb{D}(\cdot\|\cdot)$ to the quantum relative entropy $D(\cdot\|\cdot)$, the data processing inequality (Definition \ref{def_data_proc_func}) directly implies:
    \begin{equation}
        D(\rho \| \sigma) \leq D(\rho_{(\eps)} \| \sigma_{(\eps)}).\label{dpi_relative_entropy_eq}
    \end{equation}
    
    Because the weakest pure $\eps$-DP states $\rho_{(\eps)}$ and $\sigma_{(\eps)}$ (defined in \eqref{rho_pure_weak} and \eqref{sigma_pure_weak}) are diagonal in the same basis, they commute. We can therefore exactly evaluate their quantum relative entropy using their eigenvalues:
    \begin{align}
        D(\rho_{(\eps)}\|\sigma_{(\eps)}) &= \frac{e^\eps}{1+e^\eps}\log\left(\frac{e^\eps/(1+e^\eps)}{1/(1+e^\eps)}\right) + \frac{1}{1+e^\eps}\log\left(\frac{1/(1+e^\eps)}{e^\eps/(1+e^\eps)}\right)\nn\\
        &= \frac{e^\eps}{1+e^\eps}(\eps) + \frac{1}{1+e^\eps}(-\eps)\nn\\
        &= \eps\left(\frac{e^\eps - 1}{e^\eps + 1}\right)\nn\\
        &= \eps\left(\frac{e^{\eps/2}(e^{\eps/2} - e^{-\eps/2})}{e^{\eps/2}(e^{\eps/2} + e^{-\eps/2})}\right)\nn\\
        &= \eps \tanh\left(\frac{\eps}{2}\right). \label{exact_kl_pure_dp}
    \end{align}
    
    Substituting \eqref{exact_kl_pure_dp} into \eqref{dpi_relative_entropy_eq} yields $D(\rho\|\sigma) \leq  \eps \tanh\left(\frac{\eps}{2}\right)$. This completes the proof of \emph{(i)}.\hfill\\}
\emph{({\romannumeral 2})} The pair $(\rho_{(\eps)},\sigma_{(\eps)})$ satisfies the following,
    \begin{align}
       \norm{\rho_{(\eps)}-\sigma_{(\eps)}}{1} &= \abs{\frac{e^\eps - 1}{1+ e^\eps}} + \abs{\frac{1 - e^\eps}{1+ e^\eps}}\nn\\
      &= 2\left(\frac{e^\eps - 1}{1+ e^\eps}\right)
      = 2 \tanh\left(\frac{\eps}{2}\right).
    \end{align}

Note that we have $\norm{\rho-\sigma}{1} = 2 \abs{\rho-\sigma}_{+} = 2 E_{1}(\rho\|\sigma)$, where the last equality follows from Fact \ref{Fact_hockey_equiv_form}.
 Therefore, the relation \eqref{theo_QLDP_F_div_pure_bound_eq2} of Corollary \ref{theo_QLDP_F_div_pure_bound} derives that $\norm{\rho-\sigma}{1} = 2E_{1}(\rho\|\sigma) \leq  2E_{1}(\rho_{(\eps)}\|\sigma_{(\eps)}) = \norm{\rho_{(\eps)}-\sigma_{(\eps)}}{1}$. This completes the proof of \emph{(ii)}.    
\end{proof}

\section{Quantum Privatized Parameter Inference}\label{sec:fisher}
In the context of randomized response, a fundamental challenge is balancing the privacy of respondents with the statistical utility of the information gathered by an investigator.
Let's consider a scenario where binary information, represented by $\{0,1\}$, is encoded into two quantum states, $\rho$ and $\sigma$. The privacy of this encoding is quantified by the $(\eps,\delta)$-DP constraint on the pair $(\rho, \sigma)$.

We assume that 
the behavior of the respondents is considered as random sampling with replacement from the investigator's perspective,
and 
the investigator is interested only in the ratio between $0$ and $1$ so that
the underlying statistical information $\theta$
follows a binomial distribution $p_\theta$, where $p_\theta(0)=\theta$ and $p_\theta(1)=1-\theta$. 
That is, the key task 
is the inference of the parameter $\theta$
by the investigator.

With the quantum encoding, this translates to a parametrized state $\rho_\theta:=\theta \rho+(1-\theta)\sigma$, where $\theta \in [0,1]$. If the investigator collects $n$ responses, the composite state is $\rho_\theta^{\otimes n}$.
Also, we define $\rho_{(\eps,\delta),\theta}$ in the same way when respondents encode their information by using 
$\rho_{(\eps,\delta)}$ and $\sigma_{(\eps,\delta)}$.
More generally, in quantum parameter inference, the goal is to infer an unknown parameter $\theta$ encoded in a family of quantum states $\{\rho_\theta\}$, where the encoding is constrained by privacy requirements. In the privatized setting, 
the states $\rho$ and $\sigma$ must satisfy the $(\eps,\delta)$-DP condition, which restricts the distinguishability of the states and thus limits the amount of information that can be extracted about $\theta$.

\subsection{Hypothesis testing}
First, for the theoretical and fundamental analysis,
we consider the scenario where the investigator aims to distinguish between two original information sources, $\rho_{\theta_0}$ and $\rho_{\theta_1}$ with $\theta_0 < \theta_1$
only from one data while this scenario is not practical. 
The performance in this hypothesis testing task is characterized by quantities like the hypothesis testing divergence $D_H^\alpha( \rho_{\theta_0}\|\rho_{\theta_1})$.

\begin{theorem}\label{THJ}
The maximum hypothesis testing divergence under the $(\eps,\delta)$-DP constraint is given by:
\begin{align}
\max_{(\rho,\sigma) \in \mathfrak{D}_{(\eps,\delta)}} D_H^\alpha( \rho_{\theta_1}\|\rho_{\theta_0}) 
= D_H^\alpha( \rho_{(\eps,\delta),\theta_1}\|\rho_{(\eps,\delta),\theta_0}).
\end{align}
This maximum is uniformly attained by the pair $(\rho_{(\eps,\delta)}, \sigma_{(\eps,\delta)})$.
\end{theorem}

\begin{proof}
It is sufficient to show
\begin{align}
\max_{(\rho,\sigma) \in \mathfrak{D}_{(\eps,\delta)}} 
\beta^c(\alpha|\rho_{\theta_1}\|\rho_{\theta_0}) 
= \beta^c(\alpha| \rho_{(\eps,\delta),\theta_1}\|\rho_{(\eps,\delta),\theta_0})\label{VB2}.
\end{align}
where $\beta^c(\alpha| \sigma_1\|\sigma_2) :=1-\beta(\alpha| \sigma_1\|\sigma_2) $.
The relation \eqref{theo_QLDP_approx_F_div_eq1} implies
\begin{align}
\max_{(\rho,\sigma) \in \mathfrak{D}_{(\eps,\delta)}} 
\beta^c(\alpha| \rho\|\sigma) 
= \beta^c(\alpha| \rho_{(\eps,\delta)}\|\sigma_{(\eps,\delta)})
\label{VB3}.
\end{align}
Hence, the following discussion will show \eqref{VB2} from \eqref{VB3}.

Assume that the states $\rho$ and $\sigma$ satisfy the $(\eps,\delta)$-DP condition, which implies
\begin{align}
\beta^c(\alpha| \rho\|\sigma) 
\le \beta^c(\alpha| \rho_{(\eps,\delta)}\|\sigma_{(\eps,\delta)})
\label{VB4}.
\end{align}
The quantity $\beta^c(\alpha| \rho_{\theta_1}\|\rho_{\theta_0}) $
is written as 
\begin{align}
\beta^c(\alpha| \rho_{\theta_1}\|\rho_{\theta_0}) 
=&    \max_{\substack{\alpha'\in [0,1]:\\
        \theta_1 \alpha'+(1-\theta_1) \beta^c(\alpha'| \rho\|\sigma)\le \alpha 
        }}\theta_0 \alpha'+(1-\theta_0) \beta^c(\alpha'| \rho\|\sigma)  .
  \end{align}
We choose an real umber $\alpha_0$ satisfying
  $\theta_1 \alpha_0+(1-\theta_1) \beta^c(\alpha_0| \rho\|\sigma)= \alpha$.
  Then, 
$\beta^c(\alpha| \rho_{\theta_1}\|\rho_{\theta_0}) 
=\theta_0 \alpha_0+(1-\theta_0) \beta^c(\alpha_0| \rho\|\sigma)
$ because 
$\theta_0 \alpha'+(1-\theta_0) \beta^c(\alpha'| \rho\|\sigma) $
and
$\theta_1 \alpha'+(1-\theta_1) \beta^c(\alpha'| \rho\|\sigma)$
are strictly monotonically increasing for $\alpha'$.
Also, we choose a real number $\alpha_1$ satisfying
  $\theta_1 \alpha_1+(1-\theta_1) \beta^c( \alpha_1|
  \rho_{(\eps,\delta)}\|\sigma_{(\eps,\delta)})= \alpha$.
Hence, the relation \eqref{VB4} implies that
$\alpha_1\le \alpha_0$.
Hence, 
\begin{align*}
&\beta^c(\alpha| \rho_{(\eps,\delta),\theta_1}\|\rho_{(\eps,\delta),\theta_0})
-\beta^c(\alpha| \rho_{\theta_1}\|\rho_{\theta_0}) \\
=&\theta_0 \alpha_1+(1-\theta_0) \beta^c(\alpha_1| \rho_{(\eps,\delta),\theta_1}\|\rho_{(\eps,\delta),\theta_0})
-
(\theta_0 \alpha_0+(1-\theta_0) \beta^c(\alpha_0| \rho\|\sigma)) \\
=&
(\theta_1-\theta_0) (\beta^c(\alpha_1| \rho_{(\eps,\delta),\theta_1}\|\rho_{(\eps,\delta),\theta_0})
-\alpha_1)
-
(\theta_1-\theta_0) (\beta^c(\alpha_0| \rho\|\sigma)
-\alpha_0) \\
=&
(\theta_1-\theta_0) (
\beta^c(\alpha_1| \rho_{(\eps,\delta),\theta_1}\|\rho_{(\eps,\delta),\theta_0})
-\beta^c(\alpha_0| \rho\|\sigma)
+\alpha_0-\alpha_1) \ge 0,
\end{align*}
which implies \eqref{VB2}. This completes the proof of Theorem \ref{THJ}.
\end{proof}

The combination of Theorems \ref{lemma_t_hyp_F_Div} and \ref{THJ} derives the following corollary.

\begin{corollary}\label{Cor12}
        For any $\alpha \in [0,1) \cup (1,+\infty)$, we have
\begin{align}
\max_{(\rho,\sigma) \in \mathfrak{D}_{(\eps,\delta)}} D_{\alpha}( \rho_{\theta_1}\|\rho_{\theta_0}) 
&= D_{\alpha}( \rho_{(\eps,\delta),\theta_1}\|\rho_{(\eps,\delta),\theta_0}),\\
\max_{(\rho,\sigma) \in \mathfrak{D}_{(\eps,\delta)}} D( \rho_{\theta_1}\|\rho_{\theta_0}) 
&= D( \rho_{(\eps,\delta),\theta_1}\|\rho_{(\eps,\delta),\theta_0}),\label{CHJ}\\
\max_{(\rho,\sigma) \in \mathfrak{D}_{(\eps,\delta)}} \check{D}_{\alpha}( \rho_{\theta_1}\|\rho_{\theta_0}) 
&= \check{D}_{\alpha}( \rho_{(\eps,\delta),\theta_1}\|\rho_{(\eps,\delta),\theta_0}),\label{HJ1}
\end{align}
This maximum is uniformly attained by the pair $(\rho_{(\eps,\delta)}, \sigma_{(\eps,\delta)})$.
\end{corollary}

Next, we discuss the hypothesis testing 
when the investigator has $n$ data.
This scenario focuses on the hypothesis testing divergence $D_H^\alpha( \rho_{\theta_1}^{\otimes n}\|\rho_{\theta_0}^{\otimes n})$.
We study the following quantity
\begin{align}
B_\epsilon(\rho_{\theta_1} \|\rho_{\theta_0})
&:=\lim_{n\to \infty}\frac{1}{n}D_H^\epsilon( \rho_{\theta_1}^{\otimes n}\|\rho_{\theta_0}^{\otimes n}),
\end{align}
for $\epsilon\in (0,1)$.
Since it is known that the above values are given 
by $D(\rho_{\theta_1}\|\rho_{\theta_0})$ as Stein's lemma \cite{Hiai1991,887855}, 
\cite[Chapter 3]{H2017QIT}, thus, from \eqref{CHJ} of Corollary \ref{Cor12} yields the following relation.
\begin{align}
\max_{(\rho,\sigma) \in \mathfrak{D}_{(\eps,\delta)}} B_\epsilon( \rho_{\theta_1}\|\rho_{\theta_0}) 
&= D( \rho_{(\eps,\delta),\theta_1}\|\rho_{(\eps,\delta),\theta_0}).
\end{align}
The above relations show that
the asymptotic behavior of the hypothesis testing problem 
is optimized
when respondents encode their information by using 
$\rho_{(\eps,\delta)}$ and $\sigma_{(\eps,\delta)}$.

\subsection{Parameter estimation}
Next, we discuss the parameter estimation of the parameter 
$\theta$ when the investigator has $n$ data.
The investigator's goal is to estimate the parameter $\theta$. 
The quality of this estimation is typically measured by the mean square error (MSE). 
For a single-parameter model, the MSE is lower-bounded by the Cramér-Rao bound, which is the inverse of the Fisher information. In the quantum setting, we use the Symmetric Logarithmic Derivative (SLD) Fisher information. The SLD, denoted $L_\theta$, is defined by the equation:
\begin{align}
\frac{1}{2} (L_\theta \rho_\theta + \rho_\theta L_\theta) = \frac{d \rho_\theta}{d\theta} = \rho-\sigma.
\end{align}
The SLD Fisher information, $J_\theta$, is then given by:
\begin{align}
J_\theta := \tr(L_\theta^2 \rho_\theta).
\end{align}
The MSE of any unbiased estimator for $\theta$ is lower-bounded by $1/J_\theta$. This bound is asymptotically achievable, for instance, via a two-step estimation process.

In the privatized scenario, the investigator is restricted to the pair of quantum states $(\rho,\sigma) \in \mathfrak{D}_{(\eps,\delta)}$.
Consequently, 
to optimize the estimation, it is natural to maximize the SLD Fisher information $J_\theta$ subject to the $(\eps,\delta)$-DP privacy constraint on the states $\rho$ and $\sigma$.
This leads to a constrained optimization problem,
\[
\max_{(\rho,\sigma) \in \mathfrak{D}_{(\eps,\delta)}} J_\theta,
\]
where the maximum is taken over all pairs of quantum states satisfying the privacy constraint. The optimal value quantifies the fundamental tradeoff between privacy and estimation accuracy in quantum settings.

{ Furthermore, the structure of the optimal $(\eps,\delta)$-DP pair, as characterized in Lemma~\ref{lemma_diff_priv_quant}, allows for explicit computation of the Fisher information and the corresponding Cram\'er-Rao bound. While foundational classical works, such as \cite[Theorem 18]{KOV15}, established the structural paradigm that optimal statistical utility is attained by extremal privacy mechanisms, they largely leave the exact bounds implicit. By leveraging our explicit extremal state pair, Theorem 4 goes further to provide a concrete, closed-form benchmark for the best possible estimation performance under quantum differential privacy constraints.}

\begin{theorem}\label{theo_fisher_info_bound}
The maximum SLD Fisher information achievable under the $(\eps,\delta)$-DP constraint is given by:
\begin{align}
\max_{(\rho,\sigma) \in \mathfrak{D}_{(\eps,\delta)}} J_\theta = \frac{\delta}{\theta(1-\theta)} + \frac{(1-\delta)(1-e^\eps)^2}{e^\eps + (1-e^\eps)^2 \theta(1-\theta)}.\label{XM3}
\end{align}
This maximum is uniformly attained by the pair $(\rho_{(\eps,\delta)}, \sigma_{(\eps,\delta)})$.
\end{theorem}

\begin{proof}
Let $\rho_{(\eps,\delta),\theta} := \theta \rho_{(\eps,\delta)} + (1-\theta)\sigma_{(\eps,\delta)}$. 
First, we denote the SLD Fisher information for 
this specific states $\rho_{(\eps,\delta),\theta}$ by $J_{(\eps,\delta),\theta}$, and 
we will show the relation
\begin{align}
\max_{(\rho,\sigma) \in \mathfrak{D}_{(\eps,\delta)}} J_\theta = 
J_{(\eps,\delta),\theta}.\label{XM1}
\end{align}
{ 
We recall the standard relation between the SLD Fisher information \cite[Theorem 6.3]{H2017QIT},
the Bures distance, and the quantum affinity.  For two states
$\rho$ and $\sigma$, let
\(
F(\rho,\sigma):=\tr\left|\sqrt{\rho}\sqrt{\sigma}\right|
\)
be the fidelity and
\(
b(\rho,\sigma):=\sqrt{2(1-F(\rho,\sigma))}
\)
be the Bures distance.  
As presented in \cite[(6.33)]{H2017QIT}, \cite[Theorem 6.3]{Hayashi_2002},
if the SLD $L_\theta$ exists, the infinitesimal
form of the Bures metric gives 
$\frac{1}{4}J_\theta
=
\lim_{h\to0}
\frac{b^2(\rho_\theta,\rho_{\theta+h})}{h^2}$.
Equivalently, since
\(
b^2(\rho_\theta,\rho_{\theta+h})
=
2\left(1-F(\rho_\theta,\rho_{\theta+h})\right),
\)
we have
\begin{equation}
J_\theta
=
\lim_{h\to0}
\frac{8\left(1-F(\rho_\theta,\rho_{\theta+h})\right)}{h^2}.
\label{eq:SLD_via_fidelity_linear}
\end{equation}

We now rewrite this expression in logarithmic form.  Putting
\(
F_h:=F(\rho_\theta,\rho_{\theta+h})
=
\tr\left|\sqrt{\rho_\theta}\sqrt{\rho_{\theta+h}}\right|
\),
from \eqref{eq:SLD_via_fidelity_linear}, we have
$1-F_h=O(h^2)$ as $h\to0$.  Since
\(
-\log(1-x)=x+O(x^2)\) as \(x\to0\)
with $x=1-F_h$, it follows that
\(
-\log F_h
=
1-F_h+O\left((1-F_h)^2\right)
=
1-F_h+O(h^4).
\)
Therefore,
\(
\frac{-\log F_h-(1-F_h)}{h^2}\to0,
\)
and hence \eqref{eq:SLD_via_fidelity_linear} is equivalent to
$J_\theta
=
\lim_{h\to0}
\frac{-8\log F(\rho_\theta,\rho_{\theta+h})}{h^2}$.

Finally, by the order-$1/2$ measured Rényi divergence formula
\(
\check D_{1/2}(\rho\|\sigma)
=
-2\log \tr\left|\sqrt{\rho}\sqrt{\sigma}\right|
\) (Fact \ref{fact_measured_renyi_half}),
we obtain
\begin{equation}
J_\theta
=
\lim_{h\to0}
\frac{4\check D_{1/2}(\rho_\theta\|\rho_{\theta+h})}{h^2}.
\label{BNA}
\end{equation}
}
Therefore, the relation \eqref{HJ1} with the limit like \eqref{BNA}
implies \eqref{XM1}.

Next, we compute the SLD Fisher information $J_{(\eps,\delta),\theta}$.
The state $\rho_{(\eps,\delta),\theta}$ is diagonal in the standard basis $\{\ket{00}, \ket{01}, \ket{10}, \ket{11}\}$. Its eigenvalues are:
\begin{align}
p_{00} &= \theta\delta, \nn\\
p_{01} &= \theta \frac{(1-\delta)e^\eps}{1+e^\eps} + (1-\theta)\frac{1-\delta}{1+e^\eps} = \frac{1-\delta}{1+e^\eps}(\theta e^\eps + 1-\theta), \nn\\
p_{10} &= \theta \frac{1-\delta}{1+e^\eps} + (1-\theta)\frac{(1-\delta)e^\eps}{1+e^\eps} = \frac{1-\delta}{1+e^\eps}(\theta + (1-\theta)e^\eps), \nn\\
p_{11} &= (1-\theta)\delta.\nn
\end{align}

The derivative $\frac{d\rho_{(\eps,\delta),\theta}}{d\theta} = \rho_{(\eps,\delta)} - \sigma_{(\eps,\delta)}$ is also diagonal, with eigenvalues:
\begin{align}
d_{00} &= \delta, \nn\\
d_{01} &= \frac{(1-\delta)(e^\eps-1)}{1+e^\eps}, \nn\\
d_{10} &= -\frac{(1-\delta)(e^\eps-1)}{1+e^\eps}, \nn\\
d_{11} &= -\delta.\nn
\end{align}

For diagonal states, the SLD Fisher information is given by the sum of the classical Fisher informations for the eigenvalues: $J_\theta = \sum_i \frac{(\frac{d p_i}{d\theta})^2}{p_i}$. In this case, $\frac{d p_i}{d\theta} = d_i$. Thus, we have,
\begin{align}
J_{(\eps,\delta),\theta} &= \frac{d_{00}^2}{p_{00}} + \frac{d_{01}^2}{p_{01}} + \frac{d_{10}^2}{p_{10}} + \frac{d_{11}^2}{p_{11}} \nn\\
&= \frac{\delta^2}{\theta\delta} + \frac{\left(\frac{(1-\delta)(e^\eps-1)}{1+e^\eps}\right)^2}{\frac{1-\delta}{1+e^\eps}(\theta e^\eps + 1-\theta)} + \frac{\left(-\frac{(1-\delta)(e^\eps-1)}{1+e^\eps}\right)^2}{\frac{1-\delta}{1+e^\eps}(\theta + (1-\theta)e^\eps)} + \frac{(-\delta)^2}{(1-\theta)\delta} \nn\\
&= \frac{\delta}{\theta} + \frac{\delta}{1-\theta} + \frac{(1-\delta)(e^\eps-1)^2}{1+e^\eps} \left( \frac{1}{\theta e^\eps + 1-\theta} + \frac{1}{\theta + (1-\theta)e^\eps} \right) \nn\\
&= \frac{\delta}{\theta(1-\theta)} + \frac{(1-\delta)(e^\eps-1)^2}{1+e^\eps} \left( \frac{(\theta + (1-\theta)e^\eps) + (\theta e^\eps + 1-\theta)}{(\theta e^\eps + 1-\theta)(\theta + (1-\theta)e^\eps)} \right) \nn\\
&= \frac{\delta}{\theta(1-\theta)} + \frac{(1-\delta)(e^\eps-1)^2}{1+e^\eps} \left( \frac{1+e^\eps}{\theta^2 e^\eps + \theta(1-\theta) + \theta(1-\theta)(e^\eps)^2 + (1-\theta)^2 e^\eps} \right) \nn\\
&= \frac{\delta}{\theta(1-\theta)} + (1-\delta)(e^\eps-1)^2 \left( \frac{1}{\theta^2 e^\eps + \theta(1-\theta)(1+(e^\eps)^2) + (1-\theta)^2 e^\eps} \right) \nn\\
&= \frac{\delta}{\theta(1-\theta)} + (1-\delta)(e^\eps-1)^2 \left( \frac{1}{e^\eps(\theta^2+(1-\theta)^2) + \theta(1-\theta)(1+(e^\eps)^2)} \right) \nn\\
&= \frac{\delta}{\theta(1-\theta)} + \frac{(1-\delta)(1-e^\eps)^2}{e^\eps + (1-e^\eps)^2 \theta(1-\theta)}.
\label{XM2}
\end{align}
The combination of 
\eqref{XM1} and \eqref{XM2} yields the desired relation 
\eqref{XM3}.
 \end{proof}

\section{Upper-Bounds on Relative Entropy between the outputs of $(\eps,\delta)$-LDP classical channels via the integral representation}\label{sec_LDP_channel}
\subsection{Formulation of 
$(\eps,\delta)$-LDP classical channels}
In this section, we show some other applications of the dominance of quantum informativeness as discussed in Section \ref{sec_diff_priv_quant}. A primary question which is often studied in the context of $(\eps,\delta)$-DP mechanisms is to determine upper-bounds on the divergence with respect to the output induced by the mechanisms. To do this, we can easily invoke Corollary \ref{theo_QLDP_approx_F_div_bound} and easily get an upper-bound on the divergence between the states induced at the output of the mechanism in terms $D_f(\rho_{(\eps,\delta)}\|\sigma_{(\eps,\delta)})$, where $D_f$ is some $f$-divergence as defined in Definition \ref{def_f_div_integral_rep}. However, note that $\supp(\rho_{(\eps,\delta)}) \nsubseteq \supp(\sigma_{(\eps,\delta)})$ and therefore there will be many $f$-divergences for which $D_f(\rho_{(\eps,\delta)}\| \sigma_{(\eps,\delta)})$ is not well defined even for the case when the divergence $D_f(\cdot\|\cdot)$ between the output distributions induced by the $(\eps,\delta)$-DP mechanism is well defined. Hence, naively using Corollary \ref{theo_QLDP_approx_F_div_bound}, will not lead to any meaningful upper-bound on $D_f(\cdot\|\cdot)$ between the output distributions induced by an $(\eps,\delta)$-DP mechanism. 

In this section, we will obtain a meaningful bound on the relative entropy $D(\cdot\|\cdot)$ between the output distributions induced by an $(\eps,\delta)$-LDP channel (defined below) for the case when at least the support of one of the output distributions is a subset of the support of the other distribution. We will accomplish this with the help of the integral representation of the relative entropy mentioned in Fact \ref{fact_classical_KL_div_integral_rep}. This integral representation uses the contraction coefficient of $(\eps,\delta)$-LDP channels for the Hockey-Stick divergence. 

Towards this, we first formulate the channel-based setup of $(\eps,\delta)$-LDP in both classical and quantum settings. 
In the classical case, a channel from the system $\cX$ to the system $\cY$
is given as a transition matrix $P_{Y|X}$.
Assume that a respondent generates private data in $\cX$ and converts it to 
a distribution on $\cY$ via the channel $P_{Y|X}$
and an investigator can access only the system $\cY$.
In this case,
a $(\eps,\delta)$-LDP channel is formulated 
as a generalization of $(\eps,\delta)$-LDP in the following way.

\begin{definition}\label{LDP_channel_class2}
    A classical channel $P_{Y|X}$
    is defined to be $(\eps,\delta)$-LDP (locally differentially private) for some fixed $\eps \geq 0$ and $\delta \in [0,1]$,
    if any pair of $x\neq x' \in \cX$ satisfies 
        \begin{equation}
        P_{Y|X}(S|x) \leq e^\eps P_{Y|X}(S|x')  + \delta, \label{LDP_channel_class_eq2}
    \end{equation}
for every subset $S \subseteq \cY$.
Further, for $\delta = 0$, we denote $\cK$ to be a pure $\eps$-LDP or just $\eps$-LDP channel.
\end{definition}

Now, we write the channel by using the map $\cK$ from 
a distribution on $\cX$ to a distribution on $\cY$ 
as $\cK(P)(y):= \sum_{x \in \cX} P_{Y|X}(y|x)P(x)$.
This definition is rewritten as follows.

\begin{lemma}\label{LDP_channel_class}
    A classical channel $\cK $ is $(\eps,\delta)$-LDP,
if and only if 
    any pair of distributions $P,Q\in \cP(\cX)$ satisfies 
    the relation
    \begin{equation}
        \cK(P)(S) \leq e^\eps \cK(Q)(S) + \delta, \label{LDP_channel_class_eq}
    \end{equation}
for every subset $S \subseteq \cY$.
\end{lemma}

\begin{proof}
When the condition \eqref{LDP_channel_class_eq} holds, 
the classical channel $\cK $ is $(\eps,\delta)$-LDP
by considering the case when $P$ and $Q$ are delta distributions.

Assume that 
the classical channel $\cK $ is $(\eps,\delta)$-LDP.
Given $x \in \cX$, \eqref{LDP_channel_class_eq2} implies 
$
P_{Y|X}(S|x) \leq e^\eps P_{Y|X}(S|x')  + \delta
$ for any $x'\in \cX$.
Thus, 
$P_{Y|X}(S|x) =
\sum_{x'\in \cX}Q(x') P_{Y|X}(S|x) \leq 
\sum_{x'\in \cX}Q(x')(e^\eps P_{Y|X}(S|x')  + \delta)
= e^\eps \cK(Q)(S)  + \delta$.
Then, we have
$\cK(P)(S)
=\sum_{x\in \cX}P(x) P_{Y|X}(S|x) 
\leq 
= \sum_{x\in \cX}P(x)
(e^\eps \cK(Q)(S)  + \delta)
=e^\eps \cK(Q)(S)  + \delta$, which shows \eqref{LDP_channel_class_eq}.
\end{proof}

We note here that the above definition of $(\eps,\delta)$-LDP channel is relaxed version of the standard definition of $(\eps,\delta)$-LDP channel \cite{Asoodeh24}, where the privacy condition is required to hold for every pair of input symbols $x,x' \in \cX$. 
However, in the upcoming discussion, we will see that the above definition is sufficient to obtain meaningful upper-bounds on the relative entropy (when it is well defined) between the output distributions induced by $(\eps,\delta)$-LDP channels.
Now, as a quantum generalization of \eqref{LDP_channel_class_eq}, 
we have the following definition.

\begin{definition}[{\cite{AK25}}]\label{QLDP_def}
    A CP-TP map $\cN : \cH_A \to \cH_B$ is defined to be quantum $(\eps,\delta)$-LDP (locally differentially private) for some fixed $\eps \geq 0$ and $\delta \in [0,1]$,
    if for all pairs $\rho,\sigma \in \cD(\cH_A)$ and every POVM measurement $0 \preceq \Lambda \preceq \bbI$, the following holds,
    \begin{equation}
            \tr[\Lambda \cN(\rho)] \leq e^{\eps} \tr[\Lambda \cN(\sigma)]+ \delta.\label{QLDP_def_eq}
    \end{equation}  
    Further, for $\delta = 0$, we denote $\cN$ to be a pure quantum $\eps$-LDP or just quantum $\eps$-LDP CP-TP map.
\end{definition}

\begin{remark}\label{remark_relax_QLDP}
Observe that Definition \ref{QLDP_def} appears to be very strict in the sense that for a CP-TP map to be $(\eps,\delta)$-DP, it should behave differentially private for every pair of quantum states. In \cite{Nuradha25}, the authors give an example of one such map in terms of a measurement channel composed by a depolarizing channel. 
\end{remark}

In the subsections below, we define the contraction coefficient of classical and quantum $(\eps,\delta)$-LDP channels with respect to general divergences and obtain almost matching upper and lower bounds on the contraction coefficient of the Hockey-Stick divergence for  $(\eps,\delta)$-LDP mechanism.

\subsection{Contraction Coefficient for Divergences under Private Classical/Quantum Learning Algorithms}

For any classical divergence $\mathbb{D}$, from Definition \ref{def_data_proc_div_class}, we know that for any pair of probability distributions $P_1,P_2 \in \cP(\cX)$ (where $\cX$ is some finite set) and a classical channel $\cK : \cX \to \cY$ (where $\cY$ is another finite set), $ \bbD(\cK(P_1)\| \cK(P_2)) \leq \bbD(P_1\|P_2)$, where for each $i=1,2$, $\cK(P_i)$ is the marginal output distribution of $\cK$ corresponding to $P_i$.
However, this inequality is not always strict. 
Thus, we are interested in how much the channel $\cK$ shrinks the KL divergence for the privacy of the respondent. 
To clarify this, we study the largest constant 
$\eta^{(c), \cK}_{\bbD}$ satisfying the condition
$\bbD(\cK(P_1)\| \cK(P_2)) \leq \eta^{(c), \cK}_{\bbD} \bbD(P_1\|P_2)$ for all pairs of distributions  $P_1,P_2 \in \cP(\cX)$. 
The quantity $\eta^{(c), \cK}_{\bbD}$ is called the contraction coefficient of the channel $\cK$ with respect to the divergence $\bbD$.

In particular, if we consider the set of all classical $(\eps,\delta)$-LDP channels (see Definition \ref{LDP_channel_class}) for some $\eps\geq 0$ and $ \delta \in [0,1)$, denoted by $\mathfrak{Q}_{\eps,\delta}$, then it is interesting to study the worst contraction coefficient of any channel $\cK$ with respect to the classical Hockey-Stick divergence $E_{\gamma}(\cdot\|\cdot)$ (see Definition \ref{def_class_hockey_stick}). This quantity is defined as follows,
\begin{equation}
    \eta^{(c),\eps,\delta}_{E_{\gamma}} := 
    \sup_{\cK \in \mathfrak{Q}_{\eps,\delta}} 
    \eta^{(c),\cK}_{E_{\gamma}} \triangleq  \sup_{\substack{\cK \in \mathfrak{Q}_{\eps,\delta}\\ P_1,P_2 \in \cP(\cX):\\ E_{\gamma}(P_1\|P_2) \neq 0}}
    \frac{E_{\gamma}(\cK(P_1)\| \cK(P_2))}{E_{\gamma}(P_1\|P_2)},\label{contrac_class_hockey_stick_def_eq}
\end{equation}
where $\gamma \geq 1$ is some fixed constant.

Similarly, in the quantum setting, for any divergence $\varmathbb{D}$, we want to study the worst constant $\eta^{\cN}_{F}$ such that $\varmathbb{D}(\cN(\rho),\cN(\sigma))$ $ \leq \eta^{\cN}_{F} \varmathbb{D}(\rho,\sigma)$ for all pairs  $\rho,\sigma \in \cD(\cH_A)$ such that $\varmathbb{D}(\rho,\sigma) \neq 0$. The quantity $\eta^{\cN}_{F}$ is called the contraction coefficient of the CP-TP map $\cN$ with respect to the divergence $\varmathbb{D}$.
In particular, for any CP-TP map 
$\cN : \cD(\cH_A) \to \cD(\cH_B)$, it is interesting to study the contraction coefficient of $\cN$ with respect to the quantum Hockey-Stick divergence $E_{\gamma}(\cdot\|\cdot)$ (see Definition \ref{f_sym_hyp_func}) for some fixed $\gamma \geq 1$. This quantity is defined as follows,
\begin{equation}
    \eta^{\cN}_{E_{\gamma}} := \sup_{\substack{\rho,\sigma \in \cD(\cH_A):\\ E_{\gamma} (\rho\|\sigma) \neq 0}}\frac{E_{\gamma}(\cN(\rho)\|\cN(\sigma))}{E_{\gamma}(\rho,\sigma)}.\label{contrac_hockey_stick_def_eq}
\end{equation}

Hirche et al. in \cite{HRF23} studied the following bound on $\eta^{\cN}_{E_{\gamma}}$ for any CP-TP some fixed $\gamma \geq 1$.

\begin{proposition}[{\cite[Lemma \RNum{2}.4]{HRF23}}]\label{prop_contrac_hockey_stick}
    For any $\gamma \geq 1$ and a CP-TP map 
    $\cN : \cD(\cH_A) \to \cD(\cH_B)$, the contraction coefficient $\eta^{\cN}_{E_{\gamma}}$ with respect to the quantum Hockey-Stick divergence $E_{\gamma}(\cdot\|\cdot)$ (see Definition \ref{f_sym_hyp_func}) satisfies the following,
    \begin{equation}
        1 - \gamma \left(1 - \eta^{\cN}_{E_1}\right) \leq \eta^{\cN}_{E_{\gamma}} \leq  \eta^{\cN}_{E_1},\label{prop_contrac_hockey_stick_eq}
    \end{equation}
    where $\eta^{\cN}_{E_1}$ is the contraction coefficient of quantum Hockey-Stick divergence of order 1, i.e., trace distance with respect to $\cN$.
\end{proposition}

From the context of Private CP-TP maps, the contraction coefficient of quantum Hockey-Stick divergence under $(\eps,\delta)$-LDP quantum mechanisms is defined as follows,

\begin{definition}[\cite{Nuradha25}]\label{def_contr_coeff_F}
    For any $\gamma \geq 1$, the contraction coefficient of  $(\eps,\delta)$-LDP quantum CP-TP maps(see Definition \ref{QLDP_def}) with respect to $E_{\gamma}$ is defined as follows,
    \begin{equation}
        \eta^{\eps,\delta}_{E_{\gamma}} := \sup_{\cN \in \mathfrak{P}_{\eps,\delta}} \eta^{\cN}_{E_{\gamma}} \triangleq  \sup_{\substack{\cN \in \mathfrak{P}_{\eps,\delta}\\ \rho,\sigma \in \cD(\cH_A):\\ E_{\gamma}(\rho\|\sigma) \neq 0}}\frac{E_{\gamma}(\cN(\rho)\|\cN(\sigma))}{E_{\gamma}(\rho,\sigma)},\label{contrac_F_def_eq}
    \end{equation}
    where $\mathfrak{P}_{\eps,\delta}$ is the set of all quantum $(\eps,\delta)$-LDP CP-TP maps (see Definition \ref{QLDP_def}) for some $\eps\geq 0$ and $ \delta \in [0,1)$.
\end{definition}

Further, in \cite{Nuradha25}, the authors studied the contraction coefficient for the trace distance $\eta^{\eps,\delta}_{E_1}$ under quantum $(\eps,\delta)$-LDP CP-TP maps where $\eps \geq 0$ and $\delta \in [0,1]$. 
This is given in the proposition below.

\begin{proposition}[{\cite[Theorem 5]{Nuradha25}}]\label{prop_contrac_trac_dis}
For any $\eps \geq 0$ and $\delta \in [0,1)$, the contraction coefficient $\eta^{\eps,\delta}_{E_1}$ with respect to the trace distance satisfies
\begin{equation}
    \eta^{\eps,\delta}_{E_1} = \frac{(e^\eps - 1 + 2\delta)}{e^\eps + 1}.\label{prop_contrac_trac_dis_eq}
\end{equation}
\end{proposition}

{
In the following lemma, we establish bounds for the contraction coefficient of the quantum Hockey-Stick divergence $E_{\gamma}(\cdot\|\cdot)$ for any $\gamma \geq 1$ under $(\eps,\delta)$-LDP CP-TP maps. Note that when $\gamma = 1$, these bounds collapse to a strict equality, remaining perfectly consistent with the trace distance result in Proposition \ref{prop_contrac_trac_dis}.}

\begin{lemma}\label{lemma_contrac_hockey_stick}
    For any $\gamma \geq 1, \eps\geq 0$ and $\delta \in [0,1)$, the contraction coefficient $\eta_{E_{\gamma}}^{\eps,\delta}$ with respect to the quantum Hockey-Stick divergence $E_{\gamma} (\cdot\|\cdot)$ (see Definition \ref{f_sym_hyp_func}) satisfies the following,
    \begin{equation}
        \eta^{\eps,\delta}_{E_1} + {\frac{2(1-\delta)(1-\gamma)}{e^\eps+1}} \leq \eta_{E_{\gamma}}^{\eps,\delta} \leq \begin{cases}
            \eta^{\eps,\delta}_{E_1} + \frac{(1-\delta)(1-\gamma)}{e^\eps+1}, & ~\mbox{if}~ \gamma \in [1,e^\eps],\\
            \delta, & ~\mbox{if}~ \gamma > e^\eps.
        \end{cases} \label{lemma_contrac_hockey_stick_eq}
    \end{equation}
\end{lemma}

\begin{proof}
\hfill\\
\quad\textbf{(1) Lower-bound:} Consider the following series of inequalities 
\begin{align*}
\eta_{E_{\gamma}}^{\eps,\delta} &\overset{(a)}\geq 1-\gamma(1-\eta^{\eps,\delta}_{E_1})\\
&\overset{(b)}= 1 - \gamma\left(1 - \frac{e^\eps - 1 + 2\delta}{e^\eps+1}\right)\\
&= \frac{e^\eps-1+2\delta + 2 - 2\delta - 2\gamma(1-\delta)}{e^\eps+1}\\
&= \eta^{\eps,\delta}_{E_1} + \frac{2(1-\delta) - 2\gamma(1-\delta)}{e^\eps+1}\\
&= \eta^{\eps,\delta}_{E_1} + \frac{2(1-\delta)(1-\gamma)}{e^\eps+1},
\end{align*}
where $(a)$ follows from Proposition \ref{prop_contrac_hockey_stick}, $(b)$ follows from Proposition \ref{prop_contrac_trac_dis}, $(c)$ follows because $\gamma \geq 1.$

\textbf{(2) Upper-bound:}
\begin{enumerate}[label=(\roman*)]

\item For the case when  $\gamma \in [1,e^\eps]$. Consider the following series of inequalities,
    {\allowdisplaybreaks\begin{align}
        \eta_{E_{\gamma}}^{\eps,\delta} 
        &\overset{(a)}{=} \sup_{\substack{\cN \in \mathfrak{P}_{\eps,\delta},\\ \ket{\psi}\perp\ket{\phi}}} E_{\gamma}(\cN(\ketbra{\psi})\|\cN(\ketbra{\phi}))\nn\\
        &\overset{(b)}{\leq}  E_{\gamma}(\rho_{(\eps,\delta)}\|\sigma_{(\eps,\delta)})\nn\\
    &\overset{(c)}=\frac{e^\eps -\gamma + \delta(\gamma+1)}{e^\eps+1},\label{proof_lemma_contrac_hockey_stick_eq1}\\
   &= \frac{e^\eps - 1 + 2\delta}{e^\eps + 1} + \frac{(1-\delta)(1-\gamma)}{e^\eps + 1}\nn\\
   &= \eta^{\eps,\delta}_{E_1} + \frac{(1-\delta)(1-\gamma)}{e^\eps+ 1},\label{lemma_contrac_hockey_stick_eq1}
        \end{align}}
where in $(a)$, $\ket{\psi},\ket{\phi} \in \cH_A$ are two orthogonal quantum states and the { equality} follows from \cite[Theorem \RNum{2}.2]{HRF23}, 
$(b)$ follows from Corollary \ref{theo_QLDP_approx_F_div_bound} and $(c)$ follows from Claim \ref{claim_hockey_stick_bound} below.

\begin{claim}\label{claim_hockey_stick_bound}
    For any $\gamma \geq 1, \eps\geq 0$ and $\delta \in [0,1)$, the quantum Hockey-Stick divergence $E_{\gamma} (\rho_{(\eps,\delta)}\|\sigma_{(\eps,\delta)})$ satisfies the following,
    \begin{equation}
        E_{\gamma}(\rho_{(\eps,\delta)}\|\sigma_{(\eps,\delta)}) = \begin{cases}
            \frac{e^\eps -\gamma + \delta(\gamma+1)}{e^\eps+1}, & ~\mbox{if}~ \gamma \in [1,e^\eps],\\
            \delta, & ~\mbox{if}~ \gamma > e^\eps.
        \end{cases} \label{claim_hockey_stick_bound_eq}
    \end{equation}
    where $(\rho_{(\eps,\delta)},\sigma_{(\eps,\delta)})$ is the weakest (most informative) $(\eps,\delta)$-DP pair of quantum states as mentioned in \cref{eq_rho_simp,eq_rho'_simp}.
\end{claim}
\begin{proof}
    See Appendix \ref{proof_claim_hockey_stick_bound} for the proof.
\end{proof}

\vspace{5pt}
 
\item For the case when  $\gamma > e^\eps$. The RHS of \eqref{proof_lemma_contrac_hockey_stick_eq1}, will be replaced by $\delta$ using Claim \ref{claim_hockey_stick_bound}. This completes the proof of Lemma \ref{lemma_contrac_hockey_stick}.
\end{enumerate}

\end{proof}
\begin{remark}
  { In Lemma \ref{lemma_contrac_hockey_stick}, for the case when $\gamma \in [1,e^\eps]$, the difference between the upper and lower bound is $\frac{(1-\delta)(\gamma-1)}{e^\eps+1} \leq \tanh(\frac{\eps}{2}).$} For small values of $\eps, \tanh(\frac{\eps}{2}) \leq \mathcal{O}(\eps).$ Thus, making these bounds almost tight in the small $\eps$ regime. Also, it trivially follows that for $\gamma=1$ both the upper and lower bound coincide and are equal to $\eta^{\eps,\delta}_{E_1}$. 

Further, the upper-bound obtained in Lemma \ref{lemma_contrac_hockey_stick} is tighter in comparison to the upper-bound obtained in Proposition \ref{prop_contrac_hockey_stick} by a subtractive factor of $\frac{(1-\delta)(\gamma-1)}{e^\eps+1}.$  
\end{remark}

In a classical scenario, it is trivial to observe that the privatized contraction coefficient $\eta^{(c), \eps,\delta}_{E_\gamma}$ with respect to classical Hockey-Stick divergence has the same upper and lower bounds as $\eta_{E_{\gamma}}^{\eps,\delta}$ mentioned in Lemma \ref{lemma_contrac_hockey_stick} i.e. we have the following,

\begin{corollary}\label{corr_contrac_classical_hockey_stick}
    For any $\eps\geq 0$ and $\delta \in [0,1)$, the contraction coefficient $\eta_{E_{\gamma}}^{(c), \eps,\delta}$ with respect to the classical Hockey-Stick divergence (see Definition \ref{def_class_hockey_stick}) satisfies the following,
    \begin{equation}
       \eta_{E_{\gamma}}^{(c),\eps,\delta} \leq \begin{cases}
            \frac{e^\eps -\gamma + \delta(\gamma+1)}{e^\eps+1}, & ~\mbox{if}~ \gamma \in [1,e^\eps],\\
            \delta, & ~\mbox{if}~ \gamma > e^\eps.
        \end{cases} \label{corr_contrac_classical_hockey_eq}
    \end{equation}
\end{corollary}

In the quantum setting, Nuradha et al in \cite[Proposition 3]{Nuradha25} obtaied  an upper-bound on the quantum relative entropy under $\eps$-LDP CP-TP maps via an integral representation of quantum relative entropy (\cite[Corollary 2.3]{Hirche_2024}) in terms of quantum Hockey-Stick divergence. 
Moreover, the upper-bound obtained in \cite[Proposition 3]{Nuradha25} is tighter than the upper-bound obtained in $(i)$ of Corollary \ref{corr_QLDP_Divs_bound}.

However, for $\delta > 0$, it is not clear if one can obtain a similar upper-bound on quantum relative entropy under $(\eps,\delta)$-LDP CP-TP maps via an integral representation of quantum relative entropy in terms of quantum Hockey-Stick divergence. 
This is because, unlike the $\eps$-LDP channels, for $(\eps,\delta)$-LDP channels (with $\delta > 0$), the contraction coefficient of quantum Hockey-Stick divergence never becomes $0$ for any $\gamma \geq 1$ (see \eqref{lemma_contrac_hockey_stick_eq}).
This makes the integral expression for quantum relative entropy diverge to infinity.

In the subsection below, we resolve this issue in classical setting, by coming up with a technique which we call truncation. This technique allows us to distill pure DP from non-pure DP.

\subsection{Upper-bound on the relative entropy of $(\eps,\delta)$-LDP classical channels via the integral representation}
In this subsection,  we obtain a tight upper-bound on the relative entropy of the output distributions of any $(\eps,\delta)$-LDP classical channel (Markov kernel) via the integral representation of relative entropy in terms of Hockey-Stick divergence. To obtain this, we define a two-sided truncation of a $(\eps,\delta)$-DP pair of distributions $(P,Q)$ in the definition below.

\begin{definition}[Truncated Pair of Distributions]\label{projpair}
   Consider a pair of distribution $(P,Q) \in \cP(\cX)$ (where $\cX$ is any arbitrary finite set) satisfy $(\eps,\delta)$-DP. Then, a pair of distributions $(\tilde{P},\tilde{Q})$ is called the truncated pair with respect to $(P,Q)$ if it has the following form,
   \begin{align}
        \tilde{P}(x) := \frac{P'(x)}{\sum_{x' \in \cX} P'(x)},\label{projpair_eq1} \\
        \tilde{Q}(x) := \frac{Q'(x)}{\sum_{x' \in \cX} Q'(x)},\label{projpair_eq2} 
    \end{align}
    where, $P'(x) := \min\left\{P(x),e^\eps Q(x)\right\}$ and $Q'(x) := \min\left\{Q(x),e^\eps P(x)\right\}$.
\end{definition}

In the following lemma, we show that the truncated pair $(\tilde{P},\tilde{Q})$ corresponding to any $(\eps,\delta)$-DP pair of probability distributions $(P,Q)$, satisfies the following properties.

\begin{lemma}\label{lemma_trunc_prop}
    For any pair $(P,Q)$ of $(\eps,\delta)$-DP distributions, its truncated pair $(\tilde{P},\tilde{Q})$ satisfies the following,
\begin{enumerate}[label=(\roman*)]
    \item $\|\tilde{P}-P\|_1 \leq 2\delta$ and $\|\tilde{Q}-Q\|_1 \leq 2\delta.$
    \item $(\tilde{P},\tilde{Q})$ satisfies pure $\left(\eps+\log\frac{1}{(1-\delta)}\right)$-DP.
    \item if $\supp(P) \subseteq \supp(Q)$, then $\abs{D(P\|Q) -D(\tilde{P}\|\tilde{Q})} \leq 2 \delta \left( \eps + \log \frac{1}{1-\delta} +  \frac{2}{m} \right)$.
\end{enumerate}
 where $m = \min_{x \in \supp(\tilde{P})}\left\{\min\{\tilde{P}(x), P(x)\},\min\{\tilde{Q}(x), Q(x)\}\right\}.$
\end{lemma}

\begin{proof}
    See Appendix \ref{proof_lemma_trunc_prop} for the proof.
\end{proof}

Using Lemma \ref{lemma_trunc_prop} and the integral representation of relative entropy in terms of Hockey-Stick divergence (Fact \ref{fact_classical_KL_div_integral_rep}), we now obtain an upper-bound on the classical relative entropy of the output distributions of any $(\eps,\delta)$-LDP classical channel (Markov kernel) in Theorem \ref{theorem_classical_kl_bound} below.

\begin{theorem}\label{theorem_classical_kl_bound}
   Let $\cK : \cX \to \cY$ be a $(\eps,\delta)$-LDP classical channel (Markov kernel). Further, for any pair of probability distributions $P_X$ and $Q_X$ over $\cX$, let $\cK(P_X)$ and $\cK(Q_X)$ be their respective output distribution with respect to $\cK$ such that $\supp(\cK(P_X))\subseteq \supp(\cK(Q_X))$. Then, 
   \begin{align*}
 D(\cK(P_X)\| \cK(Q_X))
    &\leq \frac{1}{2}\norm{P_X -Q_X}{1} \left(\eps \tanh\left(\frac{\eps}{2}\right)+ \delta\left(\frac{2\eps}{e^\eps+1} + \frac{e^\eps-1}{e^\eps} + \log\frac{1}{1-\delta} + \frac{\delta}{e^\eps}\right)\right)\nn\\
    &\quad + \delta \left(\frac{e^\eps}{1-\delta} + 2\log\frac{e^\eps}{1-\delta}  - \frac{1-\delta}{e^\eps} + 2 \left( \eps + \log \frac{1}{(1-\delta)} +  \frac{2}{m} \right)\right),
    \end{align*}
    where $m = \min_{y \in \supp( \cK(P_X) )}
    \left\{\min\{\Tilde{P}_Y(y), \cK(P_X)(y)\},
    \min\{\Tilde{Q}_Y(y), \cK(Q_X)(y)\}\right\}$ 
    and $(\Tilde{P}_Y,\Tilde{Q}_Y)$ is 
the truncated  pair with respect to $(\cK(P_X), \cK(Q_X))$.
\end{theorem}

\begin{proof}
See Appendix \ref{proof_theorem_classical_kl_bound} for the proof.
\end{proof}

{ While the bound in Theorem \ref{theorem_classical_kl_bound} appears algebraically complex, its intuition becomes clear in standard operational regimes. This is illustrated in the example below. 

\begin{example}[Asymptotic Behavior and Uniform Noise Floor]
    Consider an approximate $(\eps,\delta)$-LDP mechanism where the output space $\cY$ is finite, and the mechanism includes a uniform noise floor such that the minimum mass is bounded by $m = \Omega(\delta)$. For small $\eps \to 0$ and small $\delta \to 0$, expanding the terms in Theorem \ref{theorem_classical_kl_bound} yields the asymptotic scaling:
    \begin{equation*}
        D(\cK(P_X)\| \cK(Q_X)) \leq \cO\left(\eps^2 \norm{P_X -Q_X}{1}\right) + \cO\left(\delta \log \frac{1}{\delta}\right).
    \end{equation*}
    This example highlights where the bound performs well: it gracefully separates the quadratic contraction provided by $\eps$-DP from the unavoidable logarithmic penalty $\delta \log(1/\delta)$ incurred by the $\delta$-leakage on the tails of the KL divergence.
\end{example}}

Alternatively, in the theorem below, we obtain a different upper-bound on 
$D(\cK(P_X) \| \cK(Q_X))$ by directly applying the continuity bound for relative entropy from $(ii)$ of Lemma \ref{lemma_trunc_prop} and substituting $\eps \leftarrow \left(\eps + \log\frac{1}{1-\delta}\right)$ in Corollary \ref{corr_QLDP_Divs_bound}.

\begin{theorem}
\label{Genh}
     Let $\cK : \cX \to \cY$ be a $(\eps,\delta)$-LDP classical channel (Markov kernel). Further, for any pair of probability distributions $P_X$ and $Q_X$ over $\cX$, let $\cK(P_X)$ and $\cK(Q_X)$ be their respective output distribution with respect to $\cK$ such that $\supp( \cK(P_X) )\subseteq \supp( \cK(Q_X))$. Then, 
\begin{align*}
    D(\cK(P_X) \| \cK(Q_X)) &\leq \eps'\tanh\left(\frac{\eps'}{2}\right)+ 2 \delta \left( \eps' +  \frac{2}{m} \right),
\end{align*}
where  $\eps' = \left(\eps+\log\frac{1}{1-\delta}\right)$ and $m = \min_{y \in \supp( \cK(P_X))}\left\{\min\{\Tilde{P}_Y(y), \cK(P_X) (y)\},\min\{\Tilde{Q}_Y(y), \cK(Q_X)(y)\}\right\}$ 
and $(\Tilde{P}_Y,\Tilde{Q}_Y)$ is 
the truncated  pair with respect to $(\cK(P_X), \cK(Q_X))$.
\end{theorem}

\begin{remark}
    It is important to note that the bound obtained in Theorem \ref{theorem_classical_kl_bound} is tighter than the bound in Theorem \ref{Genh}. This is because the bound in Theorem \ref{Genh} is obtained by first approximating the $(\eps,\delta)$-DP pair with an $\eps'$-DP pair, where $\eps' = \eps + \log\frac{1}{1-\delta}$, and then applying the known bounds for pure DP. This approximation introduces looseness, particularly in the leading term, which becomes $\mathcal{O}((\eps+\delta)\tanh(\eps+\delta))$. In contrast, Theorem \ref{theorem_classical_kl_bound} uses a more direct approach via the integral representation of the KL divergence, resulting in a leading term of $\mathcal{O}(\eps\tanh(\eps))\norm{P_X-Q_X}{1}$. For small $\eps$, the bound in Theorem \ref{theorem_classical_kl_bound} is therefore significantly tighter.
\end{remark}

\section{Conclusion}

In this work, motivated by Blackwell's order of informativeness \cite{Blackwell51,BG54} between two statistical experiments, we utilize an order of informativeness in the quantum setting, which is based on the hypothesis testing divergence dominance. { While the theoretical formulation of this quantum information order and its relation to the Hockey-Stick divergence were established in foundational works such as \cite{BG17}, we contribute a simplified, elegant proof of this relation.} In the classical case, Blackwell's order implies that there always exists a Markov kernel that maps the more informative statistical experiment to the lesser informative statistical experiment. However, in the quantum case, a CP-TP map may not always exist \cite{Matsumoto2014} from a more informative pair of quantum states to the lesser informative pair of quantum states. Despite the lack of existence of such CP-TP maps, we still upper-bound the $f$-divergence of the lesser informative pair of quantum states by that of a more informative pair of quantum states by demonstrating that Hockey-Stick dominance implies a universal $f$-divergence dominance. This approach allows us to fully characterize $(\eps,\delta)$-QDP mechanisms.

We use this ordering to study the problem of quantum privatized hypothesis testing and quantum privatized parameter estimation, where the trade-off between privacy and statistical utility is characterized via the quantum hypothesis testing divergence and the quantum Fisher information, respectively. We derived explicit expressions for the maximal Fisher information achievable under quantum $(\eps,\delta)$-DP constraints, thereby quantifying the fundamental limits of parameter estimation in the presence of quantum privacy mechanisms. 

Further, we derive near-optimal bounds on the contraction coefficient of $(\eps,\delta)$-DP CP-TP maps with respect to the Hockey-Stick divergence. This allows us to prove bounds on the relative entropy between the output pair induced by any $(\eps,\delta)$-DP classical channels.

\section*{Acknowledgments}

We are grateful to Henrik Wilming for pointing out an error in the earlier version of this manuscript. His comments have greatly helped us to make the technical contributions of this manuscript more rigorous. We thank Francesco Buscemi for pointing out several useful references and also thank Mark Wilde for his comments on the earlier version.

The work of NAW was supported by DST/NQM/QA/2026/185 and Grant ANRF/ARG/2025/012066/MS from the Department of Science \& Technology, Govt of India. The work of AD was supported by TCS Research Scholar Program (Cycle 19). The work of MH was supported in part by the General R\&D Projects of 1+1+1 CUHK-CUHK(SZ)-GDST Joint Collaboration Fund (Grant No. GRDP2025-022), the Guangdong Provincial Quantum Science Strategic Initiative (Grant No. 
GDZX2505003), 
the Shenzhen International Quantum Academy (Grant No. SIQA2025KFKT07),
and the National Natural Science Foundation of China under Grant 62171212.

\bibliographystyle{IEEEtran}
\bibliography{master}

\appendix



{ \subsection{Derivation of the Optimal State Pair $(\rho_{(\eps,\delta)},\sigma_{(\eps,\delta)})$}\label{app_optimal_states}

In this appendix, we provide the constructive proof for the existence of the pair of quantum states $(\rho_{(\eps,\delta)},\sigma_{(\eps,\delta)}) \in \mathfrak{D}_{(\eps,\delta)}$ introduced in the main text. We systematically deduce the structure of these states by requiring them to achieve the extremal/corner points of $\cR(\eps,\delta)$ under certain choices of measurements.

Consider a pair $\rho_{(\eps,\delta)}$ and $\sigma_{(\eps,\delta)}$ over a finite-dimensional Hilbert space $\cH$ such that for a one-dimensional projector $\ketbra{v}$, 
\begin{align}
    \tr[\ketbra{v}\rho_{(\eps,\delta)}] &\triangleq \delta, \label{app_vketrho}\\
    \tr[\ketbra{v}\sigma_{(\eps,\delta)}] &\triangleq 0.\label{app_vketsig}
\end{align}
Then, the pair $(\rho_{(\eps,\delta)},\sigma_{(\eps,\delta)})$ achieves the corner point $(\delta,1)$ of $\cR(\eps,\delta)$. Similarly, for the one-dimensional projector $\bbI - \ketbra{v}$, the pair achieves the corner point $(1-\delta,0)$.

Further, consider another one-dimensional projector $\ketbra{x}$ which is perpendicular to $\ketbra{v}$, such that,
\begin{align}
    \tr[\ketbra{x}\rho_{(\eps,\delta)}] &\triangleq 0, \label{app_xketrho}\\
    \tr[\ketbra{x}\sigma_{(\eps,\delta)}] &\triangleq \delta.\label{app_xketsig}
\end{align}
Then, the pair achieves the corner point $(0,1-\delta)$ of $\cR(\eps,\delta)$. Similarly, for $\bbI - \ketbra{x}$, it achieves the corner point $(1,\delta)$. 

We notice that the cardinality of the support of the state pair must be at least $3$ since we need two more corner points to be achieved. If we consider another one-dimensional projector $\bbI - \ketbra{v} - \ketbra{x}$ (perpendicular to both $\ketbra{v}$ and $\ketbra{x}$), we have:
\begin{align}
    \tr[\left(\bbI - \ketbra{v} - \ketbra{x}\right)\rho_{(\eps,\delta)}] &= 1-\delta, \label{app_xvcketrho}\\
    \tr[\left(\bbI - \ketbra{v} - \ketbra{x}\right)\sigma_{(\eps,\delta)}] &= 1 - \delta.\label{app_xvcketsig}
\end{align}
Thus, the pair achieves the point $(1-\delta,\delta)$ of $\cR(\eps,\delta)$. However, this is an interior point of $\cR(\eps,\delta)$ and not a corner point. Combining the above one-dimensional projectors can never achieve the two symmetric error points of $\cR(\eps,\delta)$. Therefore, we deduce the support of the state pair $(\rho_{(\eps,\delta)},\sigma_{(\eps,\delta)})$ must be at least $4$.

We now consider a third one-dimensional projector $\ketbra{y}$ perpendicular to both $\ketbra{v}$ and $\ketbra{x}$, and define the two-dimensional projector $(\ketbra{v} + \ketbra{y})$. We choose this projector to achieve the best-case symmetric error point, i.e., $(\frac{e^\eps + \delta}{1+e^\eps}, \frac{e^\eps + \delta}{1+e^\eps})$. (See Remark \ref{remark_worst_choice} for why choosing the worst-case symmetric error point fails the $(\eps,\delta)$-DP condition). Towards this, we assume:
\begin{align}
    \tr[(\ketbra{v} + \ketbra{y})\rho_{(\eps,\delta)}] &\triangleq \frac{e^\eps+\delta}{1+e^\eps}, \label{app_vyketrho}\\
    \tr[\left(\bbI - (\ketbra{v} + \ketbra{y})\right)\sigma_{(\eps,\delta)}] &\triangleq \frac{e^\eps+\delta}{1+e^\eps}.\label{app_vyketsig}
\end{align} 
Now, \eqref{app_vketrho}, \eqref{app_vketsig}, \eqref{app_vyketrho}, and \eqref{app_vyketsig} give us:
\begin{align}
    \tr[\ketbra{y}\rho_{(\eps,\delta)}] &= \frac{e^\eps(1-\delta)}{1+e^\eps}, \label{app_yketrho}\\
    \tr[\ketbra{y}\sigma_{(\eps,\delta)}] &= \frac{1-\delta}{1+e^\eps}.\label{app_yketsig}
\end{align}

Finally, consider a fourth one-dimensional projector $\ketbra{z} = \bbI - \ketbra{v} - \ketbra{x} - \ketbra{y}$ perpendicular to $\ketbra{v}$, $\ketbra{x}$ and $\ketbra{y}$. Using our previous trace constraints, we find:
\begin{align}
    \tr[\ketbra{z}\rho_{(\eps,\delta)}] &= \frac{1-\delta}{1+e^\eps}, \label{app_zketrho}\\
    \tr[\ketbra{z}\sigma_{(\eps,\delta)}] &= \frac{e^\eps(1-\delta)}{1+e^\eps}.\label{app_zketsig}
\end{align}
One can verify that this pair achieves the worst-case symmetric error point $(\frac{1-\delta}{1+e^\eps},\frac{1-\delta}{1+e^\eps})$ of $\cR(\eps,\delta)$ with the help of the two-dimensional projector $\ketbra{x} + \ketbra{z}$.

Therefore, we have systematically constructed a pair of quantum states:
\begin{align}
    \rho_{(\eps,\delta)} & = \delta \ketbra{v} + \frac{e^\eps(1-\delta)}{1+e^\eps} \ketbra{y} + \frac{1-\delta}{1+e^\eps} \ketbra{z} + 0 \ketbra{x},\label{app_eq_rho} \\
    \sigma_{(\eps,\delta)} & = 0 \ketbra{v}+ \frac{1-\delta}{1+e^\eps} \ketbra{y} + \frac{e^\eps(1-\delta)}{1+e^\eps} \ketbra{z} + \delta \ketbra{x}.\label{app_eq_sigma}
\end{align}
By mapping the abstract projectors $\ketbra{v}, \ketbra{y}, \ketbra{z}, \ketbra{x}$ to the four-dimensional computational basis states $\ketbra{00}, \ketbra{01}, \ketbra{10}, \ketbra{11}$ respectively, we recover the exact simplified states presented in equations \eqref{eq_rho_simp} and \eqref{eq_sigma_simp} of the main text.


\begin{remark}\label{remark_worst_choice}
\begin{itemize}
    \item[] ~
    \item It is important to note that in \eqref{app_vyketrho} and \eqref{app_vyketsig}, if we had chosen the RHS to be $\frac{1-\delta}{1+e^\eps}$ and $\frac{1-\delta}{1+e^\eps}$ respectively, we would get another pair of quantum states $(\rho'_{(\eps,\delta)},\sigma'_{(\eps,\delta)})$ which also achieves the extremal points of the characteristic region $\cR(\eps,\delta)$ for any given $\eps \geq 0$ and $\delta \in [0,1]$. This pair is given as follows,
    \begin{align}
        \rho'_{(\eps,\delta)} & = \delta \ketbra{00} + \frac{1-2\delta-e^\eps\delta}{1+e^\eps} \ketbra{01} + \frac{e^\eps + \delta}{1+e^\eps} \ketbra{10},\label{eq_rho'_simp}\\
        \sigma'_{(\eps,\delta)} & =  \frac{e^\eps + \delta}{1+e^\eps} \ketbra{01} + \frac{1-2\delta-e^\eps\delta}{1+e^\eps} \ketbra{10} + \delta \ketbra{11}.\label{eq_sigma'_simp}
    \end{align}
    However, the pair  $(\rho'_{(\eps,\delta)},\sigma'_{(\eps,\delta)})$ also achieves two points $\left(\frac{1 -\delta}{1+e^\eps} - \delta,\frac{1 -\delta}{1+e^\eps} - \delta\right)$ and $\left(\frac{e^\eps+\delta}{1 + e^\eps} +\delta, \frac{e^\eps+\delta}{1 + e^\eps} +\delta\right)$, which are outside the characteristic region $\cR(\eps,\delta)$. From this, we can conclude that the pair $(\rho'_{(\eps,\delta)},\sigma'_{(\eps,\delta)}) \not\in \mathfrak{D}_{(\eps,\delta)}$.
    \item Further, one can verify that the above pair satisfies $(\log (\frac{e^\eps}{1 - \delta(2 + e^\eps)}),\delta)$-DP.

\end{itemize}    
\end{remark}

\begin{figure}[h]
    \centering
    \begin{tikzpicture}

\coordinate (A) at (1.75,1.75);
\coordinate (B) at (0,5);
\coordinate (C) at (0,5.5);
\coordinate (D) at (0.5,5.5);
\coordinate (E) at (3.75,3.75);
\coordinate (F) at (5.5,0.5);
\coordinate (G) at (5.5,0);
\coordinate (H) at (5,0);
\coordinate (I) at (1.25,1.25);
\coordinate (J) at (4.25,4.25);

\draw[fill=red] (A) circle (0.04); 
\draw[fill=red] (B) circle (0.04); 
\draw[fill=red] (C) circle (0.04); 
\draw[fill=red] (D) circle (0.04); 
\draw[fill=red] (E) circle (0.04); 
\draw[fill=red] (F) circle (0.04); 
\draw[fill=red] (G) circle (0.04); 
\draw[fill=red] (H) circle (0.04); 
\draw[fill=red] (I) circle (0.04);
\draw[fill=red] (J) circle (0.04);
\draw[fill=red] (0,0) circle (0.04); 
\draw[fill=red] (5.5,5.5) circle (0.04);

\draw[->] (0,0) -- (6,0) node[below,xshift = -4cm,yshift = 0.1cm] {\makecell{$\alpha \longrightarrow$ \\ (Type-I error)}};
\draw[->] (0,0) -- (0,6) node[left, rotate = 90, xshift = -1.75cm, yshift = 0.5cm] {\makecell{$\beta \longrightarrow$ \\ (Type-II error)}};
\draw[dotted,thick] (0,0) -- (5.5,5.5) node[anchor = north east, rotate = 45,xshift = -0.75cm] {$\alpha = \beta$};
\draw[dotted,thick] (0,5.5) -- (5.5,0 ) node[anchor = north east, rotate = -45,xshift = -2cm] {$\alpha + \beta = 1$};
\draw[-] (0,5.5) -- (5.5,5.5) node[above, xshift = -2.5cm] {$\beta = 1$};
\draw[-] (5.5,0) -- (5.5,5.5) node[right, xshift = 0.25cm, yshift = -3cm, rotate = 90] {$\alpha = 1$};
\node[anchor = north east] at (0,0) {$(0,0)$};
\node[anchor =  east] at (0,5.5) {$(0,1)$};
\node[anchor = north west] at (5.5,0) {$(1,0)$};
\node[anchor = south west] at (5.5,5.5) {$(1,1)$};

\filldraw[fill=red, pattern=north east lines, draw=black, very thick, opacity=0.5] (1.75,1.75) -- (0,5) -- (0,5.5)-- (0.5,5.5) -- (3.75,3.75) -- (5.5, 0.5)-- (5.5,0) -- (5,0)-- cycle;
\filldraw[fill= blue, pattern=north west lines, draw=black, very thick, opacity=0.5] (I) -- (0,5) -- (0,5.5)-- (0.5,5.5) -- (J) -- (5.5, 0.5)-- (5.5,0) -- (5,0)-- cycle;

\node[anchor =  east] at (0, 5) {$(0, 1-\delta)$};
\node[anchor =  north, xshift = -0.25 cm] at (5, 0) {$(1-\delta,0)$};
\node[anchor =  south] at (D) {$(\delta, 1)$};
\node[anchor =  west, xshift = 0 cm] at (F) {$(1,\delta)$};

\node[anchor =  east, xshift = -3 cm, yshift = -0.25 cm] (WFP) at (I) {\makecell{$\left(\frac{1 -\delta}{1+e^\eps} - \delta,\frac{1 -\delta}{1+e^\eps} - \delta\right)$ \\ (Point outside $\cR(\eps,\delta)$)}};
\node[anchor =  west,xshift=  3cm, yshift = 0.25 cm] (BFP) at (J) {\makecell{$\left(\frac{e^\eps+\delta}{1 + e^\eps} +\delta, \frac{e^\eps+\delta}{1 + e^\eps} +\delta\right)$\\(Point outside $\cR(\eps,\delta)$)}};
\draw[-{Triangle[width=4pt,length=4pt]}] (WFP) -- (I);
\draw[-{Triangle[width=4pt,length=4pt]}] (BFP) -- (J);

\end{tikzpicture}
    \caption{Graphical representation of  $\cR(\rho'_{(\eps,\delta)},\sigma'_{(\eps,\delta)} )$ (in the whole shaded region) which is strictly larger than $\cR(\eps,\delta)$ (in inner shaded region), where the two extremal points of $\cR(\rho'_{(\eps,\delta)},\sigma'_{(\eps,\delta)} )$ outside $\cR(\eps,\delta)$ are $\left(\frac{1 -\delta}{1+e^\eps} - \delta,\frac{1 -\delta}{1+e^\eps} - \delta\right)$ and $\left(\frac{e^\eps+\delta}{1 + e^\eps} +\delta, \frac{e^\eps+\delta}{1 + e^\eps} +\delta\right)$.}
    \label{fig:privacy_region_invalid}
\end{figure}

In Figure \ref{fig:privacy_region_invalid} above, we illustrate a graphical representation of the characteristic region $\cR(\rho'_{(\eps,\delta)},\sigma'_{(\eps,\delta)} )$ along with the extremal points achieved by the pair $(\rho'_{(\eps,\delta)},\sigma'_{(\eps,\delta)})$.

}
\subsection{Proof of Lemma \ref{fact_equi_ah_sd}}\label{proof_fact_equi_ah_sd}
As the condition \eqref{fact_equi_ah_sd_eq}, we are given that 
for all $\alpha \in [0,1]$,
\begin{equation}
D^{\alpha}_{H}{(\rho_1\|\rho_2)} \geq D^{\alpha}_{H}{(\sigma_1\|\sigma_2)}. \label{eq:assumption}
\end{equation}
From Definition \ref{d_hyp_func} of 
the quantum hypothesis testing divergence,
the condition \eqref{eq:assumption} is rewritten as
\begin{align}
\max_{\substack{0 \preceq \Lambda \preceq \bbI:\\
\tr[\Lambda \rho_1] \leq \alpha}} -\ln \tr[(\bbI - \Lambda)\rho_2] 
\geq \max_{\substack{0 \preceq \Delta \preceq \bbI:\\
\tr[\Delta \sigma_1] \leq \alpha}} -\ln \tr[(\bbI - \Delta)\sigma_2]. \label{eq:divergence_comparison}
\end{align}
for all $\alpha \in [0,1]$.
Exponentiating both sides and reversing the inequality direction (since $-\ln$ is a decreasing function), we get:
\begin{align}
\min_{\substack{0 \preceq \Lambda \preceq \bbI:\\
\tr[\Lambda \rho_1] \leq \alpha}} \tr[(\bbI - \Lambda)\rho_2] 
\leq \min_{\substack{0 \preceq \Delta \preceq \bbI:\\
\tr[\Delta \sigma_1] \leq \alpha}} \tr[(\bbI - \Delta)\sigma_2]. \label{eq:min_comparison}
\end{align}
Now, let us define $\beta := 1 - \alpha$. Then the constraint $\tr[\Lambda \rho_1] \leq \alpha$ becomes $\tr[(\bbI - \Lambda)\rho_1] \geq \beta$. Similarly for $\sigma_1$. Thus, we can rewrite the inequality as:
\begin{align}
\max_{\substack{0 \preceq \Lambda \preceq \bbI:\\
\tr[(\bbI - \Lambda)\rho_1] \geq \beta}} \tr[\Lambda \rho_2] 
\geq \max_{\substack{0 \preceq \Delta \preceq \bbI:\\
\tr[(\bbI - \Delta)\sigma_1] \geq \beta}} \tr[\Delta \sigma_2]. \label{eq:max_comparison}
\end{align}
This inequality holds for all $\beta \in [0,1]$.

Now fix any $\gamma \geq 0$, and consider the following optimization problem:
\begin{align}
\max_{0 \preceq \Delta \preceq \bbI} \tr[\Delta \sigma_1] + \gamma \tr[(\bbI - \Delta)\sigma_2].\label{BK2}
\end{align}
Let $\Delta^{\star}_{\gamma}$ be the optimal solution to the problem \eqref{BK2}. Define:
\[
\beta^{\star}_{\gamma} := \tr[\Delta^{\star}_{\gamma} \sigma_1], \quad
\delta^{\star}_{\gamma} := \tr[(\bbI - \Delta^{\star}_{\gamma}) \sigma_2].
\]
Then the optimal value is given as 
\begin{align}
\max_{0 \preceq \Delta \preceq \bbI} \tr[\Delta \sigma_1] + \gamma \tr[(\bbI - \Delta)\sigma_2]
=\beta^{\star}_{\gamma} + \gamma \delta^{\star}_{\gamma}.\label{BK3}
\end{align}
From inequality \eqref{eq:max_comparison}, we have:
\begin{align}
\max_{\substack{0 \preceq \Lambda \preceq \bbI:\\
\tr[\Lambda \rho_1] \geq \beta^{\star}_{\gamma}}} \tr[(\bbI - \Lambda)\rho_2] 
\geq \max_{\substack{0 \preceq \Delta \preceq \bbI:\\
\tr[\Delta \sigma_1] = \beta^{\star}_{\gamma}}} \tr[(\bbI - \Delta)\sigma_2]
\geq \delta^{\star}_{\gamma}. \label{eq:delta_bound}
\end{align}

Let $\Lambda^{\star}$ be the operator achieving the maximum on the left-hand side of \eqref{eq:delta_bound}. Then:
\begin{align}
\max_{0 \preceq \Lambda \preceq \bbI} \tr[\Lambda \rho_1] + \gamma \tr[(\bbI - \Lambda)\rho_2]
&\geq \tr[\Lambda^{\star} \rho_1] + \gamma \tr[(\bbI - \Lambda^{\star})\rho_2] \nn \\
&\overset{(a)}{\geq} \beta^{\star}_{\gamma} + \gamma \delta^{\star}_{\gamma} \nn\\
&\overset{(b)}{=} \max_{0 \preceq \Delta \preceq \bbI} \tr[\Delta \sigma_1] + \gamma \tr[(\bbI - \Delta)\sigma_2],\label{eq:final_bound}
\end{align}
where $(a)$ follows from \eqref{eq:delta_bound},
and $(b)$ follows from \eqref{BK3}.{ Rewriting the left-hand and right-hand sides, we obtain,

$$ \max_{0 \preceq \Lambda \preceq \bbI} \tr[\Lambda (\rho_1 - \gamma \rho_2)] \geq \max_{0 \preceq \Delta \preceq \bbI} \tr[\Delta (\sigma_1 - \gamma \sigma_2)].$$

To complete the proof for the Hockey-Stick divergence across all $\gamma \geq 0$, we account for the equivalent form $E_{\gamma}(\rho\|\sigma) = \max_{0 \preceq \Lambda \preceq \bbI} \tr[\Lambda (\rho - \gamma\sigma)] - (\tr[\rho - \gamma \sigma])_{+}$ (see Fact \ref{Fact_hockey_equiv_form}). Because $\rho_1, \rho_2, \sigma_1,$ and $\sigma_2$ are valid quantum states with unit trace, the subtractive term evaluates identically for both pairs of states,

$$ (\tr[\rho_1 - \gamma \rho_2])_{+} = (\tr[\sigma_1 - \gamma \sigma_2])_{+} = (1 - \gamma)_{+}. $$

Subtracting this state-independent constant from both sides of the established maximization inequality perfectly preserves the relation,

$$ \max_{0 \preceq \Lambda \preceq \bbI} \tr[\Lambda (\rho_1 - \gamma \rho_2)] - (1 - \gamma)_{+} \geq \max_{0 \preceq \Delta \preceq \bbI} \tr[\Delta (\sigma_1 - \gamma \sigma_2)] - (1- \gamma)_{+} .$$

This is exactly equivalent to $E_{\gamma}(\rho_1\|\rho_2) \geq E_{\gamma}(\sigma_1\|\sigma_2)$, demonstrating that the implication holds universally for every $\gamma \geq 0$. This completes the proof of \eqref{BN1}.}
This completes the proof of Lemma \ref{fact_equi_ah_sd}.\hfill\QED

{\subsection{Proof of Lemma \ref{lemma_DP_t_hyp_rel}}\label{proof_lemma_DP_t_hyp_rel}
\textbf{Forward Direction ($\implies$):} 
Assume the pair $(\rho,\sigma) \in \mathfrak{D}_{(\eps,\delta)}$. For a fixed $\alpha \in [0,1]$, when any $ 0 \preceq \Lambda \preceq \bbI$ satisfies the relation $\tr[\Lambda \rho] \leq \alpha$, then,
\begin{align}
        \tr[(\bbI - \Lambda)\sigma] &\overset{(a)}{\geq} e^{-\eps}\left(\tr[(\bbI - \Lambda) \rho] - \delta\right)\nn\\
        &\overset{(b)}{\geq} e^{-\eps}\left(1 -\alpha - \delta\right),\label{lemma_DP_t_hyp_rel_eq1}
\end{align}
where $(a)$ follows from the first equation in \eqref{def_DP_state_eq} and $(b)$ follows from the fact that $\tr[\Lambda \rho] \leq \alpha$. 
Further, the second equation in \eqref{def_DP_state_eq} yields another evaluation of $\tr[(\bbI - \Lambda)\sigma]$ as follows,
\begin{align}
        \tr[(\bbI - \Lambda)\sigma] &= \tr[\sigma] - \tr[\Lambda \sigma]\nn\\
        &\geq 1 - \delta - e^{\eps}\left(\tr[\Lambda \rho]\right)\nn\\
        &\geq 1- \delta - e^\eps\alpha. \label{lemma_DP_t_hyp_rel_eq2}
\end{align}
Further, Definition \ref{Q-def_min_type2_err} implies the relation $ \beta(\alpha |\rho\|\sigma)\geq 0$ for any $\alpha \in [0,1]$. 
The combination of \eqref{lemma_DP_t_hyp_rel_eq1} and \eqref{lemma_DP_t_hyp_rel_eq2} yields the relation $\beta(\alpha |\rho\|\sigma) \geq \max\{1- \delta - e^\eps\alpha,e^{-\eps}\left(1 -\alpha - \delta\right),0\}$. Further, using  \eqref{fact_d_hyp_func_eq} of Definition \ref{d_hyp_func}, we have,
\begin{equation*}
        D^{\alpha}_{H}(\rho\|\sigma) \leq -\log \max\{1- \delta - e^\eps\alpha,e^{-\eps}\left(1 -\alpha - \delta\right),0\}.
\end{equation*}

\textbf{Converse Direction ($\impliedby$):} 
Conversely, assume that $D^{\alpha}_{H} {(\rho\|\sigma)} \leq -\log f_{\eps,\delta}(\alpha)$ or, equivalently, $\beta(\alpha|\rho\|\sigma) \geq  f_{\eps,\delta}(\alpha)$ for all $\alpha \in [0,1]$. 

Let $0 \preceq \Lambda \preceq \bbI$ be an arbitrary POVM effect, and define $\alpha_1 := \tr[\Lambda\rho]$. Since $\Lambda$ naturally satisfies the optimization constraint $\tr[\Lambda\rho] \leq \alpha_1$, Definition \ref{Q-def_min_type2_err} implies,
\begin{equation}
    \tr[(\bbI - \Lambda)\sigma] \geq \beta(\alpha_1|\rho\|\sigma).
\end{equation}
Using the first term inside the maximum function of $f_{\eps,\delta}(\alpha_1)$, we obtain,
\begin{align}
    1 - \tr[\Lambda\sigma] &\geq 1 - \delta - e^\eps \alpha_1 \nn\\
    1 - \tr[\Lambda\sigma] &\geq 1 - \delta - e^\eps \tr[\Lambda\rho] \nn\\
    \tr[\Lambda\sigma] &\leq e^\eps \tr[\Lambda\rho] + \delta. \label{converse_proof_cond1}
\end{align}
This explicitly recovers the first privacy condition of \eqref{def_DP_state_eq}. 

To recover the second symmetric condition, we evaluate the complementary POVM effect $\Lambda' := \bbI - \Lambda$, and define $\alpha_2 := \tr[\Lambda'\rho] = 1 - \tr[\Lambda\rho]$. Since $\Lambda'$ serves as a valid candidate for the constraint $\tr[\Lambda'\rho] \leq \alpha_2$, Definition \ref{Q-def_min_type2_err} yields,
\begin{equation}
    \tr[(\bbI - \Lambda')\sigma] \geq \beta(\alpha_2|\rho\|\sigma).
\end{equation}
Substituting $\bbI - \Lambda' = \Lambda$ alongside the second term inside the maximum function of $f_{\eps,\delta}(\alpha_2)$, we expand,
\begin{align}
    \tr[\Lambda\sigma] &\geq e^{-\eps}(1 - \delta - \alpha_2) \nn\\
    \tr[\Lambda\sigma] &\geq e^{-\eps}\left(1 - \delta - (1 - \tr[\Lambda\rho])\right) \nn\\
    e^\eps \tr[\Lambda\sigma] &\geq \tr[\Lambda\rho] - \delta \nn\\
    \tr[\Lambda\rho] &\leq e^\eps \tr[\Lambda\sigma] + \delta. \label{converse_proof_cond2}
\end{align}
Since \eqref{converse_proof_cond1} and \eqref{converse_proof_cond2} hold uniformly for any choice of $0 \preceq \Lambda \preceq \bbI$, the pair $(\rho,\sigma)$ satisfies $(\eps,\delta)$-DP. This completes the proof of Lemma \ref{lemma_DP_t_hyp_rel}.\hfill\QED}
 \subsection{Proof of Lemma \ref{lemma_diff_priv_quant}}\label{proof_lemma_diff_priv_quant}
    For every $\alpha \in [0,1]$ and the states $\rho_{(\eps,\delta)}$ and $\sigma_{(\eps,\delta)}$, from Fact \ref{T_hyp_func_equi}, 
$D^{\alpha}_{H}(\rho_{(\eps,\delta)}\|\sigma_{(\eps,\delta)})$ 
and $\beta(\alpha|\rho_{(\eps,\delta)}\|\sigma_{(\eps,\delta)})$
are written as,
\begin{align}
        D^{\alpha}_{H}(\rho_{(\eps,\delta)}\|\sigma_{(\eps,\delta)}) 
        &= -\log \beta(\alpha|\rho_{(\eps,\delta)}\|\sigma_{(\eps,\delta)})\label{lemma_diff_priv_quant_eq9} \\
\beta(\alpha|\rho_{(\eps,\delta)}\|\sigma_{(\eps,\delta)}) &=  \min_{\substack{0\preceq \Lambda\preceq \bbI :\\ \tr\left[\Lambda \rho_{(\eps,\delta)}\right] = \alpha}}\tr\left[(\bbI-\Lambda )\sigma_{(\eps,\delta)}\right]\label{lemma_diff_priv_quant_eq6}\\
\end{align}

Equation \eqref{lemma_diff_priv_quant_eq9} yields that proving Lemma \ref{lemma_diff_priv_quant} is equivalent to showing $\beta(\alpha|\rho_{(\eps,\delta)}\|\sigma_{(\eps,\delta)}) = f_{(\eps,\delta)}$. Further, from the definition of $f_{(\eps,\delta)}$ in Lemma \ref{lemma_DP_t_hyp_rel}, we have the following,

\begin{equation}
    f_{(\eps,\delta)}(\alpha) = \begin{cases}
        1 - \delta - e^{\eps}\alpha, & \mbox{ if}~ \alpha \in [0,\frac{1-\delta}{e^\eps + 1}], \\
        e^{-\eps}(1 -\delta-\alpha), & \mbox{ if}~ \alpha \in [\frac{1-\delta}{e^\eps + 1},1-\delta],\\
        0 , & \mbox{otherwise}.
    \end{cases}\label{fepdel_c}
\end{equation}

In the following three steps, we will show that for all the three cases mentioned in  \eqref{fepdel_c}, we have $\beta(\alpha|\rho_{(\eps,\delta)}\|\sigma_{(\eps,\delta)}) = f_{(\eps,\delta)}(\alpha)$.

\begin{steps}

    \item For the case when $\alpha \in [0,\left.\frac{1-\delta}{e^\eps + 1}\right],$ consider a POVM $\{\Lambda,\bbI - \Lambda\}$ 
where 
\begin{equation}
   \Lambda = b \ketbra{10} + \ketbra{11},\label{POVM_1} 
\end{equation}
and $b \in [0,1]$ is such that $\alpha = \frac{b (1-\delta)}{e^\eps + 1}$.

\quad Then, 
\begin{align}
    &~~~\tr\left[\Lambda \rho_{(\eps,\delta)}\right]\nn\\ 
      &= \tr\left[\left( b \ketbra{10} + 1\ketbra{11} \right)\left( \delta \ketbra{00} + \frac{(1-\delta)e^\eps}{1+e^\eps}\ketbra{01} + \frac{1-\delta}{1+e^\eps}\ketbra{10}\right)\right]\nn\\
    &= \frac{b (1-\delta)}{e^\eps + 1}\nn\\
    &= \alpha, \nn
\end{align}


\quad Further, we have,
\begin{align}
    &~~\tr\left[(\bbI-\Lambda )\sigma_{(\eps,\delta)}\right]\nn\\
    &= 1 - \tr\left[\left(b \ketbra{10} + 1\ketbra{11} \right) \left(\frac{1-\delta}{1+e^\eps}\ketbra{01} + \frac{(1-\delta)e^\eps}{1+e^\eps}\ketbra{10} + \delta\ketbra{11}\right)\right]\nn\\
    &= 1 - \delta - e^\eps \frac{b (1-\delta)}{e^\eps + 1}\nn\\
    &= 1 -\delta -e^\eps\alpha\nn \\
    & = f_{(\eps,\delta)}(\alpha). \nn
\end{align}



\quad Since $(\rho_{(\eps, \delta)},\sigma_{(\eps, \delta)})\in  \mathfrak{D}_{(\eps,\delta)}$, \eqref{lemma_DP_t_hyp_rel_eq1} of Lemma \ref{lemma_DP_t_hyp_rel} implies that $\beta(\alpha|\rho_{(\eps, \delta)}\|\sigma_{(\eps, \delta)}) \geq f_{(\eps,\delta)}$. Therefore, the POVM obtained in \eqref{POVM_1} is optimal. Hence, for $\alpha \in [0,\left.\frac{1-\delta}{e^\eps + 1}\right],$ we have $\beta(\alpha|\rho_{(\eps,\delta)}\|\sigma_{(\eps,\delta)}) = f_{(\eps,\delta)}(\alpha).$

\item For the case when $\alpha \in \left[\frac{1-\delta}{e^\eps + 1}\right., 1 -\delta]$, consider a POVM $\{\Lambda,\bbI - \Lambda\}$ 
where 
\begin{equation}
 \mathbb{I} - \Lambda = \ketbra{00} + c \ketbra{01} ,\label{POVM_2} 
\end{equation}
and $c \in [0,1]$ is such that $\alpha = (1-\delta)\left(1 - \frac{e^\eps  c}{ 1 + e^\eps}\right)$.

\quad Then, 
\begin{align}
   &~~~\tr\left[\Lambda\rho_{(\eps,\delta)}\right] \nn\\
   &=1-\tr\left[(\mathbb{I}-\Lambda)\rho_{(\eps,\delta)}\right]\nn\\
    &= 1-\tr\left[\left(\ketbra{00} + c \ketbra{01}\right) \left(\delta \ketbra{00} + \frac{(1-\delta)e^\eps}{1+e^\eps}\ketbra{01} + \frac{1-\delta}{1+e^\eps}\ketbra{10}\right)\right]\nn\\
    &= 1-\delta - (1-\delta)\frac{e^\eps  c}{1+e^\eps }\nn\\
    &= (1-\delta)\left(1 - \frac{e^\eps  c}{1+e^\eps }\right)\nn\\
    &= \alpha.\nn
\end{align}

\quad Further, we have,
\begin{align}
    &~~\tr\left[(\mathbb{I}-\Lambda)\sigma_{(\eps,\delta)}\right]\nn\\
    &= \tr\left[\left(\ketbra{00} + c \ketbra{01}\right) \left(\frac{1-\delta}{1+e^\eps}\ketbra{01} + \frac{(1-\delta)e^\eps}{1+e^\eps}\ketbra{10} + \delta\ketbra{11}\right)\right]\nn\\
    &= c\frac{1-\delta}{1+e^\eps} \nn\\
    &= e^{-\eps}(1 -\delta)\left(\frac{e^\eps c}{1 + e^\eps}\right)\nn\\
    &= e^{-\eps} (1 -\delta - \alpha)\nn\\
    & =f_{(\eps,\delta)}(\alpha).\nn
\end{align}

\quad Since $(\rho_{(\eps, \delta)},\sigma_{(\eps, \delta)})\in  \mathfrak{D}_{(\eps,\delta)}$, \eqref{lemma_DP_t_hyp_rel_eq1} of Lemma \ref{lemma_DP_t_hyp_rel} implies that $\beta(\alpha|\rho_{(\eps, \delta)}\|\sigma_{(\eps, \delta)}) \geq f_{(\eps,\delta)}$. Therefore, the POVM obtained in \eqref{POVM_2} is optimal. Hence, for $\alpha \in [\left.\frac{1-\delta}{e^\eps + 1}, 1-\delta\right],$ we have $\beta(\alpha|\rho_{(\eps,\delta)}\|\sigma_{(\eps,\delta)}) = f_{(\eps,\delta)}(\alpha).$

    \item For the case when $\alpha \geq 1-\delta$, consider a POVM $\{\Lambda_{\alpha},\bbI-\Lambda_{\alpha}\}$, where $\Lambda_{\alpha} := \frac{\alpha+\delta-1}{\delta} \ketbra{00} + \ketbra{01} + \ketbra{10} + \ketbra{11}$. Then, 
    \begin{align}
        &~~~\tr\left[\Lambda_{a} \rho_{(\eps,\delta)}\right]\nn\\
        &= \tr\left[\left(\frac{\alpha+\delta-1}{\delta} \ketbra{00} + \ketbra{01} + \ketbra{10} + \ketbra{11}\right) \left(\delta \ketbra{00} + \frac{(1-\delta)e^\eps}{1+e^\eps}\ketbra{01}\right.\right.\nn\\
        &~~~\left.\left. + \frac{1-\delta}{1+e^\eps}\ketbra{10}\right)\right]\nn\\
        &= \alpha+\delta-1 + (1-\delta)\nn\\
        &= \alpha.\label{lemma_diff_priv_quant_eq10}
    \end{align}
    \quad Further, we have,
    \begin{align}
        &~~\tr\left[(\bbI-\Lambda_{\alpha} )\sigma_{(\eps,\delta)}\right]\nn\\
        & = 1 - \tr\left[\left(\frac{\alpha+\delta-1}{\delta}\ketbra{00} + \ketbra{01} + \ketbra{10} + \ketbra{11}\right) \left(\frac{1-\delta}{1+e^\eps}\ketbra{01} + \frac{(1-\delta)e^\eps}{1+e^\eps}\ketbra{10} \right.\right.\nn\\
    &~~~\left.\left. + \delta\ketbra{11}\right)\right]\nn\\
    &= 1 - 1\nn\\
    &= 0.\label{lemma_diff_priv_quant_eq11}
    \end{align}
    \quad  Therefore, for $\alpha \geq 1-\delta$, \eqref{lemma_diff_priv_quant_eq6} yields that 
$\beta(\alpha|\rho_{(\eps,\delta)}\|\sigma_{(\eps,\delta)}) = 0$. 
\end{steps}

This completes the proof of Lemma \ref{lemma_diff_priv_quant}.\hfill\QED

\subsection{Proof of Claim \ref{claim_hockey_stick_bound}}\label{proof_claim_hockey_stick_bound}
From Definition \ref{f_sym_hyp_func} and Fact \ref{Fact_hockey_equiv_form}, for any $\gamma \geq 1$, we have
\begin{align}
    E_{\gamma}(\rho_{(\eps,\delta)}\|\sigma_{(\eps,\delta)}) &= \tr\left[\left(\rho_{(\eps,\delta)} - \gamma \sigma_{(\eps,\delta)}\right)_{+}\right].
\end{align}

The states $\rho_{(\eps,\delta)}$ and $\sigma_{(\eps,\delta)}$ are diagonal in the standard basis $\{\ket{00}, \ket{01}, \ket{10}, \ket{11}\}$. The operator $\rho_{(\eps,\delta)} - \gamma \sigma_{(\eps,\delta)}$ is also diagonal, with the following diagonal entries,
\begin{itemize}
    \item For $\ketbra{00}$: $\delta$,
    \item For $\ketbra{01}$: $\frac{(1-\delta)e^\eps}{1+e^\eps} - \gamma \frac{1-\delta}{1+e^\eps} = \frac{1-\delta}{1+e^\eps}(e^\eps - \gamma)$,
    \item For $\ketbra{10}$: $\frac{1-\delta}{1+e^\eps} - \gamma \frac{(1-\delta)e^\eps}{1+e^\eps} = \frac{1-\delta}{1+e^\eps}(1 - \gamma e^\eps)$,
    \item For $\ketbra{11}$: $-\gamma\delta$.
\end{itemize}

The trace of the positive part is the sum of the positive eigenvalues. Since $\gamma \geq 1$ and $\delta \in [0,1)$, $\delta \geq 0$ and $-\gamma\delta \leq 0$. Also, since $\eps \geq 0$, $e^\eps \geq 1$, so $1 - \gamma e^\eps \leq 1 - e^\eps \leq 0$.
Thus, only the first two terms can be positive.
\begin{equation}
    E_{\gamma}(\rho_{(\eps,\delta)}\|\sigma_{(\eps,\delta)}) = [\delta]_{+} + \left[\frac{1-\delta}{1+e^\eps}(e^\eps - \gamma)\right]_{+}.
\end{equation}

We consider two cases for $\gamma$:
\begin{enumerate}
    \item \textbf{Case 1: $1 \leq \gamma \leq e^\eps$}. In this case, $e^\eps - \gamma \geq 0$.
    \begin{align*}
        E_{\gamma}(\rho_{(\eps,\delta)}\|\sigma_{(\eps,\delta)}) &= \delta + \frac{1-\delta}{1+e^\eps}(e^\eps - \gamma) \\
        &= \frac{\delta(1+e^\eps) + (1-\delta)(e^\eps - \gamma)}{1+e^\eps} \\
        &= \frac{\delta + \delta e^\eps + e^\eps - \gamma - \delta e^\eps + \delta\gamma}{1+e^\eps} \\
        &= \frac{e^\eps - \gamma + \delta(1+\gamma)}{1+e^\eps}.
    \end{align*}
    \item \textbf{Case 2: $\gamma > e^\eps$}. In this case, $e^\eps - \gamma < 0$.
    \begin{align*}
        E_{\gamma}(\rho_{(\eps,\delta)}\|\sigma_{(\eps,\delta)}) &= \delta + 0 = \delta.
    \end{align*}
\end{enumerate}

This completes the proof of Claim \ref{claim_hockey_stick_bound}. \hfill\QED

\subsection{Proof of Lemma \ref{lemma_trunc_prop}}\label{proof_lemma_trunc_prop}
We prove each part in turn.

\textbf{ Proof of (i):} The definition of $(\eps,\delta)$-DP distributions yields that $\delta := \max\{E_{e^\eps}(P\|Q), E_{e^\eps}(Q\|P)\}$. Now proof of $(1)$ follows directly definition of $\tilde{P}$ and $\tilde{Q}$ mentioned in \cref{projpair_eq1,projpair_eq2}. Further, observe that for each $x \in \cX$, $P'(x) = P(x) - \abs{P(x)-e^\eps Q(x)}_{+}$ and $Q'(x) = Q(x) - \abs{Q(x)-e^\eps P(x)}_{+}$. Thus, we can upper-bound the $L_1$ distance between $P$ and $P'$ as follows,

\begin{align*}
    &~ \norm{P - P'}{1} 
    = \sum_{x \in \cX} \abs{P(x) - P'(x)}\\
    &= \sum_{x \in \cX} \abs{P(x) - \left(P(x) - \abs{P(x)-e^\eps Q(x)}_{+}\right)}\\
    &= \sum_{x \in \cX} \abs{P(x)-e^\eps Q(x)}_{+}
    = E_{e^\eps}(P\|Q) \leq \delta.
\end{align*}

Similarly, we can show that $\norm{Q - Q'}{1} \leq \delta$. Further, we denote $p := \sum_{x \in \cX}P'(x)$ and $q := \sum_{x \in \cX}Q'(x)$ and it follows that $p = 1 - E_{e^\eps}(P\|Q)$ and $q = 1 - E_{e^\eps}(Q\|P)$. Thus, we can upper-bound the $L_1$ distance between $P'$ and $\Tilde{P}$ as follows,

\begin{align*}
    &~ \norm{P' - \Tilde{P}}{1} 
    = \sum_{x \in \cX} \abs{P'(x) - \Tilde{P}(x)}\\
    &= \sum_{x \in \cX} \abs{P'(x) - \frac{P'(x)}{p}}
    = \sum_{x \in \cX} P'(x)\abs{1 - \frac{1}{p}}\\
    &\overset{(a)}{= p\abs{1 - \frac{1}{p}}}
    \overset{(b)}{=} (1-p)
     = E_{e^\eps}(P\|Q) \leq \delta,
\end{align*}
where $(a)$ follows from the definition of $p$ and $(b)$ follows from the fact that $p \leq 1$ as $p = 1 - E_{e^\eps}(P\|Q)$ and $E_{e^\eps}(P\|Q) \geq 0$. Similarly, we can show that $\norm{Q' - \Tilde{Q}}{1} \leq \delta$. Thus, using triangle inequality we have $\norm{P - \Tilde{P}}{1} \leq 2\delta$ and $\norm{Q - \Tilde{Q}}{1} \leq 2\delta$.
This completes the proof of $(i)$ of Lemma \ref{lemma_trunc_prop}.

\textbf{Proof of (ii):} For any $x \in \cX$, consider the following,
\begin{align}
    \frac{\tilde{P}(x)}{\tilde{Q}(x)} &\overset{(a)}{=} \frac{\min\{P(x),e^\eps Q(x)\}( 1 - E_{e^\eps}(Q\|P))}{\min\{Q(x),e^\eps P(x)\}  (1 - E_{e^\eps}(P\|Q))}\nn\\
    &\overset{(b)}{\leq} e^\eps \frac{( 1 - E_{e^\eps}(Q\|P))}{  (1 - E_{e^\eps}(P\|Q))}\nn\\
    &\overset{(c)}{\leq} e^\eps \frac{1}{1-\delta}\label{eq_max_div_bound},
\end{align}
where $(a)$ follows from the definition of $\Tilde{P}$ and $\tilde{Q}$ mentioned in \cref{projpair_eq1,projpair_eq2}, the validity of inequality $(b)$ can be verified from the following case studies,

\begin{itemize}
    \item Case 1: If $P(x) \leq e^\eps Q(x)$ and $Q(x) \leq e^\eps P(x)$, then $\frac{\min\{P(x),e^\eps Q(x)\}}{\min\{Q(x),e^\eps P(x)\}} = \frac{P(x)}{Q(x)} \leq e^\eps $.
    \item Case 2: If $P(x) \leq e^\eps Q(x)$ and $Q(x) > e^\eps P(x)$, then $\frac{\min\{P(x),e^\eps Q(x)\}}{\min\{Q(x),e^\eps P(x)\}} = \frac{P(x)}{e^\eps P(x)} = \frac{1}{e^\eps} \leq e^\eps$, as $\eps \geq 0$.
    \item Case 3: If $P(x) > e^\eps Q(x)$ and $Q(x) \leq e^\eps P(x)$, then $\frac{\min\{P(x),e^\eps Q(x)\}}{\min\{Q(x),e^\eps P(x)\}} = \frac{e^\eps Q(x)}{Q(x)} = e^\eps$.
    \item Case 4: If $P(x) > e^\eps Q(x)$ and $Q(x) > e^\eps P(x)$, then $ P(x) > e^\eps Q(x) > e^{2\eps} P(x)$, which implies $e^{2\eps} < 1$. But this is a contradiction since $\eps \geq 0$. Thus, this case is not possible.
\end{itemize}
 Further, inequality $(c)$ follows from the fact that $1 - E_{e^\eps}(P\|Q) \geq 1-\delta$ and $1 - E_{e^\eps}(Q\|P) \leq 1$. Thus, we have shown that the distributions $\Tilde{P}$ and $\Tilde{Q}$ satisfy the following,
\begin{align*}
    D_{\max}(\Tilde{P}\|\Tilde{Q}) = \log \max_{x \in \cX} \frac{\tilde{P}(x)}{\tilde{Q}(x)}
     \overset{(a)}{\leq} \eps + \log\frac{1}{1-\delta},\\
\end{align*}
where $(a)$ follows from \eqref{eq_max_div_bound}. Similarly, we can show that for any $x \in \cX$, $\frac{\tilde{Q}(x)}{\tilde{P}(x)} \leq e^\eps \frac{1}{1-\delta}$ and thus, we can show that $D_{\max}(\Tilde{Q}\|\Tilde{P}) \leq \eps + \log\frac{1}{1-\delta}$. This completes the proof of $(ii)$ of Lemma \ref{lemma_trunc_prop}..

{\textbf{Proof of (iii):}}
Before proceeding to the proof, we first observe that the definition of $\tilde{P}$ and $\tilde{Q}$ yields that $\supp(\tilde{P}) = \supp(\tilde{Q}) = \supp(P)$. We now consider the following,

{\allowdisplaybreaks\begin{align}
\left|D(\tilde{P}\|\tilde{Q}) - D(P\|Q)\right| & = \left|\sum_{x\in \supp(P)}\left(\tilde{P}(x)-P(x)\right)\log\frac{\tilde{P}(x)}{\tilde{Q}(x)} +\sum_{x\in\supp(P)}P(x)\left(\log \frac{\tilde{P}(x)}{\tilde{Q}(x)} -\log \frac{P(x)}{Q(x)}\right)\right|\nonumber\\
&\leq \sum_{x\in \supp(P)}\left|\left(\tilde{P}(x)-P(x)\right)\log\frac{\tilde{P}(x)}{\tilde{Q}(x)}\right| +\left|\sum_{x\in\supp(P)}P(x)\left(\log \frac{\tilde{P}(x)}{\tilde{Q}(x)} -\log \frac{P(x)}{Q(x)}\right)\right|\nonumber\\
&\overset{(a)}{\leq} \norm{P - \tilde{P}}{1}\max_{x \in \supp(\tilde{P})}\abs{\log\frac{\tilde{P}(x)}{\tilde{Q}(x)}} +\sum_{x\in\supp()}P(x)\left|\log \frac{\tilde{P}(x)}{\tilde{Q}(x)} -\log \frac{P(x)}{Q(x)}\right|\nonumber\\
\label{u1}
& \overset{(b)}{\leq} 2\delta \left(\eps+\log\frac{1}{\delta}\right) +  \sum_{x\in\mathcal{X}}P(x)\left| \log \frac{\tilde{P}(x)}{P(x)}\right|+\sum_{x\in\mathcal{X}}P(x) \left|\log \frac{Q(x)}{\tilde{Q}(x)}\right|,
\end{align}}
where $(a)$ follows from \eqref{Holder_fact_classical_p1_qinf} of Fact \ref{holder_classic} and $(b)$ follows from $(i)$ and $(ii)$ of Lemma \ref{lemma_trunc_prop}. We now analyze the second term in the RHS of \eqref{u1}.

From the mean value theorem it yields that for every $x \in \supp(\tilde{P}(x)),$

\begin{align}
\label{u4}
|\log(\tilde{P}(x)) - \log (P(x))| & \leq \frac{|\tilde{P}(x)-P(x)|}{m},\\
\label{u5}
|\log(\tilde{Q}(x)) - \log (Q(x))| & \leq \frac{|\tilde{Q}(x)-Q(x)|}{m}.
\end{align}
Thus, the second term in the RHS of \eqref{u1} is upper-bounded by $\frac{4\delta}{m}.$ This completes the proof of $(iii)$ of Lemma \ref{lemma_trunc_prop}.\hfill\QED

\subsection{Proof of Theorem \ref{theorem_classical_kl_bound}}\label{proof_theorem_classical_kl_bound}
For the pair $(\Tilde{P}_Y, \Tilde{Q}_Y)$ of truncated distributions with respect to $({\cK}(P_X), {\cK}(Q_X))$, the following holds from $(iii)$ of Lemma \ref{lemma_trunc_prop},
\begin{align}
    &~ D({\cK}(P_X)\|{\cK}(Q_X)) \leq D(\Tilde{P}_Y\|\Tilde{Q}_Y) + 2 \delta \left( \eps + \log \frac{1}{\delta} +  \frac{2}{m} \right),\label{classical_kl_p_eq1}\end{align}
where $m = \min_{y \in \supp({\cK}(P_X))}\left\{\min\{\Tilde{P}_Y(y), {\cK}(P_X)(y)\},\min\{\Tilde{Q}_Y(y), {\cK}(Q_X)(y)\}\right\}.$ 

It thus left us to upper-bound $D(\Tilde{P}_Y\|\Tilde{Q}_Y)$ as follows,

\begin{align}
    &~D(\Tilde{P}_Y\|\Tilde{Q}_Y)\nn\\
    &\overset{(a)}{=} \int_{1}^{\infty} \frac{1}{\gamma} E_{\gamma}(\Tilde{P}_Y\|\Tilde{Q}_Y) d\gamma + \int_{1}^{\infty} \frac{1}{\gamma^2}E_{\gamma}(\Tilde{Q}_Y\|\Tilde{P}_Y) d\gamma\nn\\
    &\overset{(b)}{=} \int_{1}^{e^{D_{\max}(\Tilde{P}_Y\|\Tilde{Q}_Y)}} \frac{1}{\gamma} E_{\gamma}(\Tilde{P}_Y\|\Tilde{Q}_Y) d\gamma + \int_{1}^{e^{D_{\max}(\Tilde{Q}_Y\|\Tilde{P}_Y)}} \frac{1}{\gamma^2}E_{\gamma}(\Tilde{Q}_Y\|\Tilde{P}_Y) d\gamma\nn\\
    &\overset{(c)}{\leq}  \int_{1}^{\frac{e^\eps}{1-\delta}} \frac{1}{\gamma} E_{\gamma}(\Tilde{P}_Y\|\Tilde{Q}_Y) d\gamma + \int_{1}^{\frac{e^\eps}{1-\delta}} \frac{1}{\gamma^2}E_{\gamma}(\Tilde{Q}_Y\|\Tilde{P}_Y) d\gamma\nn\\
    &\overset{(d)}{\leq}  \int_{1}^{\frac{e^\eps}{1-\delta}} \frac{1}{\gamma} \left(E_{\gamma}({\cK}(P_X)\|{\cK}(Q_X)) + \frac{1}{2} \norm{\Tilde{P}_Y - {\cK}(P_X)}{1} + \frac{\gamma}{2}\norm{\Tilde{Q}_Y - {\cK}(Q_X)}{1}\right) d\gamma\nn\\
    &\quad + \int_{1}^{\frac{e^\eps}{1-\delta}} \frac{1}{\gamma^2}\left(E_{\gamma}({\cK}(Q_X)\|{\cK}(P_X)) + \frac{1}{2} \norm{\Tilde{Q}_Y - {\cK}(Q_X)}{1} + \frac{\gamma}{2}\norm{\Tilde{P}_Y - {\cK}(P_X)}{1}\right) d\gamma\nn\\
    &\overset{(e)}{\leq} \int_{1}^{\frac{e^\eps}{1-\delta}} \frac{1}{\gamma} \left(E_{\gamma}({\cK}(P_X)\|{\cK}(Q_X)) + \delta(\gamma + 1)\right) d\gamma + \int_{1}^{\frac{e^\eps}{1-\delta}} \frac{1}{\gamma^2}\left(E_{\gamma}({\cK}(Q_X)\|{\cK}(P_X)) + \delta(\gamma + 1)\right) d\gamma\nn\\
    &= \int_{1}^{\frac{e^\eps}{1-\delta}} \frac{1}{\gamma}E_{\gamma}({\cK}(P_X)\|{\cK}(Q_X)) d\gamma + \int_{1}^{\frac{e^\eps}{1-\delta}} \frac{1}{\gamma^2}E_{\gamma}({\cK}(Q_X)\|{\cK}(P_X))  d\gamma + \delta \int_{1}^{\frac{e^\eps}{1-\delta}} \left(1 + \frac{2}{\gamma} +\frac{1}{\gamma^2}\right)d\gamma\nn\\
    &= \int_{1}^{\frac{e^\eps}{1-\delta}} \left(\frac{1}{\gamma}E_{\gamma}({\cK}(P_X)\|{\cK}(Q_X)) + \frac{1}{\gamma^2}E_{\gamma}({\cK}(Q_X)\|{\cK}(P_X))\right)  d\gamma + \delta \left(\frac{e^\eps}{1-\delta} + 2\log\frac{e^\eps}{1-\delta}  - \frac{1-\delta}{e^\eps} \right)\nn\\
    &= \int_{1}^{e^\eps} \left(\frac{1}{\gamma}E_{\gamma}({\cK}(P_X)\|{\cK}(Q_X)) + \frac{1}{\gamma^2}E_{\gamma}({\cK}(Q_X)\|{\cK}(P_X))\right)  d\gamma \nn\\
    &\quad+ \int_{e^\eps}^{\frac{e^\eps}{1-\delta}} \left(\frac{1}{\gamma}E_{\gamma}({\cK}(P_X)\|{\cK}(Q_X)) + \frac{1}{\gamma^2}E_{\gamma}({\cK}(Q_X)\|{\cK}(P_X))\right)  d\gamma\nn\\
    &\quad + \delta \left(\frac{e^\eps}{1-\delta} + 2\log\frac{e^\eps}{1-\delta}  - \frac{1-\delta}{e^\eps} \right)\nn\\
    &\overset{(f)}{\leq} \frac{1}{2}\norm{P_X -Q_X}{1} \left(\int_{1}^{e^\eps}  \frac{e^\eps + \delta + \gamma(\delta - 1)}{e^\eps + 1} \left(\frac{1}{\gamma} + \frac{1}{\gamma^2}\right) d\gamma + \int_{e^\eps}^{\frac{e^\eps}{1-\delta}}  \delta\left(\frac{1}{\gamma} + \frac{1}{\gamma^2}\right) d\gamma\right)\nn\\
    &\quad + \delta \left(\frac{e^\eps}{1-\delta} + 2\log\frac{e^\eps}{1-\delta}  - \frac{1-\delta}{e^\eps} \right)\nn\\
    &= \frac{1}{2}\norm{P_X -Q_X}{1} \left(\eps \tanh\left(\frac{\eps}{2}\right)+ \delta\left(\frac{2\eps}{e^\eps+1} + \frac{e^\eps-1}{e^\eps}\right) + \delta\left(\log\frac{1}{1-\delta} + \frac{\delta}{e^\eps}\right)\right)\nn\\
    &\quad + \delta \left(\frac{e^\eps}{1-\delta} + 2\log\frac{e^\eps}{1-\delta}  - \frac{1-\delta}{e^\eps} \right)\nn\\
    &= \frac{1}{2}\norm{P_X -Q_X}{1} \left(\eps \tanh\left(\frac{\eps}{2}\right)+ \delta\left(\frac{2\eps}{e^\eps+1} + \frac{e^\eps-1}{e^\eps} + \log\frac{1}{1-\delta} + \frac{\delta}{e^\eps}\right)\right)\nn\\
    &~~~+ \delta \left(\frac{e^\eps}{1-\delta} + 2\log\frac{e^\eps}{1-\delta}  - \frac{1-\delta}{e^\eps} \right),\label{lemma_DP_KL_classical_eq2}
\end{align}
where $(a)$ follows from the integral representation of relative entropy (see \cite[Eq. $428$]{Sason_2016}), $(b)$ follows from the fact that for all $\gamma \geq e^{D_{\max}(P\|Q)}$ $P \leq e^\gamma Q$ and thus $E_{\gamma}(P\|Q) = 0$, $(c)$ follows since from $(ii)$ of Lemma \ref{lemma_trunc_prop} we have $D_{\max}(\Tilde{P}_Y\|\Tilde{Q}_Y) \leq \eps + \log(1/(1-\delta))$ and $D_{\max}(\Tilde{Q}_Y\|\Tilde{P}_Y) \leq \eps + \log(1/(1-\delta))$, $(d)$ follows from Fact \ref{fact_hockey_approx}, $(e)$ follows from the definition of truncated distributions (see Definition \ref{projpair}) and $(f)$ follows from Corollary \ref{corr_contrac_classical_hockey_stick} since the channel $\cK$ is $(\eps,\delta)$-DP. 

Therefore, combining \eqref{classical_kl_p_eq1} and \eqref{lemma_DP_KL_classical_eq2}, we have the following,
\begin{align}    
&~ D({\cK}(P_X)\|{\cK}(Q_X)) \nn\\
    &\leq \frac{1}{2}\norm{P_X -Q_X}{1} \left(\eps \tanh\left(\frac{\eps}{2}\right)+ \delta\left(\frac{2\eps}{e^\eps+1} + \frac{e^\eps-1}{e^\eps} + \log\frac{1}{1-\delta} + \frac{\delta}{e^\eps}\right)\right)\nn\\
    &\quad + \delta \left(\frac{e^\eps}{1-\delta} + 2\log\frac{e^\eps}{1-\delta}  - \frac{1-\delta}{e^\eps} + 2 \left( \eps + \log \frac{1}{\delta} +  \frac{2}{m} \right)\right).\label{lemma_DP_KL_classical_eq3}
\end{align}

This completes the proof of Lemma \ref{theorem_classical_kl_bound}.\hfill\QED

\end{document}